\documentclass[11pt,twoside,a4paper]{article}

\usepackage{fullpage}
\usepackage{enumerate}
\usepackage{subfigure}
\usepackage{enumitem}
\usepackage{longtable,geometry}
\usepackage[utf8]{inputenc}
\usepackage[T1]{fontenc}
\usepackage{color}
\usepackage{cancel}   
\usepackage[normalem]{ulem}

\usepackage[english]{babel}

\usepackage{enumitem}
\usepackage{amsmath,amsfonts,amssymb,graphicx,bbm,stmaryrd}

\usepackage{amsthm}

\theoremstyle{plain}

\newtheorem{theorem}{Theorem}[section]
\newtheorem{proposition}[theorem]{Proposition}
\newtheorem{lemma}[theorem]{Lemma}

\newtheorem{corollary}[theorem]{Corollary}
\newtheorem{conjecture}[theorem]{Conjecture}
\newtheorem{question}[theorem]{Question}

\theoremstyle{remark}
\newtheorem{remark}{Remark}
\theoremstyle{definition}
\newtheorem{definition}[theorem]{Definition}

\numberwithin{equation}{section}

\def\Pf{{\bf Pf}}
\def\be{\begin{equation}}
\def\ee{\end{equation}}
\usepackage{mathtools}

\usepackage{dsfont}

\newcommand{\lr}[2][\quad]{\xleftrightarrow{\:\:}#2 \text{ in $#1$}}

\newcommand{\calC}{\mathcal{C}}

\newcommand{\calE}{\mathcal{E}}

\newcommand{\calK}{\mathcal{K}}
\newcommand{\calL}{\mathcal{L}}
\newcommand{\calM}{\mathcal{M}}
\newcommand{\calN}{\mathcal{N}}

\newcommand{\calP}{\mathcal{P}}

\newcommand{\calR}{\mathcal{R}}
\newcommand{\calS}{\mathcal{S}}

\newcommand{\bbH}{\mathbb{H}}
\newcommand{\bbI}{\mathbb{I}}

\newcommand{\bbP}{\mathbb{P}}

\newcommand{\bbR}{\mathbb{R}}
\newcommand{\bbS}{\mathbb{S}}

\newcommand{\bbZ}{\mathbb{Z}}

\newcommand{\ep}{\varepsilon}

\newcommand{\Z}{\mathbb{Z}}

\newcommand{\g}{\gamma}

\newcommand{\G}{G}
\newcommand{\E}{E}
\newcommand{\V}{V}

\newcommand{\N}{\mathbb{N}}
\newcommand{\R}{\mathbb{R}}

\DeclareMathOperator{\sgn}{sgn}

\newcommand{\n}{{\bf n}}
\newcommand{\m}{{\bf m}}

\renewcommand{\cal}{\mathcal}
\renewcommand{\g}{{\mathbf{\Gamma}}}

\usepackage{constants}
\newconstantfamily{small}{symbol=c}

\usepackage{verbatim}

\usepackage{tabularx}

\usepackage{hyperref}

 \title{Emergent Planarity in two-dimensional Ising Models with finite-range Interactions} 
 
 \author{Michael Aizenman \and Hugo Duminil-Copin \and Vincent Tassion \and  Simone Warzel}
\date{\today}

\begin{document}

\maketitle

\begin{abstract}  
The known Pfaffian structure of the boundary spin correlations, and more generally  order-disorder correlation functions, is given a new  explanation through simple topological considerations within the model's random current representation.  This perspective is then employed in the proof that the Pfaffian structure of boundary correlations emerges asymptotically at criticality in Ising models on  $\mathbb Z^2$ with  finite-range interactions.   The analysis is enabled by new results on the stochastic geometry of the corresponding random currents.    The proven statement establishes an aspect of universality, seen here in the emergence  of  fermionic structures in two dimensions  beyond the solvable cases.
\end{abstract}

\section{Introduction}

\subsection{Outline} 

Among the fascinating discoveries about the two-dimensional Ising model is the existence of non-interacting fermionic excitations, which play a fundamental role in the model's solvability.   A notable manifestation of this phenomenon is the fact that in the planar case the partition function can be expressed as a determinant, and that some of the model's correlation functions are Pfaffians of the corresponding two-point functions.   In particular, this applies to the multi-spin correlation functions of spins located along the boundaries of Ising models with arbitrary nearest-neighbor pair interactions on planar graphs.

The Ising model
is perhaps the most studied example of a system undergoing a phase transition.  It has provided  the testing ground for a large variety of techniques. 
These include exact methods, a number of which have been developed following Onsager's remarkable announcement~\cite{Ons44} of the exact solution of the model on $\Z^2$.

The presence of fermionic structures was noted soon after  the appearance of  Onsager's solution.   
The point was stressed early on in the algebraic analysis of Kaufman~\cite{Kau49},  and then in the reformulation by  Schultz-Mattis-Lieb~\cite{SchMatLie64} of the transfer matrix in terms of free fermions.  Hints of fermionic structure can also be seen in the reductions of the planar Ising model's partition function to Pfaffians, by Hurst-Green \cite{HurGre60}, Kasteleyn~\cite{Kas63} and Fisher~\cite{Fis66}.  
An insightful geometric explanation was added by Kadanoff and Ceva~\cite{Kad66a,KadCev71} who called attention to the fermionic spinor structure seen in the model's order-disorder variables.  The understanding of the theory progressed rapidly through the thorough analysis which was enabled by these exact methods~(cf.~\cite{Bax78,McCWu73} and references therein).   
More recent  results on the subject, which has continued to draw attention,  are mentioned below where directly relevant.

Groeneveld-Boel-Kasteleyn~\cite{GBK78} first pointed out  that for planar  Ising systems, the correlation functions for spins located along the boundaries are given at any temperature by Pfaffians of  the corresponding   two point function, i.e.~satisfy a \emph{fermionic} version of the Wick rule.  
Our discussion starts with the explanation of this property in terms of simple  arguments which are enabled by the model's random current representation~\cite{Aiz82,GriHurShe70}.  Curiously, our derivation of the fermionic rules for planar models  employs the same identities as were previously used to establish that in high dimensions, the scaling limits of  critical correlations of ferromagnetic models obey the \emph{bosonic} Wick rule.

This analysis also shows that the proven statement  does not hold in the strict sense beyond the planar case.   In particular,   it does not extend to non-planar interactions, e.g.~two-dimensional models with pair interactions reaching beyond the range of nearest-neighbors.   
Nevertheless, for such extensions, we then prove that  Pfaffian relations emerge at the models' critical points.   
An  intuitive  explanation of this fact is suggested  by the expected picture of fractality in the stochastic geometry underlying the critical models' correlation functions (in the spirit of the extension of the RSW theory to the Ising model~\cite{CheDumHon13,DumHonNol11},  and the fractality criteria of~\cite{AizBur99}).
In the proof, use is made of the coupling between three different graphical representations of the Ising model (the high-temperature expansion, the random current representation, and the Fortuin-Kasteleyn percolation), which enables one to combine their different convenient features.

The asymptotic  emergence  of  the Pfaffian structure in two dimensions, proven here beyond the solvable nearest-neighbor case, is consistent with the expected picture of universality in critical phenomena.   Universality results were previously obtained for small enough perturbations of the nearest-neighbor model on the square lattice in \cite{GiuGreMas12,PinSpe}. 
In comparison, we do not  discuss here the models' critical exponents, limiting the  analysis to the correlations' multi-particle structure.  However, 
the results are proven for all finite-range models with no restriction on the relative strength of the interaction terms.

\subsection{The Ising model -- definition and notation}\label{sec:def}

In the following discussion, the symbol $\G = (\V(\G), \E(\G)) $ denotes  a finite graph with vertex-set $\V(\G)$  and edge-set  $\E(\G) $.
We focus  on finite graphs as the theory's extension  to infinite graphs does not require more than known arguments (cf.~\cite{AizDumSid15,Geo11}). A {\em planar} graph is a graph embedded in the plane $\R^2$  in such a way that its edges, depicted by bounded simple arcs, intersect  only at their endpoints.  The {\em faces} of the graph are the connected components of the plane minus the edges.

A configuration of Ising spin variables on a graph $\G$ is a binary valued function $\sigma : \V(\G) \mapsto \{-1,1\}$. Following standard notation, we write $\sigma_x=\sigma(x)$ for the {\em spin} at $x$. The Hamiltonian of the model is the function, defined on the configurations $\sigma\in\{\pm1\}^{V(G)}$, given by \be {\bf H}_\G(\sigma) := - \sum_{\{x,y\} \in \E(\G)} J_{x,y} \sigma_x \sigma_y - h \sum_{x \in \V(\G)} \sigma_x. \ee Its parameters $J:=\{J_{x,y}\}_{\{x,y\}\in E(\G)}$ are referred to as the {\em coupling} coefficients, and $h$ is the {\em magnetic field parameter}.   
Throughout the discussion, we shall refer to the set of edges with  non-zero coupling constants as the edge-set of the graph associated with the model (or with the couplings $J$).
Our discussion will focus on the case $h=0$, which is where the model's phase transition lies.   

The corresponding Gibbs equilibrium state at inverse temperature $\beta \ge0$ is the probability measure  such that
$$ \langle f\rangle_{\G,\beta} :=  \frac{\sum_{\sigma\in\{\pm1\}^{\V(\G)}}f(\sigma)\exp(-\beta {\bf H}_G(\sigma)) }{ 2^{|V(\G)|} \, \, Z(\G,\beta)}   $$
for any $f:\{\pm1\}^{\V(\G)}\longrightarrow \mathbb C$,  where the normalizing factor
$$Z(\G,\beta) \  := \ 2^{-|V(\G)|} \sum_{\sigma \in \{\pm1\}^{\V(\G)}} \exp(-\beta {\bf H}_G(\sigma))$$ is referred to as  the {\em partition function}.

\begin{figure}[h]
\begin{center}
\includegraphics [width=1.00 \textwidth]{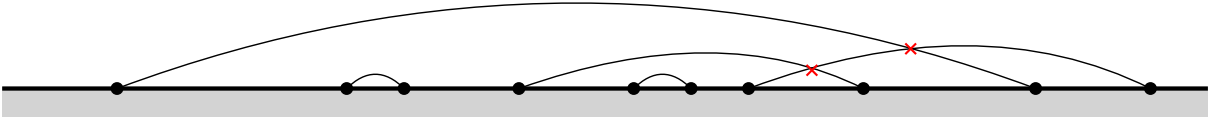}
\caption{An example of a pairing $\pi$ (with $\sgn(\pi)=1$) of points lying along the boundary of a planar region.  The pairing's signature coincides with the parity of the number $K$ of mutual crossings, among the different lines linking the paired sites within the planar domain; i.e.~$\sgn (\pi)$ is equal to  $(-1)^K$.   
}  \label{fig:boundary}
\end{center}
\end{figure}
A guiding example will be the graph whose vertex-set is $ \Z^2$ and the edge-set  consists of  pairs of vertices at Euclidean distance one.  
A \emph{finite-range} extension of the model is  defined by any extensions of $J_{x,y}$ to a larger edge-set, yet one for which $J_{x,y}=0$ for pairs with 
$\|x-y\|>R$.  Such extensions are generically no longer planar, though in many respects still two-dimensional.

 For the planar case, the model's partition function and correlation functions have been solved exactly, but no exact solutions  were found beyond planarity.  Nevertheless, it is generally expected that the model's critical behavior, in particular critical exponents, will be the same among the set of coupling constants  which are ferromagnetic ($J_{x,y} \geq 0$), translation invariant, and transitive (i.e.~with respect to which the graph does not decompose into two uncoupled components).   Partial universality results of this nature have been obtained for weak enough extensions of the nearest-neighbor model on $\Z^2$~\cite{GiuGreMas12}.

\subsection{Pfaffian structure of the boundary spin correlation in 2D Ising models}

Onsager's exact solution of the model on $\Z^2$ \cite{Ons44} utilizes non-commutative algebraic structures whose nature has been successively explored and further clarified in the rich collection of works that followed~\cite{CheSmi12,KadCev71,Kau49,SchMatLie64,McCWu73}.    In particular, Schultz-Mattis-Lieb \cite{SchMatLie64} represented the model's transfer matrix in terms of non-interacting fermions,   and  Groeneveld-Boel-Kasteleyn~\cite{GBK78} noted and proved that in any planar version of the model, the $n$-point correlation functions of spins located along a common boundary line are given by a Pfaffian involving the $2$-point function only.  

For a more explicit statement, let us recall that the Pfaffian of an upper triangular array $A=[A_{i,j}]_{1\le i<j\le 2n}$ is defined  as
\be  \label{def:Pf}
\Pf_n(A)  \ := \
\sum_{\pi \in \Pi_{n}}  \sgn(\pi)  A_{\pi(1),\pi(2)}
\cdots  A_{\pi(2n-1),\pi(2n)} \, ,
\ee
where $\Pi_{n}$ is the collection of pairings of $\{1,\dots,2n\}$ and $\sgn(\pi) $ the pairing's signature.
A 
{\em pairing} of $\{1,\dots,2n\}$ is a permutation $\pi $ such that
$\pi(2j-1)<\pi(2j)$ for every $j\in\{1,\dots,n\}$ and $\pi(2j-1)<\pi(2j+1)$ for every $j\in\{1,\dots,n-1\}$. Pfaffians appeared very early in the exact solution of the free energy of two-dimensional Ising models \cite{Fis66,HurGre60,Kas63} through the   relation $\Pf_n(A)^2  = \text{det}( A)$,
where $ A $ is the antisymmetric matrix whose entries   for $j<k$ are given by the entries $  A_{j,k} $.

The above quoted result of Groeneveld-Boel-Kasteleyn~\cite{GBK78}  reads as follows. 

\begin{theorem}[Pfaffian structure for boundary spin correlation functions]\label{thm:pf_boundary}
Fix a planar graph $\G$, arbitrary nearest-neighbor couplings $J$, and $\beta\ge0$.  
Then, for any cyclically ordered $2n$-tuple $(x_1,\dots,x_{2n})$ of sites located along the boundary of a fixed face of $\G$, we have 
\be \label{Pf_boundary}
 \langle \sigma_{x_1}\cdots\sigma_{x_{2n}}\rangle_{\G,\beta} \ = \
\Pf_n \big( \big[\langle \sigma_{x_i}\sigma_{x_j}\rangle_{\G,\beta}\big]_{1\le i<j\le 2n} \big)
\, .
\ee
\end{theorem}

The  relation \eqref{Pf_boundary} can be viewed as  the fermionic counterpart of the Wick rule of the  (bosonic)  correlations of Gaussian fields.   The fermionic denomination comes from the fact that Pfaffians are a familiar feature of  the vacuum expectation values of products of Majorana spinors, and of the thermal equilibrium states  of systems of {\em non-interacting} fermions.

The above statement is usually made in reference to spins located along  the outer boundary of the graph, listed as  $x_1,\dots,x_{2n}$ in the boundary's cyclic order.   
However,  all faces of the graph are equivalent in this respect,  as can be seen
 by properly applying an inversion of the plane onto itself which induces a graph isomorphism, as in the example depicted in Fig.~\ref{fig:boundary2}.

\begin{figure}[h]
\begin{center}
\includegraphics [width = 0.450 \textwidth]{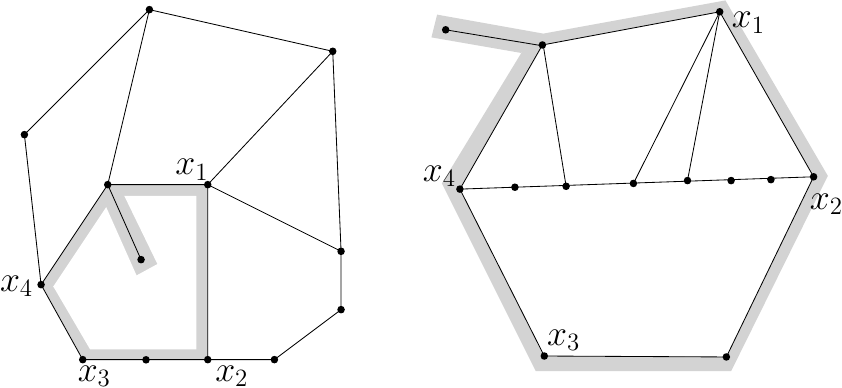}
\caption{A pair of isomorphic planar graphs  related through an inversion mapping under which an inner face (on the left), whose boundary is marked in gray, is turned into an outer face (on the right). }  \label{fig:boundary2}
\end{center}
\end{figure}

The main result presented in this article is the following extension of Theorem~\ref{thm:pf_boundary} to finite-range models  on the half space  $\bbH:=\bbZ\times\Z_+$. 
In the non-planar case the Pfaffian relations no longer hold in the strict sense (a fact which was noted by \cite{GBK78}, and is also easily seen from the argument given below). 
However, we prove that  the Pfaffian  structure of correlations re-emerges  at the critical points,  as an asymptotic relation at large spin separations.

\begin{theorem}\label{thm:Pf_finite_range}
Let $J$ be a set of coupling constants for an Ising model over the  upper half-plane $\bbH$ which are:  
\begin{itemize}[noitemsep]
\item[(i)] ferromagnetic, i.e.~that $J_{x,y}\ge0$ for every $ x,y\in\mathbb \bbH$, 

\item[(ii)]  finite-range  and such that the associated graph is connected, 

\item[(iii)]  translation invariant, i.e.~that $J_{x,y}=J(x-y)$, and 
 
\item[(iv)] reflection invariant: $J_{0,(a,b)}=J_{0,(-a,b)}=J_{0,(a,-b)}=J_{0,(b,a)}$ for every $a,b\in \bbZ$.   \end{itemize}
  Then, for any collection of boundary points  $x_1=(k_1,0),\dots,x_{2n}=(k_{2n},0)$  satisfying $k_1<k_2<\dots<k_{2n}$, we have
 \be \label{Pf_cor}
 \langle \sigma_{x_1}\cdots\sigma_{x_{2n}}\rangle_{\mathbb H,\beta_c} \ = \
\Pf_n \big( \big[\langle \sigma_{x_i}\sigma_{x_j}\rangle_{\mathbb H,\beta_c}\big]_{1\le i<j\le 2n} \big)\big[1 +o(1)\big]
\, ,
\ee
where $\beta_c$ is the critical inverse-temperature of the model, and $o(1)$ is a function of the points $x_1,\dots,x_{2n}$   which tends to zero for configuration sequences with  $\min_{1\le i<j\le 2n} \{|x_i-x_j|\} \to \infty$. 
\end{theorem}


Our discussion focuses on $\beta=\beta_c$, but let it be pointed out that the Pfaffian relation holds asymptotically also off criticality.  However, there it is valid for rather trivial reasons: for  $\beta < \beta_c$ a single leading term dominates, and for $\beta >\beta_c$ all terms tend to a non-zero constant.  Hence, Theorem~\ref{thm:Pf_finite_range} concerns the case of main interest.

\subsection{Stochastic geometry meets two-dimensional  topology} 

We prove Theorem~\ref{thm:Pf_finite_range} using a graph theoretical representation of Ising model's correlation functions which is known as the random current representation and is valid for arbitrary graphs.  
For an intuitive outline of the underpinning of this result, it is helpful to first present an 
elementary proof of  Theorem~\ref{thm:pf_boundary} in which  this 
stochastic geometric representation is combined with an elementary topological argument applying only in the planar situation.

While it is unorthodox to include a proof in the introduction, we believe that it will be useful for the discussion of other results. Readers unfamiliar with random currents may return to the proof after reading Section~\ref{sec:RC}.

For a quick grasp of the essence of the argument, let us start with a known expression for the Ursell $4$-point function of Ising spin systems.    Expressing the correlations in terms of the random current representation (see Section~\ref{sec:RC}), and applying the classical switching principle~(Lemma~\ref{lem:switch}), one obtains 
the following formula,  which is valid for any graph and any choice of labeling of a given set of four points,
\begin{align} 
U_4(x_1,\dots, x_4)  &:= \langle \sigma_{x_1}  \sigma_{x_2}  \sigma_{x_3} \sigma_{x_{4}}\rangle -  
\left[    \langle \sigma_{x_1}  \sigma_{x_2} \rangle \langle \sigma_{x_3} \sigma_{x_{4}}\rangle 
+    \langle \sigma_{x_1}  \sigma_{x_3} \rangle \langle \sigma_{x_2} \sigma_{x_{4}}\rangle 
+ \langle \sigma_{x_1}  \sigma_{x_4} \rangle \langle \sigma_{x_2} \sigma_{x_{3}}\rangle
\right]   \nonumber 
\\  
&\phantom{:} = -  2\,  \langle \sigma_{x_1}  \sigma_{x_3}   \sigma_{x_2} \sigma_{x_4}\rangle  \
 \mathbf{E}^{\{x_1,x_2,x_3,x_4\},\emptyset }_{G,\beta} 
(\mathbb{I}[x_1 \stackrel{\n_1+\n_2}{   \longleftrightarrow} x_j, j=2,3,4])      
\nonumber 
\\
& \phantom{:} = -  2\,  \langle \sigma_{x_1}  \sigma_{x_3} \rangle \langle \sigma_{x_2} \sigma_{x_4}\rangle  \
 \mathbf{E}^{\{x_1,x_3\},\{x_2,x_4\} }_{G,\beta} 
 (\mathbb{I}[x_1 \stackrel{\n_1+\n_2}{   \longleftrightarrow} x_2])\,  , \    \label{U4} 
\end{align} 
where we omitted the common subscript $(\G,\beta)$ on the correlation functions.  In these relations:
\begin{enumerate}[noitemsep,nolistsep]
\item[{\it i.}] The expectation value $\mathbf E^{A,B}_{G,\beta} $ is taken over a duplicated system of configurations of edges of $\E(\G)$, counted with multiplicities $\n_1(x,y)$ and $\n_2(x,y)$ for each $\{x,y\}\in \E(\G)$, with the specified  \emph{source sets} $\partial \n_1 =A$ and $\partial \n_2 =B$ which are
defined through:
\be
\partial \n :=  \{x\in \V(\G) \, :\,  (-1)^{\sum_{\{x,y\}\in E(G)} \n(x,y)} =-1 \} \,.
\ee
\item[{\it ii.}] We denote by $ \mathbb{I} [x_i \stackrel{\n_1+\n_2}{ \longleftrightarrow} x_j]$ the indicator functions of the condition that 
 $x_i $ and $x_j$ are connected by a path of edges $\{x,y\}\in E(G)$ along which $\n_1(x,y)+\n_2(x,y)> 0$.   The opposite condition will be indicated by  
$$ \mathbb{I} [x_i \stackrel{\n_1+\n_2}{ \kern+4.5pt\arrownot\kern-4.5pt\longleftrightarrow} x_j]:=1 -  \mathbb{I} [x_i \stackrel{\n_1+\n_2}{ \longleftrightarrow} x_j].$$ 
\end{enumerate}
The explicit law of $(\n_1,\n_2)$ does not matter for the topological argument which follows.
What matters is  that the multigraph obtained by connecting each pair of neighboring sites by $\n_1(x,y)+\n_2(x,y) $ edges can be represented as a union of paths pairing the sources and a collection of loops.

The relation \eqref{U4} plays a key role in the proof  of asymptotic Gaussianness of the correlation function of the critical models in dimensions~$d>4$ \cite{Aiz82}, in which the intersection probability which appears there vanishes asymptotically.  However, for the present purpose let us note that \eqref{U4} may also be rewritten in the following form, which allows a simple explanation of the Pfaffian structure in cases the opposite is the case:    
\begin{align}
 \langle \sigma_{x_1}  \sigma_{x_2}  \sigma_{x_3} \sigma_{x_{4}}\rangle &-  
\Pf_2 \big( [\langle \sigma_{x_i}\sigma_{x_j}\rangle ]_{1\le i<j\le 4} \big)\nonumber
 \\
&=    \langle \sigma_{x_1}  \sigma_{x_2}  \sigma_{x_3} \sigma_{x_{4}}\rangle   -  
\left[    \langle \sigma_{x_1}  \sigma_{x_2} \rangle \langle \sigma_{x_3} \sigma_{x_{4}}\rangle 
-    \langle \sigma_{x_1}  \sigma_{x_3} \rangle \langle \sigma_{x_2} \sigma_{x_{4}}\rangle 
+ \langle \sigma_{x_1}  \sigma_{x_4} \rangle \langle \sigma_{x_2} \sigma_{x_{3}}\rangle
\right]   \nonumber
\\
&=  2\,  \langle \sigma_{x_1}  \sigma_{x_3} \rangle \langle \sigma_{x_2} \sigma_{x_4}\rangle  \
 \mathbf{E}^{ \{x_1,x_3\},\{x_2,x_4\} }_{G,\beta} 
( 
 \mathbb{I} [x_1 \stackrel{\n_1+\n_2}{ \kern+4.5pt\arrownot\kern-4.5pt\longleftrightarrow} x_2 ])\, .   \label{P4_avoided}
\end{align} 
It readily follows that for $n=2$, the Pfaffian relation \eqref{Pf_cor}  holds  if and only if the two sites $x_2$ and $x_3$ are connected in any contributing realization of the random currents  with $\partial \n_1 = \{x_1,x_3\}$ and  
$\partial \n_2 = \{x_2,x_4\}$.   For elementary topological reasons, this condition is  satisfied in the nearest-neighbor model for cyclically labeled boundary sites since the path pairing $x_1$ to $x_3$ must intersect the one pairing $x_2$ to $x_4$.

The above explanation of \eqref{Pf_cor} also makes it clear that this Pfaffian relation does not hold 
in higher dimensions, and that  it also fails in 2D systems with finite-range interactions for which the couplings $J$ give non-zero weight to  paths which leap over each other without making contact.  We refer to such events as \emph{avoided crossings}.

In the following  proof of Theorem~\ref{thm:pf_boundary}, which concerns the strictly planar case, the above argument is extended to all $n$.  
\begin{proof}[Proof of Theorem~\ref{thm:pf_boundary}]
Define the quantity
\begin{equation}
\label{eq:p}Q (x_1,\dots, x_{2n}):= \langle \sigma_{x_1}\cdots\sigma_{x_{2n}}\rangle_{\G,\beta} - \sum_{\ell=2}^{2n} (-1)^{\ell}  \langle \sigma_{x_1} \sigma_{x_{\ell}}\rangle_{\G,\beta}   \, \,
\big\langle \prod_{\substack{1\le j\le 2n \\ j \notin  \{1,\ell\} }}\sigma_{x_j} \big\rangle_{\G,\beta} .\end{equation}

By arguments similar to those leading to \eqref{P4_avoided} (again, they involve the random current representation and the switching lemma described in the next section),   
one finds 
\be \label{eq:Q1}
Q(x_1,\dots,x_{2n})  =   \langle \sigma_{x_1}\cdots\sigma_{x_{2n}}\rangle_{\G,\beta} \
 \mathbf{E}^{\{x_1,\dots, x_{2n}\},\emptyset }_{G,\beta} \Big(\  \sum_{\ell=1}^{2n} (-1)^{\ell+1}  \,
 \mathbb{I} [x_1 \stackrel{\n_1+\n_2}{ \longleftrightarrow} x_\ell ]\Big)\, , 
\ee
where for convenience all sources were moved to $\n_1$.   

This above relation is valid for any graph, and any labeling of the $2n$-tuple of sites.  
However, in  the case of boundary spins of a planar graph  $\G$, listed in a cyclic order,  for any given configuration $ (\n_1,\n_2)$ with the prescribed source constraints,  the sites $x_\ell$  which are connected to $x_1$ by $\n_1+\n_2$ are of alternating parity    (as depicted in Figure~\ref{fig:boundary}).  
It follows that for \emph{every} contributing configuration $ (\n_1,\n_2)$,
the non-zero terms in the sum form an alternating sequence of $\pm 1$, which starts with $+1$ and ends with $-1$, and hence 
\be \label{eq:main_boundary}
\sum_{\ell=1}^{2n} (-1)^{\ell+1}  \,
 \mathbb{I} [x_1 \stackrel{\n_1+\n_2}{ \longleftrightarrow} x_\ell  ]  =  0 \, . 
\ee
Substituting \eqref{eq:main_boundary} in \eqref{eq:Q1}, one learns that for any cyclically ordered sequence of boundary sites,  
$ 
 Q(x_1,\dots,x_{2n}) = 0, 
$
 or more explicitly:
\begin{align} 
 \langle \sigma_{x_1}\cdots\sigma_{x_{2n}}\rangle_{\G,\beta} &= 
  \sum_{\ell=2}^{2n} (-1)^{\ell}  \langle \sigma_{x_1} \sigma_{x_{\ell}}\rangle_{\G,\beta}  \,
\big\langle \prod_{\substack{1\le j\le 2n \\ j \notin  \{1,\ell\} }}\sigma_{x_j} \big\rangle_{\G,\beta} \, .\end{align}
For  $n=2$, the above is exactly the Pfaffian relation~\eqref{Pf_boundary} for the  four-point function.  
By induction, this relation extends to all $n\geq 2$ since the $n$-level  Pfaffian can be determined iteratively through the general relation: 
\be \label{eq:PFcrit}
\Pf_n( A) = \  \sum_{\ell=2}^{2n} (-1)^{\ell}  A_{1,\ell}\, \Pf_{n-1}([A]_{1,\ell})\,.
\ee
\end{proof} 
 
\noindent {\bf Remarks:} 
\begin{enumerate} 
\item A compelling picture emerges from the comparison of  \eqref{P4_avoided} with \eqref{U4}:  
\begin{itemize} 
\item In situations in which \emph{path intersections} play only a negligible role, the correlation function exhibits a  Bosonic (Wick law) structure, which is characteristic of Gaussian processes.  
\item In situations in which \emph{avoided crossings} are ruled out, the correlation functions exhibit a Fermionic (Pfaffian) structure.  
\end{itemize} 
While the Gaussian rule applies to critical models in $d\geq 4$ dimensions, the Pfaffian structure applies to the boundary spins of planar models.  

\item In Section~\ref{sec:OD}, the argument given above is extended to  a proof of the more general statement that in planar models,  the Pfaffian structure is also found in the correlation functions of suitably defined \emph{order-disorder} variables.    

\item During the preparation of the manuscript, we were made aware of the work of Lis~\cite{Lis16} who has  obtained an alternative  random current derivation of 
the Groeneveld-Boel-Kasteleyn formula (Theorem~\ref{thm:pf_boundary}). 
\end{enumerate}

The proof given above will be used here as a springboard for an extension of the Pfaffian structure beyond strict planarity for  $\beta = \beta_c$, which is the case of main interest in this paper and for which it will be shown to hold in an asymptotic sense.    

\subsection{Proof idea of Theorem~\ref{thm:Pf_finite_range}}

We may now outline the proof idea of the new result, Theorem~\ref{thm:Pf_finite_range}, using the terminology which was presented in the above discussion of Theorem~\ref{thm:pf_boundary}.  

For an Ising model with a finite-range interaction in the half-plane  $\bbH:=\bbZ\times\bbZ_+$, let 
$\{x_1,\dots,x_{2n}\} \in \partial \bbH$ be a collection of boundary sites, labeled in the corresponding order.  
Any configuration $\n_1: \E(\bbH) \to \N$ with  $\partial \n_1 = \{x_1,\dots,x_{2n}\} \in \partial \bbH$ can be presented as giving the flux numbers of a collection of (possibly overlapping) loops and paths linking pairwise the specified sources.   If all points are at distances greater than the interaction range $R$, then the pairing paths will be crossing each other, as they do in the strictly planar case.  The difference is that with $R>1$, the paths may cross without intersecting and without being connected by other means (e.g.~the  mentioned loops).   Thus, while the representation \eqref{eq:Q1} still applies, the topological argument which led to \eqref{eq:main_boundary} may no longer be counted upon.  

To quantify the above observation, let us define as \emph{occurrence of an avoided intersection}
 the event, for $i<j<k<\ell$, given by 
%
\be 
A_{i,j,k,\ell}(\n_1+\n_2) \, :=\, \{x_i \stackrel{\n_1+\n_2}{ \longleftrightarrow} x_k\}\cap\{x_j \stackrel{\n_1+\n_2}{ \longleftrightarrow} x_\ell\}\cap\{x_i \stackrel{\n_1+\n_2}{\kern+4.5pt\arrownot\kern-4.5pt\longleftrightarrow} x_j\} \,. 
\ee 
 
For configurations in which avoided intersections do not occur, the sum in \eqref{eq:main_boundary} still vanishes (and regardless of this condition, the sum is dominated by $n$).  It  follows that
 \be\label{eq:p1}
|Q(x_1,\dots,x_{2n})|\ \leq \   n\ \langle\sigma_{x_1}\cdots\sigma_{x_{2n}}\rangle_{G,\beta}\sum_{i<j<k<\ell}\ {\bf P}^{ \{x_1,\dots, x_{2n}\} , \emptyset }_{G,\beta} (A_{i,j,k,\ell})
\ee
with ${\bf P}^{  \{x_1,\dots, x_{2n}\} , \emptyset }_{G,\beta} $ being the probability measure corresponding to the expectation value in \eqref{eq:Q1}.  

Consequently, the Pfaffian rule is {\em approximately} true in situations where the {\em probability} of avoided intersections is small 
for every quadruples in $\{x_1,\dots,x_{2n}\}$ as, under this assumption, the expression \eqref{eq:p} and the bound \eqref{eq:p1} yield 
 \be\label{eq:pf_approx}
 \langle\sigma_{x_1}\cdots\sigma_{x_{2n}}\rangle_{G,\beta} \,  \big[ 1+ o(1) \big] \ = \   \Pf_n \big( \big[\langle \sigma_{x_i}\sigma_{x_j}\rangle_{\mathbb H,\beta_c}\big]_{1\le i<j\le 2n} \big)  \,,
\ee
which is the claimed statement. 

Two different, but not unrelated, mechanisms may make the probability of avoided intersections (for sets of widely separated source points) small for the Ising model at its critical point.   

The first reason is the expected fractality of the connecting paths (and even better, of the connected clusters).    When fractal paths appear to pass over each other, they do so with multiple switchbacks.  Consequently, for paths  whose end points are at a distance  $K\gg1$ apart, there typically would be a large number of opportunities for the actual intersection  (growing at least as $K^c$).  
The probability that all such opportunities to connect would be missed may be expected to vanish at a rate $\exp(-K^s)$ for some $s>0$.  This reason does not apply in case  the fractal paths pass over each other where only a few switchbacks occur, but this is also expected to have small probability, typically of order $K^{-t}$ for some $t>0$. 

While the above is very suggestive, and there exist possibly relevant results for proving path fractality~\cite{ AizBur99,CheDumHon13,DumHonNol11}, a proof along these lines would require some still unaccomplished technical work.   
However,  since what matters are connections with respect to the sum $\n_1+\n_2$ of two currents, it is possible to invoke a  related but slightly simpler argument leading to the same conclusion.  
For that, the essential result to prove is: 
\begin{theorem}\label{thm:condition}
For the Ising model on $\bbZ^2$ with couplings satisfying the conditions\footnote{By this, we mean that we consider the unique extension to $\bbZ^2$ of the coupling constants (defined on $\bbH$) considered in Theorem~\ref{thm:Pf_finite_range}.} (i)-(iv) listed in Theorem~\ref{thm:Pf_finite_range} at $\beta=\beta_c$, the infinite-volume sourceless random current $\n$  (as defined in Section~\ref{sec:infinite_n}) almost surely contains infinitely many odd circuits surrounding the origin.
\end{theorem}
By invariance under translations, the assertion made there holds equally well for any preselected point, including points identified as sites of close approach by the paths of $\n_1$ pairing the boundary vertices $x_1,\dots,x_{2n}$.  This observation will be used in the proof of  Theorem~\ref{thm:Pf_finite_range} where it is shown  that  avoided intersections are unlikely due to the existence, with high probability, of many odd circuits in $\n_2$ surrounding places where the previously mentioned paths in $\n_1$ cross. From this, one learns that the  paths in  $\n_1$ would typically be connected to each other through edges with positive $\n_2$-current.  This  argument is somewhat simpler to carry out than the one invoking path fractality, though, at their roots, the two mechanisms are related.   

Let us add that in the proof of Theorem~\ref{thm:condition}, we employ a coupling between three classical percolation models: 
\begin{itemize}[noitemsep,nolistsep]
\item[(i)] the  random even subgraphs obtained from the high-temperature expansion (sometimes called the loop $O(1)$ model), 

\item[(ii)]  the percolation model obtained by taking the trace of one random current, and 

\item[(iii)] the classical Fortuin-Kasteleyn percolation, also called the random-cluster model.   
\end{itemize}
We refer to Section~\ref{sec:coup} for  details on the coupling and its potential applications, of which more may be coming.  

\paragraph{Organization of the paper}   

As outlined above, in  analyzing the model we  employ a number of its known stochastic geometric representations which come with  useful combinatorial identities and  monotonicity  relations.     
For the completeness of presentation,  in Section \ref{sec:RC},  we briefly restate the random current representation and following that, in Section~\ref{sec:coup}, the other two.    
Also, we present there the coupling between the random even subgraphs, the random currents and the random clusters.   

 Section~\ref{sec:4} contains the proof of Theorem~\ref{thm:Pf_finite_range} {\em conditioned on} the validity of  Theorem~\ref{thm:condition}. The latter allows us to show that there are no avoided intersections far from the boundary.   A separate analysis  is needed to show that the claim is not defeated by avoided intersections near the boundary.  
 
 A stronger version (Theorem~\ref{prop:condition}) of Theorem~\ref{thm:condition} is formulated and proved in Section~\ref{sec:5}.   The proof goes 
 through an extension of the Russo-Seymour-Welsh theory to sourceless random currents which is made difficult by the absence of positive association (in particular the FKG inequality) for this model.

Section~\ref{sec:OD} is devoted to the discussion of the Pfaffian structure of order-disorder correlation functions, Theorem~\ref{thm:Pf_OD}.  As explained there, these offer an extension to the bulk of the structure which was first identified to apply to boundary correlation functions. 

In the Appendix, we include a number of relations which are used in our discussion. In Appendix A, an extension of the Messager-Miracle-Sol\'e inequality is proved for finite-range interactions. In Appendix B, we derive relevant ``gluing principles'', whose purpose may be easily grasped but whose proofs require one to struggle with certain technicalities.  
  
 \section{The random current representation}     \label{sec:RC}

The Ising model admits a number of useful stochastic geometric representations, which apply to all finite graphs.   
They are obtained by starting from one of the following three representations of the interaction term: 
\begin{align} \label{eq:J_exp}
\exp(\beta J_{x,y} \sigma_x \sigma_y) &=  \sum_{n=0}^\infty \frac{(\beta J_{x,y})^{n}}{ n!} (\sigma_x \sigma_y)^{n} \notag \\
& =   \cosh(\beta J_{x,y}) \big[ 1+ \tanh(\beta J_{x,y})  \sigma_x \sigma_y \big] \phantom{\sum_{n=0}} \notag\\
  &=   
   e^{-\beta J_{x,y}} + \big[e^{\beta J_{x,y}} - e^{-\beta J_{x,y}} \big] \, \mathbb{I}[\sigma_x = \sigma_y].
\end{align}  
In the case of ferromagnetic interactions, each of the resulting expansions yields a decomposition of 
 the partition function into a sum of positive terms, and along with that comes a decomposition of the expectation value functional which gives the thermal equilibrium state.   
The three resulting representations, called the random current, high temperature and random cluster representations, can be coupled in a manner indicated by the algebraic relations expressed in  \eqref{eq:J_exp}. This powerful tool will be discussed in Section~\ref{sec:coup}. Since the random current will be the main object of interest to us, we start by describing it in detail in this section.

\subsection{Currents and correlation functions}

The so-called random current representation is obtained by expanding the partition function using the first of the relations in \eqref{eq:J_exp}.  
Its early use was made by Griffiths-Hurst-Sherman in~\cite{GriHurShe70}.  The term was coined in  \cite{Aiz82} where it was recognized to offer a useful depiction and a convenient tool for the study of correlations among distant spins.   To cite but a few applications, it was instrumental in the proof  that  the nearest-neighbor Ising model's scaling limits are Gaussian in dimensions $d>4$~\cite{Aiz82}, and that the critical exponents take their mean field values when $d\ge4$~\cite{AizFer86,AizFer88}.  Further applications include proofs of sharpness of the phase transition \cite{AizBarFer87} (see also \cite{DumTas15}), and the vanishing of the spontaneous magnetization at the critical points in dimension 
$d \ge3$ \cite{AizDumSid15}.

Starting from the power series expansion of the exponential function (first line of \eqref{eq:J_exp}), the partition function
$$
Z(G,\beta):=  2^{-|V(G)|} \sum_{\sigma\in\{\pm1\}^{\V(\G)}}\prod_{\{x,y\}\in E(G)}\exp[\beta J_{x,y}\sigma_x\sigma_y] $$
can be presented~\cite{GriHurShe70} as
\begin{align} \label{eq:Z}
Z(G,\beta) &=  \sum_{ \n:  E(\G)  \to \N  } \frac{\displaystyle(\beta J_{x,y})^{{\mathbf n}(x,y)}}{{\mathbf n}(x,y)!}
\prod_{x\in V(\G)} \Big( \tfrac{1}{2}  \sum_{\sigma_x\in\{\pm1\}} \sigma_x^{\sum_y \n(x,y)} \Big)
\nonumber \\
&=    \sum_{ \n: \partial \n = \emptyset } 
w(\mathbf n),
\end{align}
where $\n$ is summed over integer-valued functions
\be 
\n:   E(\G)   \to \N :=  \{0,1,2,\dots\}\, ,
\ee
 which are referred to as {\em currents configurations}, with the weight function 
\be\nonumber
w(\mathbf n)=w_{\beta}(\mathbf n):=\prod_{\{x,y\}\in E(G)}\frac{\displaystyle(\beta J_{x,y})^{{\mathbf n}(x,y)}}{{\mathbf n}(x,y)!} \,   
\ee 
and 
\be 
\partial \n \ := \  \big\{ x\in  V(\G) \,:\,  (-1)^{\sum_y \n(x,y)} = -1 \big\} \, ,  
\ee 
to which we refer as the {\em set of sources} of $\n$.   

Underlying the random current terminology is the observation that 
any configuration $\n$  with a given set of sources $A=\partial \n$ can be associated (in a non-unique way) with a configuration of loops and paths 
 linking pairwise the sites of $A$, in such a way that the values of $\n$ over the edge-set $E(\G)$  give the total number of times the edges of $E(\G)$ are traversed by paths and loops in that set.  In particular, 
a configuration with 
$\partial\n = \{x,y\}$ can be viewed as giving the total ``flux numbers'' of a family  of loops together with a path from $x$ to $y$.

The strategy leading to \eqref{eq:Z} can be applied to spin correlations,  yielding: 
\begin{align} \label{eq:spin correlations}
\langle \prod_{j=1}^k \sigma_{x_j} \rangle_{\G,\beta}  
& :=  \frac{1}{ Z(G,\beta)} \, \, \frac{1}{2^{|V(\G)|} } \sum_{\sigma\in\{\pm1\}^{\V(\G)}}
\big[ \prod_{j=1}^k \sigma_{x_j} \big] 
\, \, \exp\big[{\sum_{\{x,y\} \in E(\G)}\beta  J_{x,y} \sigma_x \sigma_y}\big]
\nonumber \\ 
& \phantom{:}=   
%
\frac{1}{Z(G,\beta)}   \displaystyle \sum_{\n: \, \partial \n = \{x_1,\dots,x_k\} }  w(\n)
\, .
 \end{align}
 
We conclude this section by introducing two important normalized probability distributions. By \eqref{eq:spin correlations}, the following quantity is a probability distribution on random currents constrained by the source condition $\partial \n = A$:
\be \label{eq:Pr}
\mathbf P_{\G,\beta}^{A} \left(\n \right)  :=   
\frac{ w(\n)} {\langle \prod_{x\in A} \sigma_x \rangle_{G,\beta} \, Z(\G,\beta) }  \bbI [\partial \n = A].
\ee 
We should also be interested in stochastic systems consisting of pairs of independent random currents, for which we denote the product probability measure on pairs $(\n_1,\n_2)$   by
\be \label{eq:2Pr}
\mathbf P_{\G,\beta}^{A,B} \  := \ \mathbf P_{\G,\beta}^{A} \otimes \mathbf P_{\G,\beta}^{B}. \ee 
The subscript may be omitted when no confusion can arise, however the superscripts usually need to be kept since the source constraints may vary throughout the discussions of a given system.

\subsection{Random currents as multigraphs, and the switching principle}

A combinatorial symmetry of random currents is most readily recognized in the representation in which  
 a currents configuration $\n$ is presented  as a (multi-)graph $\mathcal N$ obtained by replacing each edge  $\{x,y\} \subset E(\G)$ by $\n(x,y)$ edges, all linking the same pair of vertices.   By default, we shall denote the multigraph corresponding to $\n$ or $\m$ by the appropriate calligraphic script symbol,  
$\mathcal N$ or $\mathcal M$ in the above examples.  
We extend the above correspondence to the  weight and source notation, so that 
$\partial \mathcal N := \partial \mathcal \n$ and $w(  \mathcal N)  := w( \mathcal \n)$, and similarly for $\mathcal M$ in relation to $\m$.   

\begin{figure}[h]
\begin{center}
\includegraphics [width = 0.70 \textwidth]{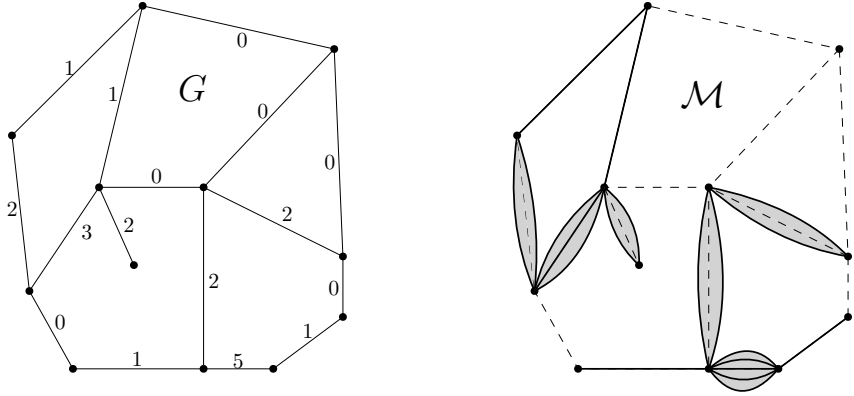}
\caption{On the left, the graph $G$ with a current $\m$ on it (the numbers denote the value of the current on each edge). On the right, the associated graph $\mathcal M$. In the planar case, some faces (denoted in gray) of the new graph do not correspond to any face of the original graph $G$.
\label{fig:graph_M}
}
\end{center}
\end{figure}
One may note that for every fixed $\m$ and for every pair of currents $(\n_1,\n_2)$ with $\m = \n_1+ \n_2$, we find that 
\be \label{comb}
w(\n_1) \, w(\n_2) = w(\m) \prod_{\{x,y\} \in \E(\G)} \binom{\m(x,y)}{\n_1(x,y)}
=: w(\m) \binom{\m}{\n_1}\,,
\ee 
where the combinatorial factor in \eqref{comb} coincides with the number of ways of partitioning   $\mathcal M$ into two multigraphs with a common vertex-set and edge numbers given by $\n_1$ and $\n_2$ respectively.  This can be viewed as a divisibility property of the weight distribution.   

As a consequence, we have two equivalent descriptions of pairs of random currents:  

(i) as  pairs of integer-valued functions $(\n_1,\n_2)$, each defined over $E(\G)$, 

(ii) as pairs of nested multigraphs $\mathcal N \subset \mathcal M$, constructed over the  vertex-set  $V(\G)$.   
The edge multiplicity of  $\mathcal M$ being given by $\n_1+\n_2$, and that of $\mathcal N$ by $\n_1$.   

Equation \eqref{comb} implies the following relation between these two representations:   for any given function $F(\n_1,\n_2)$ defined over pairs of currents, 
\be\label{eq:key}
\sum_{\substack{\n_1:\partial \n_1=A\\ \n_2:\partial \n_2=B }} w(\n_1)w(\n_2) F(\n_1,\n_2)
\  =\  \sum_{\substack{\mathcal M :\, \partial \mathcal M=A\Delta B}} w(\mathcal M)\sum_{\substack{\mathcal N\subset\mathcal M  \\  \partial\mathcal N=A} }F(\mathcal N,\mathcal M\setminus \mathcal N)\, , 
\ee 
where $\Delta $ denotes the symmetric difference operator on sets, and it is taken as understood that 
$\mathcal M$ and $\mathcal N$ range over multigraphs with vertex-set $V(\G)$ with certain constraints.

Thus, the random current representation and its multigraph version yield the following expression for the product of correlation functions:
\begin{align}  \nonumber
\langle \prod_{x\in A} \sigma_{x} \rangle_{\G,\beta}  \times  
\langle \prod_{y\in B} \sigma_{y} \rangle_{\G,\beta}    
&=   
\frac{1}{Z(G,\beta)^2}      \sum_{\substack {\n_1: \, \partial \n_1 = A  \\ \n_2: \, \partial \n_2 = B} } 
         w(\n_1) \,  w(\n_2) \, \, 
      \\
      \label{eq:two_corr}
 &=    \frac{1}{Z(G,\beta)^2}       \sum_{\mathcal M: \, \partial \mathcal M = A\Delta B } 
         w(\mathcal M) \,   \, 
          \sum_{\substack {\mathcal N \subset \mathcal M \\  \partial \mathcal N = A } }  1\, \, 
             \, .  
 \end{align} 
This relation motivates the following elementary combinatorial observation: 

\begin{lemma}[Switching principle] \label{lem:switch}
 For any multigraph $\mathcal M$  with  vertex-set $V(\G)$,  any set $A\subset V(\G)$ and any function $ f $ of a current:
\be \label{eq:switch} 
  \sum_{\substack {\mathcal N \subset \mathcal M \\  \partial \mathcal N = A } }  f(\mathcal{N})  \ = \  
 \begin{cases}   
       \displaystyle \sum_{\substack {\mathcal N \subset \mathcal M \\  \partial \mathcal N = \emptyset } }    f(\mathcal{N} \Delta \mathcal{K}) 
 & \mbox{if there is at least one  $\mathcal K \subset \mathcal M$ with  $ \partial \mathcal K = A$},  \\[6ex]  
 0 & \mbox{otherwise} .
 \end{cases}
\ee 
\end{lemma} 
\begin{proof} 
If the multigraph $\mathcal M$ has a subgraph $\mathcal K$ with $ \partial \mathcal K = A$, then the mapping $\mathcal N \mapsto  \mathcal N \Delta \mathcal K$ yields a bijection between the collection of subgraphs of $\mathcal M$ with $\partial \mathcal N = A$ and the collection of subgraphs with $\partial \mathcal N = \emptyset$, hence the first alternative.  The second is evident. 
\end{proof} 
In applying Lemma~\ref{lem:switch}, one may bear in mind that in the case where $A$ contains exactly two points, e.g.~$A = \{x,y\}$, the condition in  \eqref{eq:switch} is equivalent to the fact that $x$ and $y$ are connected in $\mathcal M$, meaning that they are in the same connected component of  $\mathcal M$.
Lemma~\ref{lem:switch} has the following useful   implication.  
\begin{corollary}\label{cor:switch}  For any finite graph, set of couplings, and collection of sites $\{x_1,\dots,x_{2n}\}\subset V(\G)$, we have 
\begin{align}\label{eq:aza}
 \langle \sigma_{x_1} \sigma_{x_{\ell}}\rangle_{\G,\beta}   \, \,
\big\langle \prod_{\substack{1\le j\le 2n \\ j \notin  \{1,\ell\} }}\sigma_{x_j} \big\rangle_{\G,\beta} 
&= \langle \sigma_{x_1}\cdots \sigma_{x_{2n}}\rangle_{\G,\beta}\cdot 
\mathbf P_{\G,\beta}^{\{x_1,\dots,x_{2n}\}, \emptyset} [x_1 \stackrel{\n_1+\n_2} {\longleftrightarrow} x_\ell] \, .
\end{align}
\end{corollary} 
\begin{proof} 
Through the combination of  \eqref{eq:two_corr} with \eqref{eq:switch}, one gets 
\begin{multline}  \label{eq:101}
\langle \sigma_{x_1} \sigma_{x_{\ell}}\rangle_{\G,\beta}   \, \,
\big\langle \prod_{\substack{1\le j\le 2n \\ j \notin  \{1,\ell\} }}\sigma_{x_j} \big\rangle_{\G,\beta}   
\ = \    \frac{1}{Z(G,\beta)^2}       \sum_{\mathcal M: \, \partial \mathcal M = \{x_1,...,x_{2n}\} } 
         w(\mathcal M) \,   \mathbb{I} [x_1 \stackrel{\mathcal M}  {\longleftrightarrow} x_\ell]
          \sum_{\substack {\mathcal N \subset \mathcal M \\  \partial \mathcal N = \emptyset } }  1\, \\[2ex]  
\ = \   \big\langle \prod_{ 1\le j\le 2n }\sigma_{x_j} \big\rangle_{\G,\beta} 
      \frac{\displaystyle\sum_{\substack{ \n_1: \, \partial \n_1 = \{x_1,...,x_{2n}\} \\ 
      \n_2: \, \partial \n_2 =\emptyset}}
         w(\n_1) \,   w(\n_2)  \, \,  \mathbb{I} [x_1 \stackrel{\n_1+\n_2}  {\longleftrightarrow} x_\ell]}{\displaystyle\sum_{\substack{ \n_1: \, \partial \n_1 = \{x_1,...,x_{2n}\} \\       \n_2: \, \partial \n_2 =\emptyset}}
         w(\n_1) \,   w(\n_2)}\, \, 
             \, .  
             \hspace{2cm} 
 \end{multline} 
The last fraction on the right is, by the definition \eqref{eq:2Pr}, $\mathbf P_{\G,\beta}^{\{x_1,\dots,x_{2n}\},\emptyset} [x_1 \stackrel{\n_1+\n_2} {\longleftrightarrow} x_\ell]$. 
\end{proof} 

\begin{remark}The identity \eqref{eq:aza}  readily implies the relations  \eqref{U4}, \eqref{P4_avoided}, and \eqref{eq:Q1}, whose consequences  were discussed in the introduction.\end{remark}

\subsection{Infinite volume limit}  \label{sec:infinite_n}

It is useful to know that the probability measures  ${\bf P} _{\G,\beta}^A$ (and correspondingly ${\bf P} _{\G,\beta}^{A,B}$) can be defined for infinite graphs (for instance $G$ equal to $\bbZ^2$ and $\bbH$) by taking the limit of measures on finite subgraphs.   For ferromagnetic interactions,  the corresponding measures converge in the natural (weak) sense.   We call any such  limit an  {\em infinite-volume random-current measure}.

In the case of translation invariant interaction and $A=\emptyset$, the probability measure ${\bf P}_{\Z^2,\beta}^\emptyset$ is invariant under translation and ergodic.  Proofs and an application of these  statements can be found in \cite{AizDumSid15}. The general case of $A$ arbitrary follows from the same proofs.

\section{RC coupling with the FK random cluster model and the vdW high-temperature expansion}\label{sec:coup}

\subsection{Three percolation models}

A percolation model is described by a collection of binary-valued random variables $\eta : \ E(\G) \to \{0,1\}$,  each  indicating whether the corresponding bond is {\em open}, i.e.~connecting (if $\eta(x,y) =1$),  or not  ($\eta(x,y) =0$). The configuration $\eta$ can be seen as a subgraph of $G$ with vertex-set $V(G)$ and edge-set $E(\omega):=\{\{x,y\}\in E(G):\eta(x,y)=1\}$. We may therefore speak of connection properties of $\eta$ as being those of the corresponding graph.
We say that $A\leftrightarrow B$ if there exists a connected component of $\eta$ intersecting both $A$ and $B$. We say that $0\leftrightarrow\infty$ if the connected component of 0 is infinite.
 
In this section, we describe three such models which are naturally related with the Ising spin system.  Each has already been employed successfully for some insight on the model.   We bring them up together to emphasize the existence of a coupling between the three percolation models.  Properties of this coupling, stated in Theorem~\ref{thm:connection} below, will be used extensively in the proof of  Theorem~\ref{thm:condition}.   We expect  this structure to be of interest for  many more results on the Ising model.

The three percolation models are: 
 \bigbreak
\noindent{\em A. Random currents.}    For each currents configuration $\n$, 
let  $\widehat\n$ be the percolation configuration defined by
\be 
\widehat\n(x,y) :=  \bbI[ \n(x,y) \neq 0].
\ee
 In other words, with respect to this process, $\{x,y\}$ is open if and only if the current is positive at $\{x,y\}$.
We denote by  $\widehat{\bf P}^A_{G,\beta}$  the push forward of ${\bf P}^A_{G,\beta}$ under the map 
$\n\mapsto \widehat\n$.
 \bigbreak
\noindent{\em B. High-temperature expansion.}  Starting again from a currents configuration $\n$, let a percolation configuration $\eta$ be defined as the parity of $\n$, i.e.
\be 
\eta(x,y)  :=    \bbI[ \n(x,y) \text{ is odd}] .
\ee  
This mapping preserves the set of sources
\be
\partial \eta (\n)  :=\big\{x\in V(G):\sum_{y}\eta(x,y)\text{ odd} \big\}  = \partial \n
\ee 
so that we may define,   
for each $A\subset V(G)$, the push forward ${\rm P}^A_{G,\beta}$ of ${\bf P}^A_{G,\beta}$ by the map $\n\mapsto\eta(\n)$, which is given by
\be {\rm P}^A_{G,\beta}[\eta]=\begin{cases}\displaystyle\frac{x_\beta(\eta)}{Z_{\mathrm{HT}}(G,\beta,A)} &\text{if $\partial\eta=A$},\\
\qquad 0&\text{otherwise},
\end{cases}
\ee
where $Z_{\mathrm{HT}}(G,\beta,A)$ is a normalizing constant and
$$x_\beta(\eta):=\prod_{\{x,y\}\in E(G):\eta(x,y)=1}\tanh(\beta J_{x,y}).$$
This measure naturally emerges by expanding the partition function using 
$$
\exp(\beta J_{x,y}\sigma_x\sigma_y)=\cosh(\beta J_{x,y})(1+\tanh(\beta J_{x,y})\sigma_x\sigma_y)
$$
(i.e.~the second expression in \eqref{eq:J_exp}) instead of the Taylor expansion of $\exp(\beta J_{x,y}\sigma_x\sigma_y)$. 
The resulting expansion is known as the Ising model's {\em high-temperature expansion}.  Its  early appearance can be found in the work of van der Waerden~\cite{vdW}.
 \bigbreak
\noindent{\em C. The random-cluster model.} The random-cluster model (or Fortuin-Kasteleyn percolation) on $G$ with free boundary conditions is defined as follows.  For a percolation configuration $\omega$, set
$$\phi_{G,\beta}[\omega]:=\frac1{Z_{\mathrm{RCM}}(G,\beta)} \big(\prod_{\{x,y\}\in E(G):\omega(x,y)=1} p_{x,y}\big)\cdot\big( \prod_{\{x,y\}\in E(G):\omega(x,y)=0}(1-p_{x,y}) \big)\cdot 2^{k(\omega)},$$
where $k(\omega)$ is the number of connected components of $\omega$, $p_{x,y}:=1-\exp[-2\beta J_{x,y}]$ and $Z_{\mathrm{RCM}}(G,\beta)$ is a normalizing constant.

Among the convenient features of the model, which are discussed in greater detail in \cite{ACCN88,Dum17,Gri06},  one finds: 
\begin{enumerate} 
\item[{\it i)}] There is no source constraint on the percolation configurations.
\item [{\it ii)}]   
The model is {\em positively correlated}, i.e.~that it satisfies the FKG inequality \cite[Theorem~3.8]{Gri06}
\be\label{eq:FKG}
\phi_{G,\beta}[A\cap B]\ge \phi_{G,\beta}[A]\phi_{G,\beta}[B]
\ee
for every increasing events $A$ and $B$ (an event $A$ is {\em increasing} if $\omega\le\omega'$ and $\omega\in A$ implies $\omega'\in A$).  
\item [{\it iii)}] It allows to embed the Ising model within the family of random-cluster models, which also includes  other important models (for instance Bernoulli percolation) of statistical mechanics with which it is linked by useful inequalities.  
\end{enumerate} 

The positive association enables one to define infinite-volume measures \cite[Definition~4.15]{Gri06}.  Of particular interest for us will be the measures $\phi_{\mathbb Z^2,\beta}$ and $\phi_{\mathbb H,\beta}$ defined on the plane $\mathbb Z^2$ and the upper half-plane $\mathbb H$. They are obtained by taking limits of measures in finite volume and are called measures with free boundary conditions.

Within this model, the Ising spin correlation functions are rewritten as follows.
\begin{proposition}\label{prop:coupling} For the Ising model on a finite graph $\G$ and any $A\subset V(G)$, we have
\be 
\big\langle\prod_{x\in A}\sigma_x\big\rangle_{G,\beta}=\phi_{G,\beta}[\mathcal F_A],
\ee 
where $\mathcal F_A$ is the set of percolation configurations for which each connected component intersects $A$ an even number of times. \end{proposition}

Note that $A$ must have even cardinality for the terms in the previous equality to be non-zero.

\subsection{The three-way  coupling}\label{sec:coupling}

In this section, we discuss a coupling between the high-temperature expansion, random current and random-cluster models. While this coupling was already presented partially before in \cite{GriJan09,LupWer15} (see below for more details), we wish to highlight some
 important applications that will be used in the next sections. 
 
Consider the following coupling between $\eta$, $\widehat \n$ and $\omega$. First, define $\eta$ to be sampled according to ${\rm P}^A_{G,\beta}$. Then, define $\widehat \n$ obtained from $\eta$ by the following formula
$$\widehat \n(x,y):=\max\{ \eta(x,y), \alpha(x,y)\},$$
where the $\alpha(x,y)$ are independent Bernoulli random variables of parameter $1-1/\cosh(\beta J_{x,y})$. Finally, let $\omega$ be obtained from $\widehat \n$ by the formula
$$\omega(x,y):=\max\{ \widehat\n(x,y), \beta(x,y)\},$$
where the $\beta(x,y)$ are independent Bernoulli random variables of parameter $1-\exp(-\beta J_{x,y})$.
Note that by construction, $\omega(x,y)=\max\{\eta(x,y),\gamma(x,y)\}$, where the $\gamma(x,y)$ are independent Bernoulli random variables of parameter $r_{x,y}:=\tanh(\beta J_{x,y})$.

\begin{theorem}\label{thm:connection}
For any finite graph $G$ and any $A\subset V(G)$, consider a triplet $(\eta,\widehat\n,\omega)$ constructed as above. Then, we have that
\begin{itemize}[noitemsep,nolistsep]
\item $\eta$ has law ${\rm P}^A_{G,\beta}$,
\item $\widehat\n$ has law $\widehat{\bf P}^A_{G,\beta}$,
\item $\omega$ has law $\phi_{G,\beta}[\,\cdot\,|\mathcal F_A]$, where $\mathcal F_A$ is defined in Proposition~\ref{prop:coupling}.
\end{itemize}
\end{theorem}

In the case $A=\emptyset$, the connection between $\eta$ and $\omega$ was proved in \cite{GriJan09} by Grimmett and Janson, and the connection between $\widehat \n$ and $\omega$ was presented in \cite{LupWer15} by Lupu and Werner. Since we mention the coupling in the general case of arbitrary $A$, and since we wish to highlight some consequences below, we include the proof for completeness.

\begin{proof}
The first bullet is trivial by construction. For the second bullet, the expansion leading to $\eta$ can be understood as a partially resumed expansion in currents so that $\eta(x,y)=1$ if and only if $\n(x,y)$ is odd. Now, conditioned on $\n(x,y)$ being even, we readily obtain that the probability of being equal to 0 divided by the probability of being equal to a strictly positive (even) number is equal to
$1/\cosh(\beta J_{x,y})$. This implies that conditioned on $\n(x,y)$ being even, $\n(x,y)>0$ with probability $1-1/\cos(\beta J_{x,y})$ independently of $\eta(x,y)$. The second bullet follows.

It remains to prove the third bullet. In order to do this, we inspect the joint law $\mathbb P$ of $(\eta,\omega)$. First, observe that it is supported on the space $\Omega_A$ of pairs $(\eta,\omega)$ satisfying that $\eta\le \omega$, $\partial\eta=A$, and $\omega\in \mathcal F_A$. Indeed, the last assertion follows from the fact that $\eta\in \mathcal F_A$ due to its source constraints, and that therefore $\omega\ge\eta$ must be in $\mathcal F_A$ (since this event is increasing).

Furthermore, using that $r_{x,y}=\tanh(\beta J_{x,y})$ and $p_{x,y}=1-e^{-2\beta J_{x,y}}$, we find that 
\begin{align*}\mathbb P[(\eta,\omega)]&=\frac{\mathbb I[(\eta,\omega)\in \Omega_A]
}{Z_0}\Big(\prod_{\{x,y\}\in E(G):\eta(x,y)=1}\tanh(\beta J_{x,y})\Big)\\
&\qquad\qquad\qquad\qquad\times\Big(\prod_{\{x,y\}\in E(G):\omega(x,y)=1,\eta(x,y)=0} r_{x,y}\Big)\Big(\prod_{\{x,y\}\in E(G):\omega(x,y)=0} (1-r_{x,y})\Big)\\
&=\frac{\mathbb I[(\eta,\omega)\in \Omega_A]
}{Z_0}\Big(\prod_{\{x,y\}\in E(G):\omega(x,y)=1}r_{x,y}\Big)\Big(\prod_{\{x,y\}\in E(G):\omega(x,y)=0} (1-r_{x,y})\Big)\\
&=\frac{1}{Z_0\prod_{\{x,y\}\in E(G)}(1+e^{-2\beta J_{x,y}})/2}\times\mathbb I[(\eta,\omega)\in \Omega_A]\\
&\qquad\qquad\qquad\qquad\qquad\qquad
\times \Big(\prod_{\{x,y\}\in E(G):\omega(x,y)=1} \frac{p_{x,y}}2\Big)\Big(\prod_{\{x,y\}\in E(G):\omega(x,y)=0}(1-p_{x,y})\Big),\end{align*}
where $Z_0$ is a normalizing constant. 

Now, define the following measure $\widetilde {\mathbb P}$ on $\Omega_A$ as follows. Choose $\widetilde \omega$ according to the law $\phi_{G,\beta}(\cdot|\mathcal F_A)$ and then choose  $\widetilde\eta$ uniformly in the set $$S_A(\widetilde\omega):=\{\widetilde\eta\subset\widetilde \omega:\partial\widetilde\eta=A\}.$$
The number of even subgraphs (i.e.~graphs $\widetilde \eta$ such that $\widetilde\eta\subset\widetilde\omega$ and $\partial\widetilde\eta=\emptyset$) of $\tilde\omega$ is given by $2^{k(\widetilde\omega)+|E(\widetilde\omega)|-|V(\widetilde\omega)|}$, where $E(\widetilde\omega)$ and $V(\widetilde\omega)=V(G)$ are the sets of edges and vertices of $\widetilde \omega$. Since $\widetilde\omega\in \mathcal F_A$, there exists a graph $\widetilde\eta_0\subset\widetilde\omega$ with $\partial\widetilde\eta_0=A$. The mapping $\widetilde\eta\mapsto\widetilde\eta\Delta \widetilde\eta_0$ (the symmetric different should be understood in terms of edge-sets of the respective graphs here) bijectively maps the set of even subgraphs of $\widetilde\omega$ to $S_A(\widetilde\omega)$ so that
$$|S_A(\widetilde \omega)|=2^{k(\widetilde\omega)+|E(\widetilde\omega)|-|V(G)|}.$$
Since $\widetilde\eta$ is chosen uniformly in $S_A(\widetilde\omega)$, we deduce that
\begin{align*}\widetilde{\mathbb P}[(\widetilde\eta,\widetilde\omega)]&=\frac{\mathbb I[(\widetilde\eta,\widetilde\omega)\in \Omega_A]
}{Z_1}\Big(\prod_{\{x,y\}\in E(G):\widetilde\omega(x,y)=1}p_{x,y} \Big)\\
&
\qquad\qquad\qquad\qquad\quad\quad\ \times\Big(\prod_{\{x,y\}\in E(G):\widetilde\omega(x,y)=0}(1-p_{x,y})\Big)\times 2^{k(\widetilde\omega)}\times 2^{-k(\widetilde\omega)-|E(\widetilde\omega)|+|V(G)|}\\
&=\frac{
2^{|V(G)|}\mathbb I[(\widetilde\eta,\widetilde\omega) \in \Omega_A]}{Z_1}\times \Big(\prod_{\{x,y\}\in E(G):\widetilde\omega(x,y)=1}\frac{p_{x,y}}{2}\Big)\times \Big(\prod_{\{x,y\}\in E(G):\widetilde\omega(x,y)=0}(1-p_{x,y})\Big),\end{align*}
where $Z_1$ is a normalizing constant.

Altogether, the measures $\mathbb P$ and $\widetilde{\mathbb P}$ are equal. Since the second marginal of $\widetilde{\mathbb P}$ has law $\phi_{G,\beta}[\cdot|\mathcal F_A]$, this implies the third bullet and concludes the proof.
\end{proof}

The proof of Theorem~\ref{thm:connection} is instructive since it tells us what is the recipe to obtain a high-temperature expansion in terms of the random-cluster configuration: one should take uniform subgraphs with $\partial\eta=A$. Note that for arbitrary $A$, picking a uniform subgraph with $\partial\eta=A$ boils down to picking a uniform even subgraph and then taking the symmetric difference with a subgraph $\eta_0$ with $\partial\eta_0=A$ chosen in advance.

For $A=\emptyset$, the space $S_\emptyset(\omega)$ of even subgraphs of $\omega$ can be seen as a finite $(\mathbb Z/2\mathbb Z)$-vector space with the symmetric difference $\Delta$ working as a sum. Therefore, a way of sampling uniformly a random even subgraph of $\omega$ consists in first choosing a basis $(\eta_1,\dots,\eta_k)$, second, choosing  each element of this basis independently with probability 1/2 and, third, taking the symmetric difference of these subgraphs.

\begin{remark}There is a natural way of choosing a basis. First, pick a spanning tree $T$ of $\omega$. Second, observe that for each pair $\{x,y\}\in E(\omega)$ which is not in $T$, there exists a unique cycle $\ell(x,y)$ composed of $\{x,y\}$ together with edges of $T$.  The family of cycles $\ell(x,y)$ with $\{x,y\}$ running over every edge of $\omega\setminus T$ is then a basis for our vector space.
\end{remark}

We now list a few easy consequences of this coupling which will be very useful.
\begin{corollary}\label{rmk:cycle}
Conditioned on $\omega$, the law of $\eta$ and $\eta\Delta c$ are the same for each cycle $c\subset \omega$.
\end{corollary}

\begin{proof}
The theorem of the incomplete basis applied to the previous $(\mathbb Z/2\mathbb Z)$-vector space allows to pick, for each cycle $c\subset\omega$, a basis including $c$. The proposition follows readily.\end{proof}

Another simple corollary of the previous coupling is the following useful infinite-volume convergence. Note that it is a priori difficult to define infinite-volume versions of the law of $\eta$ and $\widehat\n$.
\begin{corollary}
If $\phi_{\mathbb Z^2,\beta}[0\leftrightarrow \infty]=0$, then ${\rm P}_{G,\beta}^\emptyset$ and $\widehat{\bf P}_{G,\beta}^\emptyset$ converge weakly as $G$ tends to $\mathbb Z^2$ to two measures ${\rm P}_{\mathbb Z^2,\beta}^\emptyset$ and $\widehat{\bf P}_{\mathbb Z^2,\beta}^\emptyset$. The same holds for the upper half-plane $\mathbb H$ (we write ${\rm P}_{\mathbb H,\beta}^\emptyset$ and $\widehat{\bf P}_{\mathbb H,\beta}^\emptyset$ for the respective measures).
\end{corollary}

\begin{proof}
The proof is straightforward since the coupling can be defined locally in each of the (finite) connected components of $\omega\sim\phi_{\mathbb Z^2,\beta}$.
\end{proof}

The coupling is built in such a way that $\eta$ somehow encodes the dependencies in $\omega$. We illustrate this fact in the following corollary.

\begin{corollary}\label{cor:correlation}
For any $n\le N$ and any two increasing events $A$ and $B$ depending on edges with both endpoints in $V$ and $W$ respectively, we have
\be\label{eq:1a}
\phi_{G,\beta}[A]\phi_{G,\beta}[B]\le \phi_{G,\beta}[A\cap B]\le \phi_{G,\beta}[A]
\phi_{G,\beta}[B]+\mathrm{P}_{G,\beta}^\emptyset[V\stackrel{\eta}{\longleftrightarrow} W].\ee
\end{corollary}

Note that the left-hand side is simply the FKG inequality \eqref{eq:FKG}. The non-trivial point is the right-hand side.

\begin{remark}\label{rem:3}
Later, we will use the following direct consequence of Corollary~\ref{cor:correlation}:
  \begin{equation}
\phi_{\mathbb Z^2,\beta}[A]\phi_{\mathbb Z^2,\beta}[B]\le \phi_{\mathbb Z^2,\beta}[A\cap B]\le \phi_{\mathbb Z^2,\beta}[A] \phi_{\mathbb Z^2,\beta}[B]+\widehat{\mathbf{P}}_{\mathbb Z^2,\beta}^\emptyset[V\stackrel{\widehat\n}{\longleftrightarrow} W].\label{eq:1}
\end{equation}
To prove the last inequality, first, use that $\eta$ is stochastically dominated by $\widehat\n$ to obtain a finite-volume version of \eqref{eq:1a}. Second, consider a sequence of finite graphs converging to $\mathbb Z^2$ and use the definition of $\widehat{\bf P}_{\bbZ^2,\beta}^\emptyset$.
\end{remark}

\begin{proof}
As mentioned above, the left-hand side is a consequence of \eqref{eq:FKG}. For the right-hand side, consider the coupling $\mathbb P$  between $\omega$ and $\eta$ introduced in Theorem~\ref{thm:connection}.

Consider the random set $\mathcal S$ of vertices in $G$ connected to $V$ in $\eta$. Conditioned on $\mathcal S=S$, the distribution of $\eta$ outside $S$ is the same as the one of the  sourceless high temperature expansion in $G\setminus S$. Since the coupling between $\eta$ and $\omega$ consists in opening edges independently of $\eta$, we deduce that, conditioned on $\mathcal S=S$ and anything related to the random-cluster configuration inside $S$, the distribution of the random-cluster model outside $S$ is the same as $\phi_{G\setminus S,\beta}$. Therefore,
\begin{align*}\mathbb P[\{\omega\in A\cap B\}\cap \{V\stackrel{\eta}{\kern+4.5pt\arrownot\kern-4.5pt\longleftrightarrow}W\}]&=\sum_{S\cap F=\emptyset}  \mathbb P[\{\omega\in A\}\cap\mathcal S=S]\times\phi_{G\setminus S,\beta}[B]\\
&\le  \sum_{S\cap F=\emptyset} \mathbb P[\{\omega\in A\}\cap\mathcal S=S]\times\phi_{G,\beta}[B]\\
&= \mathbb P[\{\omega\in A\}\cap \{V\stackrel{\eta}{\kern+4.5pt\arrownot\kern-4.5pt\longleftrightarrow}W\}]\times\phi_{G,\beta}[B]\le \phi_{G,\beta}[A]\phi_{G,\beta}[B].\end{align*}
In the first inequality, we used that $W\subset G\setminus S$ and the monotonicity with respect to the coupling constants for the random-cluster model; see e.g.~\cite[Lemma 4.14]{Gri06}.
Altogether, we find that
$$\phi_{G,\beta}[A\cap B]\le \mathbb P[\{\omega\in A\cap B\}\cap \{V\stackrel{\eta}{\kern+4.5pt\arrownot\kern-4.5pt\longleftrightarrow}W\}]+\widehat{\bf P}^\emptyset_{G,\beta}[V\stackrel{\eta}{\longleftrightarrow}W],$$
which, combined with the last displayed inequality, concludes the proof readily.
\end{proof}

\section{Proof of Theorem~\ref{thm:Pf_finite_range}}\label{sec:4}

\subsection{Notation}\label{sec:4.1}

Below, $J$ will be a fixed set of coupling constants on $\mathbb Z^2$ satisfying the hypotheses (i)--(iv) of Theorem \ref{thm:Pf_finite_range} (which also make sense for coupling constants defined over the entire $\mathbb Z^2$), and let $\G$ be the associated graph; that is, the vertex-set of $\G$ is $\mathbb Z^2$ (or $\mathbb H$ as will be clear from the context), and the edge-set consists of the pairs $\{x,y\}$ with $J_{x,y} >0$. 
  
 For $N>0$, we denote the \emph{box} of size $N$ around $0$ and the one centered at $x$  by $\Lambda_N:=[-N,N]^2$ and $\Lambda_N(x):=x+\Lambda_N$, correspondingly.  We denote its boundary by $\partial\Lambda_N$ (seen as a subset of $\mathbb R^2$).

Below, based on this graph structure, we define the notion of paths (which correspond to edge-self-avoiding paths in the discrete graph) and continuous paths (parametrized continuous curves in $\mathbb R^2$ associated with these paths).    

\begin{definition}We call \emph{path} a sequence of vertices $\pi=(v_0,\dots,v_k)$ such that the unordered pairs $\{v_0,v_1\}, \{v_1,v_2\},\ldots, \{v_{k-1},v_k\}$ are distinct edges of $G$. \end{definition}
A path is always assumed to be edge-avoiding but not necessarily vertex-avoiding. Given a percolation configuration $\omega$ on $E$, the path $\pi$ is said to be \emph{open in $\omega$} if $\omega(v_{i-1},v_{i})=1$ for every $i$. Given a current $\mathbb \n$ on $E$,  the path $\pi$ is said to be \emph{odd in $\n$} if $\n(v_{i-1},v_{i})$ is odd for every $1\le i\le k$. 

To every path $\pi$, we associate a continuous (parametrized) curve  $\gamma:[a,b]\rightarrow \bbR^2$ which proceeds in  straight segments, and at constant speed, between $v_{i-1}$ and $v_i$, for every $1\le i\le k$. 

\begin{definition} We call  \emph{continuous path} any restriction to an interval $[c,d]\subset [a,b]$  of a continuous curve associated with a path from a point $a$ to $b$. \end{definition}

With this definition, a continuous path may start and end anywhere on the projection of an edge. A continuous path is said to be \emph{open} in $\omega$, resp.\@  \emph{odd} in $\n$, if it is the restriction of a continuous curve associated with an open path in $\omega$, resp.\@  odd in $\n$.

The notion of  continuous path is important when discussing topological properties of paths (for instance, existence of circuits, existence of a path from $A$ to $B$ passing below a certain point, etc.).   Two continuous paths may intersect (as planar curves) while their corresponding discrete paths do not (in the sense that they do not share a vertex). This phenomenon is the main difficulty arising when working with a general finite-range models on $\mathbb Z^2$ (compared to planar models).   
\begin{definition}
 A {\em circuit} (or a \emph{cycle}) is a path  $\pi=(v_0,\dots,v_k)$ with $v_k=v_0$. A circuit is said to {\em surround} a bounded connected set $S$ if the associated continuous path cannot be retracted into a point in $\bbR^2\setminus S$. 
\end{definition}

\subsection{Setting of the proof}

In the whole section, we assume that the phase transition is continuous, or more precisely that $\langle\sigma_0\sigma_x\rangle_{\mathbb H,\beta}$ tends to zero as $x\in\partial\mathbb H$ tends to infinity. Note that this assumption is harmless. If this is not the case, there exists $m>0$ such that $\langle\sigma_{x_i}\sigma_{x_j}\rangle_{\mathbb H,\beta}$ tends to $m^2$ uniformly in $\|x_i-x_j\|\rightarrow\infty
$. The mixing property \cite[Chapter 4]{Gri06} of the random-cluster measure then
 implies that $\langle\sigma_{x_1}\cdots\sigma_{x_{2n}}\rangle_{\mathbb H,\beta}\longrightarrow m^{2n}$, which is coherent with the Pfaffian formula of Theorem~\ref{thm:Pf_finite_range}. In conclusion, the case where $\langle\sigma_0\sigma_x\rangle_{\mathbb H,\beta}$ tends to 0 is the only case in which we need to prove something non-trivial. Since this can occur at the critical point only (this fact is easy to prove using the result of \cite{ManRao16}), $\beta$ will be assumed to be equal to $\beta_c$ (we drop it from the notation).

As mentioned in the introduction, the proof runs along the same lines as the proof of Theorem~\ref{thm:pf_boundary} except that \eqref{eq:main_boundary} is not true anymore. In fact, we will work with a slightly weaker inequality than \eqref{eq:p1}, which makes the notion of pairing sources in $\n_1$ more concrete. For this, we shall use what is known as the current's backbone~\cite{Aiz82}.  Its definition involves 
an arbitrarily preselected order $\prec$ on the collection of the graph's oriented edges, but this choice does not affect the analysis. Below, we say that two paths (which are necessarily edge-self-avoiding by definition) are {\em disjoint} if they do not use any common edge. 

\begin{definition}
Let  $\n$ be a currents configuration with only finite clusters, and $A= \partial\n$ the set of its sources.  Its {\em backbone}, denoted  $\g=\g(\n)$, is the collection of oriented disjoint paths $\gamma_1,\dots,\gamma_n$ such that:
\begin{itemize}[noitemsep,nolistsep]
\item the paths are supported on edges $\{x,y\}$ with $\n(x,y)$ odd, 
\item the paths start and end in $A$, and pair  the vertices in $A$, 
\item the collection of paths is minimal for the lexicographic order (induced by   the preselected order  $\prec$) among collections of paths having the two first properties. 
\end{itemize}
\end{definition}

We will base our analysis on the following inequality.

\begin{lemma}\label{lem:pp}For any pair of currents $(\n_1,\n_2)$ with $\partial\n_1=X$ and $\partial\n_2=\emptyset$, we have that\begin{align} \label{eq:main_boundary_finite_range}
\Big|\sum_{\ell=1}^{2n} (-1)^{\ell}  \,
 \mathbb{I} [x_1 \stackrel{\n_1+\n_2}{ \longleftrightarrow} x_\ell  ]\Big|
 &\le n\sum_{i<j<k<\ell} \ \mathbb I[\{x_i \stackrel{\g_1}{ \longleftrightarrow} x_k\}\cap\{x_j \stackrel{\g_1}{ \longleftrightarrow} x_\ell\}\cap\{x_i \stackrel{\n_1+\n_2}{\kern+4.5pt\arrownot\kern-4.5pt\longleftrightarrow} x_j\}],
\end{align}
where $\g_1$ is the backbone of $\n_1$.
\end{lemma}

\begin{remark}
One may wonder why we did not replace the condition on the right by one of the following weaker ones for which the previous claim would look more intuitive:
\begin{align*}&{\bf C1}\quad \{x_i \stackrel{\n_1+\n_2}{ \longleftrightarrow} x_k\}\cap\{x_j \stackrel{\n_1+\n_2}{ \longleftrightarrow} x_\ell\}\cap\{x_i \stackrel{\n_1+\n_2}{\kern+4.5pt\arrownot\kern-4.5pt\longleftrightarrow} x_j\}~(=A(i,j,k,\ell)),\\
&{\bf C2}\quad\{x_i \stackrel{\n_1}{ \longleftrightarrow} x_k\}\cap\{x_j \stackrel{\n_1}{ \longleftrightarrow} x_\ell\}\cap\{x_i \stackrel{\n_1+\n_2}{\kern+4.5pt\arrownot\kern-4.5pt\longleftrightarrow} x_j\}.\end{align*}
The reason will be apparent in the next steps of the proof: {\bf C1} would not be adapted since we want the connectivity conditions on $x_i\longleftrightarrow x_k$ and $x_j\longleftrightarrow x_\ell$ to depend only on $\n_1$, so that one may use Theorem~\ref{thm:condition} for $\n_2$ freely to show that $x_i$ not connected to $x_j$ in $\n_1+\n_2$ is unlikely. Condition {\bf C2} could work for the previous purpose. Nevertheless, Lemma~\ref{lem:pp} is not simpler to prove in this case, while a few additional technical difficulties arise in the proof of Proposition~\ref{prop:bulk intersection} stated below. We therefore prefer to work with $\g_1$ instead of $\n_1$.
\end{remark}

\begin{proof}
When the sum on the right is non-zero, the term is larger than or equal to $n$, which is clearly larger than the sum on the left. Now, assume that the sum on the right is zero. In this case, look at the connected component $\mathcal C$ of $x_1$ in $\n_1+\n_2$. If $\mathcal C$ contains all the sources, then the sum on the left is an alternating sum of $+1$ and $-1$, which is therefore zero. If $\mathcal C$ does not contain all the sources, remove one by one the arcs of the backbones between sources $x_i$ and $x_j$ that are not in $\mathcal C$. Since they are not connected to $x_1$, we know that for each $x_\ell$ in $\mathcal C$, either $1<\ell< i<j$ or $1<i<j<\ell$ (otherwise the sum on the right-hand side would not be zero), so that this operation does not alter the sum on the left. After removing all these arcs, we end up in the first case where $\mathcal C$ contains all the remaining sources of the current, so that the sum is zero by the previous argument.  
\end{proof}

Of course, contrarily to the planar case, the right-hand side of \eqref{eq:main_boundary_finite_range} has no reason to be zero. Nonetheless, we will show that the expectation of each indicator function on the right is small.
The idea is that while $x_i$ and $x_j$ are not necessarily connected in $\n_1+\n_2$, they create {\em coarse intersections}, i.e.~vertices $y\in\mathbb H$ such that $x_i$ and $x_j$ are connected to $\Lambda_R(y)$ in $\g_1$.
Now,  the odd part of the second current $\n_2$, which is independent of the first one, will create, by Theorem~\ref{thm:condition}, many circuits around each one of these coarse intersections. Finally, the even-positive part of the second current (note that even edges of the current $\n_2$ are even-positive with probability which is independent of everything so far) will have large probability of connecting the connected components of $x_i$ and $x_j$ to the same odd circuit of $\n_2$, and therefore connect $x_i$ to $x_j$ in $\n_1+\n_2$.

The previous reasoning works perfectly when the coarse intersections are far from the boundary, but is harder to implement for coarse intersections close to the boundary. In order to guarantee the existence of such coarse intersections, we will need to show that there are unlikely to remain (as the distance between the $x_i$ tend to infinity) all in the horizontal strip $\mathbb S_M:=\mathbb Z\times \llbracket 0, M\rrbracket$ of fixed height $M$ ($\bbS_M$ is the $M$-neighborhood of the boundary). This technical statement (Proposition~\ref{prop:bulk intersection}) will be slightly tricky to prove. We therefore choose to present the proof of Theorem~\ref{thm:Pf_finite_range} in two steps. In Section~\ref{sec:preliminary}, we present the proof of the theorem without proving Proposition~\ref{prop:bulk intersection} below, and in Section~\ref{sec:proof intersection}, the proof of the proposition.

\subsection{Proof of Theorem~\ref{thm:Pf_finite_range} using Proposition~\ref{prop:bulk intersection}}\label{sec:preliminary}

The set $X$ will always denote $\{x_1,\dots,x_{2n}\}$, and $d(X):=\min\{\|x_r-x_s\|,r\ne s\}$. We will write limits as $d(X)$ tends to infinity as limits which are uniform in the choice of $X$ provided that $d(X)$ tends to infinity.

Theorem~\ref{thm:Pf_finite_range} can be derived following exactly the proof of Theorem~\ref{thm:pf_boundary} line by line -- except that we use \eqref{eq:main_boundary_finite_range} instead of \eqref{eq:main_boundary} --  if we show that for each $i<j<k<\ell$,
\be \label{eq:21}
\lim_{d(X)\rightarrow \infty}{\bf P}_{\mathbb H}^{X,\emptyset}[\{x_i \stackrel{\g_1}{ \longleftrightarrow} x_k\}\cap\{x_j \stackrel{\g_1}{ \longleftrightarrow} x_\ell\}\cap\{x_i \stackrel{\n_1+\n_2}{\kern+4.5pt\arrownot\kern-4.5pt\longleftrightarrow} x_j\}]~=~ 0.\ee

For $K>0$ and a set $Y\subset \bbH$, let $E_{K,Y}$ be the event that there exist $K$ strongly disjoint odd circuits surrounding  some $y\in Y$, where  {\em strongly disjoint} means that any two such circuits remain at a distance $2R$ of each other. Also, for $i<j<k<\ell$, set
$$B=B(i,j,k,\ell):=\{x_i \stackrel{\g}{ \longleftrightarrow} x_k\}\cap\{x_j \stackrel{\g}{ \longleftrightarrow} x_\ell\}\cap\{x_i \stackrel{\n}{\kern+4.5pt\arrownot\kern-4.5pt\longleftrightarrow} x_j\}.$$
For $\n\in B$, let ${\bf Y}={\bf Y}(\n)$ be the set of coarse intersections of $\n$.

\medbreak
We proceed in three steps to show \eqref{eq:21}. The first one consists in saying that if there are many strongly disjoint circuits of odd edges in $\n_2$ surrounding the coarse intersections of $\n_1$, then the probability that the even-positive part of the current $\n_2$ does not connect $x_i$ to $x_j$ is small. The second step shows that the probability of having many such circuits in $\n_2$ is large when there is at least one coarse intersection in $\n_1$ which is far from the boundary. The last step claims that the probability that there is no coarse intersection in $\n_1$ which is far from the boundary is small.
\medbreak\noindent
{\bf Step 1.} There exists a constant $c=c(J,\beta,R)>0$ such that for every $K>0$, 
\be\label{eq:ffff}{\bf P}_{\mathbb H}^{X,\emptyset}[x_i \stackrel{\n_1+\n_2}{\kern+4.5pt\arrownot\kern-4.5pt\longleftrightarrow} x_j~|~\n_1\in B,\n_2\in E_{K,{\bf Y}(\n_1)}]\le (1-c)^K.
\ee
Indeed, the event on which we condition depends only on $\n_1$ and the odd part of $\n_2$. The pairs $\{x,y\}$ which are not in the odd part of $\n_2$ have positive (even) current $\n_2$ with probability $1-1/\cosh(\beta J_{x,y})$ independently of everything else. Therefore, there exists a constant $c=c(J,\beta,R)>0$ such that, conditioned on $\n_1$ and the odd part of $\n_2$ in the previous event, each of the $K$ strongly disjoint circuits of odd edges in $\n_2$ is connected in $\n_2$ using two boxes of size $R$ to both paths of ${\bf \Gamma}_1$ between $x_i$ and $x_k$, and $x_j$ and $x_\ell$, with probability larger than $c$.
Since the circuits are strongly disjoint, this happens independently for each circuit and we deduce \eqref{eq:ffff}.
\medbreak\noindent
{\bf Step 2.} For every $K>0$ and $\ep>0$, there exists $M=M(K,\ep)>0$ such that 
\be\label{eq:fffff} {\bf P}_{\mathbb H}^{X,\emptyset}[\n_2\notin E_{K,{\bf Y}(\n_1)}~|~\n_1\in B, {\bf Y}(\n_1)\not\subset \bbS_M]\le \ep.\ee
Indeed, note that the event on which we condition depends only on $\n_1$. The current $\n_2$ is therefore independent of the conditioning. Now, Theorem~\ref{thm:condition} together with the convergence of measures ${\bf P}_G^\emptyset$ to the infinite-volume measure ${\bf P}_{\mathbb Z^2}^\emptyset$ implies that for $M=M(K,\varepsilon)$ large enough, we have that for any $x\notin \mathbb S_M$,
\be\label{eq:condition3}{\bf P}_{\mathbb H}^\emptyset[E_{K,\{x\}}]\ge 1-\varepsilon,\ee
This immediately implies \eqref{eq:fffff}.
\medbreak\noindent
{\bf Step 3.} Fix $\ep>0$, choose $K=K(\ep)>0$ such that $(1-c)^K\le \ep$, where $c$ is given by Step 1. Then, set $M=M(K,\ep)$ defined in Step 2. We have that 
$${\bf P}_{\mathbb H}^{X,\emptyset}[\{x_i \stackrel{\g_1}{ \longleftrightarrow} x_k\}\cap\{x_j \stackrel{\g_1}{ \longleftrightarrow} x_\ell\}\cap\{x_i \stackrel{\n_1+\n_2}{\kern+4.5pt\arrownot\kern-4.5pt\longleftrightarrow} x_j\}]\le 2\ep+{\bf P}_{\mathbb H}^X[B\cap \{{\bf Y}(\n)\subset \bbS_M\}].$$
 Therefore, the theorem follows from the next proposition.
 \begin{proposition}\label{prop:bulk intersection}
Fix $i<j<k<\ell$. For any $M>0$,
\be\nonumber
\lim_{d(X)\rightarrow \infty}{\bf P}_{\mathbb H}^X\left[B\cap\{{\bf Y}(\n)\subset\bbS_M\}\right]~=~ 0.
\ee
\end{proposition}

\begin{remark}
Proposition~\ref{prop:bulk intersection} seems like a technical step that could probably be removed if one could prove that a point on the boundary is typically surrounded by odd paths in $\n_2$ going from boundary to boundary instead of circuits (note that these paths are necessarily part of circuits by the source constraint on $\n_2$). Interestingly, the predicted scaling limits of long odd paths in $\n_2$ is such that this is not conjectured to happen. More precisely, the SLE(3) and CLE(3) processes which are describing the scaling limit of odd paths in $\n_2$ for the nearest-neighbor model (see \cite{BenHon16,CheDumHonKemSmi14}) are such that no macroscopic loop intersects the boundary. The predicted universality would imply that this also occurs in our context. This means that, without Proposition~\ref{prop:bulk intersection}, one would need to use the {\em combined} fractality of $\n_1$ and $\n_2$ to prove that avoided intersections are unlikely, which would make the proof much more complicated.  
\end{remark}

\subsection{Proof of Proposition~\ref{prop:bulk intersection}}\label{sec:proof intersection}

We start this section by some properties of the Ising model that will be useful in the proof (more precisely, a few classical properties of backbones and a classical inequality on spin-spin correlations, called the Messager-Miracle-Sol\'e inequality). After this, we present the proof of Proposition~\ref{prop:bulk intersection}.

\subsubsection{Preliminaries}
\paragraph{Two useful properties of backbones.} Consider a collection $\gamma$ of paths, and $\partial \gamma$ the sources of the subgraph obtained by the union of these paths. Let us introduce the {\em weight} of $\gamma$ as
$$\rho_G(\gamma)=\rho_G(\beta,J,\gamma):=\frac{1}{Z(G,\beta)}\sum_{\n:\partial\n=\gamma}w(\n)\mathbb{I}[\Gamma(\n)=\gamma],$$
where the sum is over all currents $\n$ in $G$. While this definition works only for finite graphs, one can easily take the limit as $G$ tends to $\bbH$ to obtain the definition of weights in the case of $G=\bbH$. With this notation, one gets that for any $A\subset V(G)$,
\be\label{eq:cu}\displaystyle \langle \prod_{x\in A}\sigma_x\rangle_{G,\beta}=\sum_{\partial\gamma=A}\rho_G(\gamma).\ee
Also, note that if the backbone $\gamma$ is the concatenation of two backbones $\gamma_1$ and $\gamma_2$ (this is denoted by $\gamma=\gamma_1\circ\gamma_2$), then
\be\label{eq:cuu}\rho_G(\gamma)=\rho_G(\gamma_1)\rho_{G\setminus \overline\gamma_1}(\gamma_2),\ee
where $\overline\gamma_1$ is the set of $\{x,y\}\in E(G)$ whose $\n_1$-parity is determined by the fact that $\gamma_1$ is the beginning of the backbone (this includes $\{x,y\}$ with both endpoints in $\gamma_1$ together with some pairs $\{x,y\}$ with only one vertex in $\gamma_1$ for which $\n_1(x,y)$ must be even), and $G\setminus \overline\gamma_1$ denotes the original graph $G$ with coupling constants $J_{x,y}$ set to 0 for each $\{x,y\}\in \overline\gamma_1$ (this is a slight abuse of notation, but we believe it will not lead to any confusion later on). In particular, for any fixed $\gamma_1$ from $x$ to $y$, summing over every $\gamma_2$ from $y$ to $z$ in $G\setminus\overline\gamma_1$ and applying  Griffiths' inequality \cite{Gri67} gives that
\begin{align}\label{eq:ki}\sum_{\gamma_2}\rho_{\mathbb H}(\gamma_1\circ\gamma_2)&=\rho_{\mathbb H}(\gamma_1)\sum_{\gamma_2}\rho_{\mathbb H\setminus \overline{\gamma_1}}(\gamma_2)=\rho_{\mathbb H}(\gamma_1)\langle\sigma_{y}\sigma_{z}\rangle_{\bbH\setminus\overline{\gamma_1}}\le\rho_{\mathbb H}(\gamma_1)\langle\sigma_{y}\sigma_z\rangle_{\bbH}.\end{align}
\paragraph{Monotonicity of spin-spin correlations.} The following inequality generalizes the classical Messager-Miracle-Sol\'e inequality (see \cite{Heg77,MesMir77,Sch77}) to the upper half-plane and to finite-range models: there exists some constant $ c_1 \in (0,\infty ) $ such that for all $ x$ and $y$ with the same second coordinate and first coordinates satisfying $ 0\le x_1 \leq y_1 $, we have
\be
 \langle\sigma_0\sigma_y\rangle_{\bbH}\le c_1\langle\sigma_0\sigma_x\rangle_{\bbH}.
\ee 
In case $ x_1$ and $ y_1 $ have the same parity, this is an immediate consequence of Theorem~\ref{thm:mono} in the Appendix. The general case follows using Griffiths' inequality as in~\eqref{eq:Griffiths} below.

The previous inequality will in fact be used under the following form. For fixed $K>0$, there exists $c_0=c_0(K,J,\beta)<\infty$  such that for every $x=(x_1,0)$, every $z\in \Lambda_K(y)$ such that $y=(y_1,0)$ satisfies $0\le x_1\le y_1$, we have   
\be\label{eq:lower bound coarse}
\langle\sigma_{0}\sigma_{z}\rangle_{\mathbb H}\le c_0\langle\sigma_{0}\sigma_{x}\rangle_{\mathbb H}.
\ee
For a proof, we simply apply Griffiths' inequality to obtain
\be\label{eq:Griffiths}
\langle\sigma_{0}\sigma_{y}\rangle_{\mathbb H}\ge \langle\sigma_{0}\sigma_{z}\rangle_{\mathbb H}\langle\sigma_{z}\sigma_{y}\rangle_{\mathbb H}
\ee
and set $c_0:=c_1\max\{\langle\sigma_{y}\sigma_{z}\rangle_{\mathbb H}^{-1}:z\in\Lambda_K(y)\}$>0.

\begin{remark}
The above inequality is the only place where we use that we are in the upper half-plane. If one could find an alternative proof of Proposition~\ref{prop:bulk intersection} which is not based on this inequality, we would obtain a result valid for any planar domain $G$ and not only the upper half-plane.
\end{remark}

\subsubsection{Proof of Proposition~\ref{prop:bulk intersection} for $n=2$}

We first restrict our attention to $n=2$. In this case, $X=\{x_1,x_2,x_3,x_4\}$ and $(i,j,k,\ell)=(1,2,3,4)$.
Let $\Gamma_{13}=\Gamma_{13}(\n)$ be the path in the backbone $\g(\n)=\g$ of $\n$ from $x_1$ to $x_3$, and $\Gamma_{24}=\Gamma_{24}(\n)$ the one from $x_2$ to $x_4$. Note that $\g$ is the union of $\Gamma_{13}$ and $\Gamma_{24}$ and that they are well-defined on the event that $x_1$ and $x_4$ are not connected. Also, in such case, we may construct them in the order we wish, since they depend on distinct connected components of $\widehat\n$.

We start by excluding the case in which $\Gamma_{13}$ comes close to $x_2$. A similar statement holds for the case in which $\Gamma_{24}$ comes close to $x_3$.
\bigbreak
\noindent{\em {\bf Claim 1.} Fix $K>0$. There exists $C=C(K)>0$ such that for any $x_1,x_2,x_3,x_4$, \begin{align}\label{eq:cc}
{\bf P}_{\mathbb H}^X[B\cap\{\Gamma_{13}\cap \Lambda_K(x_2)\ne \emptyset \}]&\le C\langle\sigma_{x_2}\sigma_{x_3}\rangle_\bbH.
\end{align}}
\begin{proof}[Proof of Claim 1]
By decomposing with respect to possible values for $\Gamma_{13}$, we find
\begin{align}\nonumber
{\bf P}_{\mathbb H}^X[B\cap&\{\Gamma_{13}\cap \Lambda_K(x_2)\ne \emptyset \}]\le \frac1{\langle\sigma_{x_1}\sigma_{x_2}\sigma_{x_3}\sigma_{x_4}\rangle_{\mathbb H}}\sum_{x\in\Lambda_K(x_2)}\sum_{\gamma\ni x}\rho_{\mathbb H}(\gamma)\\
&\le \frac1{\langle\sigma_{x_1}\sigma_{x_2}\sigma_{x_3}\sigma_{x_4}\rangle_{\mathbb H}}\sum_{x\in\Lambda_K(x_2)}\sum_{\gamma_1,\gamma_2,\gamma_3}\rho_{\mathbb H}(\gamma_1\circ\gamma_2\circ\gamma_3),
\end{align}
where $\gamma_1$ is the part of the backbone from $x_1$ to $x$, $\gamma_2$ the one from $x$ to $x_3$, and $\gamma_3$ the one from $x_2$ to $x_4$. Using \eqref{eq:ki} when summing over $\gamma_3$, we deduce that
\begin{align}\nonumber
{\bf P}_{\mathbb H}^X[B\cap\{\Gamma_{13}\cap \Lambda_K(x_2)\ne \emptyset \}]
&\le\frac{\langle\sigma_{x_2}\sigma_{x_4}\rangle_\mathbb H}{\langle\sigma_{x_1}\sigma_{x_2}\sigma_{x_3}\sigma_{x_4}\rangle_{\mathbb H}} \sum_{x\in\Lambda_K(x_2)}\sum_{\gamma_1,\gamma_2}\rho_{\mathbb H}(\gamma_1\circ\gamma_2),
\end{align}
Using \eqref{eq:ki} when summing over $\gamma_2$ and then $\gamma_1$, we deduce that
\begin{align}\nonumber
{\bf P}_{\mathbb H}^X[B\cap\{\Gamma_{13}\cap \Lambda_K(x_2)\ne \emptyset \}]
&\le\frac{\langle\sigma_{x_2}\sigma_{x_4}\rangle_\mathbb H}{\langle\sigma_{x_1}\sigma_{x_2}\sigma_{x_3}\sigma_{x_4}\rangle_{\mathbb H}} \sum_{x\in\Lambda_K(x_2)}\langle\sigma_{x_1}\sigma_{x}\rangle_{\mathbb H}\langle\sigma_{x}\sigma_{x_3}\rangle_{\mathbb H}.
\end{align}
Using \eqref{eq:lower bound coarse} in the first two inequalities, and finally Griffiths' inequality \cite{Gri67} in the last one, we deduce that
\begin{align*}\nonumber
{\bf P}_{\mathbb H}^X[B\cap\{\Gamma_{13}\cap \Lambda_K(x_2)\ne \emptyset \}]
&\le\frac{c_0^2\langle\sigma_{x_2}\sigma_{x_4}\rangle_\mathbb H}{\langle\sigma_{x_1}\sigma_{x_2}\sigma_{x_3}\sigma_{x_4}\rangle_{\mathbb H}} |\Lambda_K(x_2)|\langle\sigma_{x_1}\sigma_{x_2}\rangle_{\mathbb H}\langle\sigma_{x_2}\sigma_{x_3}\rangle_{\mathbb H}\nonumber\\
&\le\frac{c_0^3\langle\sigma_{x_3}\sigma_{x_4}\rangle_\mathbb H}{\langle\sigma_{x_1}\sigma_{x_2}\sigma_{x_3}\sigma_{x_4}\rangle_{\mathbb H}} |\Lambda_K(x_2)|\langle\sigma_{x_1}\sigma_{x_2}\rangle_{\mathbb H}\langle\sigma_{x_2}\sigma_{x_3}\rangle_{\mathbb H}\nonumber\\
&\le
c_0^3 |\Lambda_K(x_2)|\langle\sigma_{x_2}\sigma_{x_3}\rangle_{\mathbb H}\le C\langle\sigma_{x_2}\sigma_{x_3}\rangle_{\mathbb H}.
\end{align*}
\end{proof}
Let $\gamma=(\gamma(0),\dots,\gamma(k))$ be a path in the upper half-plane. We say that $\gamma[r,s]$ (with $s> r+1$) is a {\em flight (at height $M$)} of $\gamma$ if $\gamma(r),\gamma(s)\in \mathbb S_M$ and $\gamma(u)\notin \mathbb S_M$ for $r<u<s$. We say that a flight $\gamma[r,s]$ is {\em above }$x_2$ (resp. $\{x_2,x_3\}$) if, among $\gamma(s)$ and $\gamma(r)$, one is on the left of $x_2+(0,1)\mathbb R$ and the other between $x_2+(0,1)\mathbb R$ and $x_3+(0,1)\mathbb R$ (resp.~on the right of $x_3+(0,1)\mathbb R$); see Fig.~\ref{fig:S_M}. Note that a flight above $\{x_2,x_3\}$ is {\em not} a flight above $x_2$.
\bigbreak
\noindent{\em {\bf Claim 2.}
$\displaystyle
\lim_{d(X)\rightarrow\infty}{\bf P}_{\mathbb H}^X[B\cap\{{\bf Y}(\n)\subset\bbS_M\}\cap\{{\Gamma_{13}}\text{ contains a flight above }x_2\}]=0.
$
}
\begin{proof}[Proof of Claim 2]
Define $\mathcal V(\gamma):=\Lambda_M(\gamma(r))\cup\Lambda_M(\gamma(s))$, where $\gamma[r,s]$ is the first flight above $x_2$ (see Fig.~\ref{fig:S_M} for an illustration). Note that for $B\cap\{{\bf Y}(\n)\subset\bbS_M\}$ to occur, $\Gamma_{24}$ must intersect $\mathcal V(\Gamma_{13})$ since it cannot come within a distance $R$ of $\gamma[r,s]$ (other otherwise coarse intersection outside of $\mathbb S_M$), but must go from $x_2$ to $x_4$.

Therefore, if we decompose into possible values $\gamma_1$ and $\gamma_2$ for $\Gamma_{13}$ and $\Gamma_{24}$ compatible with the occurrence of the event under consideration, we find, by integrating successively the different parts of the backbone,
\begin{align}
{\bf P}_{\mathbb H}^X[B\cap\{{\bf Y}(\n)\subset\bbS_M\}\cap\{\Gamma_{13}\text{ contains }&\text{a flight above }x_2\}\cap\{\Gamma_{13}\cap \Lambda_K(x_2)= \emptyset\}]\nonumber \\
&\le  \frac1{\langle\sigma_{x_1}\sigma_{x_2}\sigma_{x_3}\sigma_{x_4}\rangle_{\mathbb H}}\sum_{\gamma_1}\rho(\gamma_1)\sum_{x\in \mathcal V(\gamma_1)}\sum_{\gamma_2\ni x}\rho_{\mathbb H\setminus \overline\gamma_1}(\gamma_2)\nonumber\\
&\le  \frac1{\langle\sigma_{x_1}\sigma_{x_2}\sigma_{x_3}\sigma_{x_4}\rangle_{\mathbb H}}\sum_{\gamma_1}\rho(\gamma_1)\sum_{x\in \mathcal V(\gamma_1)}\langle\sigma_{x_2}\sigma_x\rangle_\mathbb H\langle\sigma_{x}\sigma_{x_4}\rangle_\mathbb H\nonumber \\
&\le  \frac{2|\Lambda_M|\cdot \ep(K)\cdot c_0\langle\sigma_{x_3}\sigma_{x_4}\rangle_{\mathbb H}}{\langle\sigma_{x_1}\sigma_{x_2}\sigma_{x_3}\sigma_{x_4}\rangle_{\mathbb H}}\sum_{\gamma_1}\rho(\gamma_1)\nonumber \\
&= \frac{2|\Lambda_M|\cdot \ep(K)\cdot c_0\langle\sigma_{x_3}\sigma_{x_4}\rangle_{\mathbb H}\langle\sigma_{x_1}\sigma_{x_3}\rangle_{\mathbb H}}{\langle\sigma_{x_1}\sigma_{x_2}\sigma_{x_3}\sigma_{x_4}\rangle_{\mathbb H}}\nonumber\\
&\le  c_1\ep(K)\frac{\langle\sigma_{x_3}\sigma_{x_4}\rangle_{\mathbb H}\langle\sigma_{x_1}\sigma_{x_2}\rangle_{\mathbb H}}{\langle\sigma_{x_1}\sigma_{x_2}\sigma_{x_3}\sigma_{x_4}\rangle_{\mathbb H}}\le c_1\ep(K),\label{eq:cccc}
\end{align}
where $\ep(K):=\max_{y\notin \Lambda_{K-M}}\langle\sigma_0\sigma_y\rangle_{\bbH}$ and $c_1=c_1(M)>0$.  In the third inequality, we used \eqref{eq:lower bound coarse} and the fact that,  {\em by definition of a flight above $x_2$}, $\mathcal V(\Gamma_{13})$ is the union of two boxes of size $M$ whose centers are on the left of $x_3$. In the last inequality, we also used Griffiths' inequality one more time. 

Now, fix $\ep>0$. Since we assumed that $\langle\sigma_0\sigma_x\rangle_\bbH$ tends to zero, we can pick $K=K(M)$ large enough that $c_1\ep(K)$ is smaller than $\ep$. The first claim enables us to pick $d(X)$ large enough that ${\bf P}_{\mathbb H}^X[B\cap\{\Gamma_{13}\cap \Lambda_K(x_2)\ne \emptyset\}]\le \ep$. This concludes the proof since $\ep$ is chosen arbitrarily.
\end{proof}
\begin{figure}[h]
\begin{center}
\includegraphics [width = .9 \textwidth]{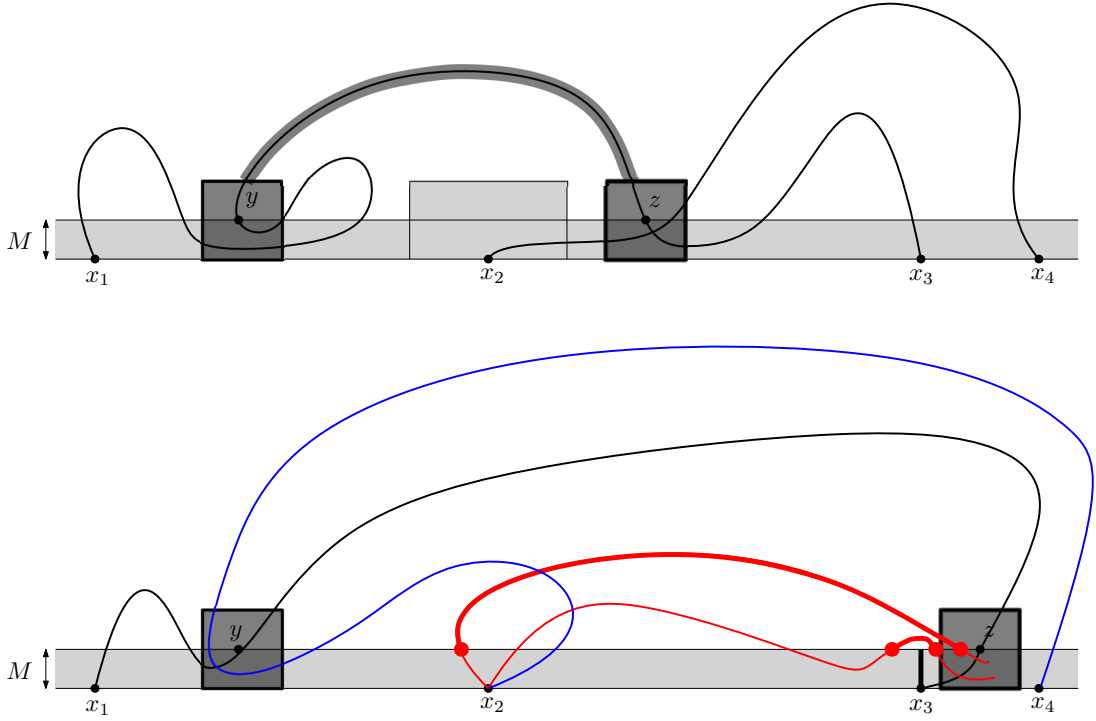}
\caption{On top, the set $\mathcal V(\gamma)$ (the two dark gray boxes) for a flight above $x_2$. By definition, $x_2$ is disconnected from $x_4$ by the union of $\mathcal V(\gamma)$ and the $R$-neighborhood of the flight. We also depicted the box of size $K$ around $x_2$. Below, we drew the set $\mathcal V(\gamma)$. The curve in blue is $\Gamma_{24}$. The two curves in red are possible paths exiting through the right side. Both contain either a flight above $x_3$, or a flight above $\{x_2,x_3\}$ (simply look at one flight of $\Gamma_{13}$ arriving on the right -- we highlighted such an example -- which is of either kind).}  \label{fig:S_M}
\end{center}
\end{figure}

We conclude the proof of Proposition~\ref{prop:bulk intersection}. Fix $M,K>0$ and consider the event that\begin{itemize}[noitemsep,nolistsep]
\item $B\cap\{{\bf Y}(\n)\subset\bbS_M\}$ occurs,
\item $\Gamma_{13}\cap \Lambda_K(x_2)=\emptyset$,
\item  $\Gamma_{24}\cap \Lambda_K(x_3)=\emptyset$,
\item $\Gamma_{13}$ contains no flight of height $M$ above $x_2$,
\item $\Gamma_{24}$ contains no flight of height $M$ above $x_3$.
\end{itemize}
On this event, $\Gamma_{13}$ and $\Gamma_{24}$ must have flights above $\{x_2,x_3\}$, since they cannot cross (recall that $\n\in B$). This enables us to define an inner most flight above $\{x_2,x_3\}$ in $\Gamma_{13}\cup\Gamma_{24}$. Let $B_{13}$ (respectively $B_{24}$) be the intersection of all the events above with the fact that the innermost flight in $\Gamma_{13}\cup\Gamma_{24}$ is in $\Gamma_{13}$ (respectively in $\Gamma_{24}$).
We wish to prove that for $K=K(M)$ large enough,
\begin{equation}{\bf P}_\mathbb H^X[B_{13}]\le \varepsilon\label{eq:bf}\end{equation}
when $d(X)$ is large enough.

We claim that this would conclude the proof. Indeed, on the one hand, one could prove in a similar fashion that ${\bf P}_\mathbb H^X[B_{24}]\le \ep$ provided $d(X$) is large enough. On the other hand, the previous two claims, which apply to $\Gamma_{13}$ as well as to $\Gamma_{24}$, imply that $${\bf P}_\mathbb H^X[B\cap\{{\bf Y}(\n)\subset\bbS_M\}\setminus(B_{13}\cup B_{24})]\le 4\ep$$ provided $d(X)$ is large enough, where the factor 4 comes from the fact that we bound the probability of the intersection of $B$ with the complement of each event in the four last bullets above by $\ep$ (for the second and third ones, we also use that by assumption, $\langle\sigma_0\sigma_x\rangle_{\bbH}$ tends to 0). In conclusion, we would get that for $d(X)$ large enough,
$${\bf P}_\mathbb H^X[B\cap\{{\bf Y}(\n)\subset\bbS_M\}]\le{\bf P}_\mathbb H^X[B\cap\{{\bf Y}(\n)\subset\bbS_M\}\setminus(B_{13}\cup B_{24})]+{\bf P}_\mathbb H^X[B_{13}]+{\bf P}_\mathbb H^X[B_{24}]\le 6\ep,$$ an inequality which finishes the proof since $\ep$ was chosen arbitrarily.

Let us turn to the proof of \eqref{eq:bf}. Define $\mathcal V_{\rm left}(\Gamma_{13})$ to be the box of size $M$ around the extremity of the innermost flight above $\{x_2,x_3\}$ which is on the left of $x_2$. In the event $B_{13}$, $\Gamma_{24}$ must intersect $\mathcal V_{\rm left}(\Gamma_{13})$. Indeed, $\Gamma_{24}$ is leaving from $x_2$, which is disconnected from $x_4$ by the union of the innermost flight above $\{x_2,x_3\}$ with the balls of size $M$ around its extremities. Since $\Gamma_{24}$ cannot come within a distance $R$ from the flight itself, if it does not intersect $\mathcal V_{\rm left}(\Gamma_{13})$, then it must intersect $\mathbb S_M$ near the extremity on the right. Yet, this would imply that either $\Gamma_{24}$ contains a flight above $x_3$ or a flight above $\{x_2,x_3\}$ which is inside the innermost flight above $\{x_2,x_3\}$ of $\Gamma_{13}$ (see Fig.~\ref{fig:S_M}), two facts which are excluded since we are in $B_{13}$.

Now, since $\mathcal V_{\rm left}(\Gamma_{13})$ is at a distance at least $K-M$ of $x_2$, the same reasoning as for \eqref{eq:cccc} applies and \eqref{eq:bf} follows by choosing $K=K(M)$ large enough. This concludes the proof of Proposition~\ref{prop:bulk intersection} for $n=2$.

\subsubsection{Proof of Proposition~\ref{prop:bulk intersection} for general $n$}
Fix $X=\{x_1,\dots,x_{2n}\}$. The proof works exactly in the same way as for $n=2$ except that we need to decompose with respect to the different possible backbones, which now pair all $\{x_1,\dots,x_{2n}\}$.

Let $\pi$ be some pairing and $E_\pi$ be the event that $x_{\pi(2r-1)}$ is paired with $x_{\pi(2r)}$ by the backbone for every $r\le n$. Also, let $E_{\pi,u,v}\subset E_\pi$ be the event that the vertices $x_r$ in the connected component of $x_i$ are explored by the backbone in a certain order $u$, and those in the connected component of $x_j$ in a certain order $v$.

Then, we may prove that for any $\pi$ pairing $x_i$ with $x_k$ and $x_j$ with $x_\ell$,
$${\bf P}_\mathbb H^X[B\cap \{{\bf Y}(\n)\subset\bbS_M\}\cap E_{\pi,u,v}]\le \varepsilon\cdot \frac{\prod_{j=1}^n\langle\sigma_{x_{\pi(2j-1)}}\sigma_{x_{\pi(2j)}}\rangle_\mathbb H}{\langle\sigma_{x_1}\dots\sigma_{x_{2n}}\rangle_\mathbb H}\le \varepsilon.$$
The proof follows exactly the same lines as before, except that we work in the plane minus the paths of the backbone $\gamma_1,\dots,\gamma_{p-1}$ and $\tilde\gamma_1,\dots,\tilde\gamma_{q-1}$ corresponding to vertices in the connected components of $\{x_i,x_k\}$ and $\{x_j,x_\ell\}$ already explored before exploring the two paths from $x_i$ to $x_k$ and from $x_j$ to $x_\ell$.

One may be worried that  \eqref{eq:lower bound coarse} cannot be used anymore since it relied on the Messager-Miracle-Sol\'e inequality (and therefore on invariance under translations of coupling constants). Nonetheless, the inequality can still be used since anyway we bound the quantities by what happens in the half-plane, meaning that we may first use Griffiths inequality \cite{Gri67} to set back the coupling constants to the original ones (we delete the effect of depletion by the first backbones), and then use \eqref{eq:lower bound coarse}. Also, note that in order for the previous claims to be true, we used the fact that we may condition on $\Gamma_{13}$ before $\Gamma_{24}$, or vice-versa, and obtain the same paths. This is not valid in general (in a connected component, the backbone may depend heavily on the order in which sources are discovered), nevertheless, the claim is still true here since $x_i$ and $x_j$ are not in the same connected component of $\n$.

The claim follows by summing on all possible pairings mapping $x_i$ to $x_k$ and $x_j$ to $x_\ell$, and all possible ordering $u$ and $v$.

\section{Nested circuits of odd current, an existence result}\label{sec:5}

In this section, we borrow the notation from Section~\ref{sec:4.1} (in particular, the definitions of paths, continuous paths, and circuits surrounding a set). 

\subsection{A refined version of Theorem~\ref{thm:condition}}
This section contains the proof of the following statement, which clearly implies Theorem~\ref{thm:condition}.

\begin{theorem}\label{prop:condition}
Let $E$ be the event that there exist infinitely many disjoint circuits surrounding 0. Then, for any $\beta\ge\beta_c$, we have
\be
{\rm P}^\emptyset_{\mathbb Z^2,\beta}[\eta\in E]={\bf P}_{\mathbb Z^2,\beta}^\emptyset[\widehat\n\in E]=\phi_{\mathbb Z^2,\beta
}[\omega\in E]=1.
\ee
\end{theorem}

 This theorem  implies Theorem~\ref{thm:condition} since a circuit in $\eta$ is an odd circuit in $\n$. Note that the statement of Theorem~\ref{prop:condition} is wrong for $\beta<\beta_c$ since the connectivity probabilities in $\phi_{\mathbb Z^2}$ decay exponentially fast (see \cite{AizBarFer87,DumTas15}).
 
As in the proof of the previous theorem, the main discussion focuses on the case  $\beta=\beta_c$ and $\phi_{\mathbb Z^2,\beta_c}[0\leftrightarrow \infty]=0$.   The latter condition is expected to hold for any finite-range Ising model at its critical point (for nearest-neighbor models, that is proven in \cite{AizDumSid15,DumMar17,Yan52}). \\

The proof of Theorem~\ref{prop:condition} is organized as follows. The proof in the case $\beta=\beta_c$ and $\phi_{\mathbb Z^2,\beta_c}[0\leftrightarrow \infty]=0$ is presented in Section~\ref{sec:outline}, except that the proofs of Lemmas~\ref{lem:P2} and \ref{lem:P3} are postponed to Sections~\ref{sec:P2} and \ref{sec:P3} respectively. 
 The proof in the case $\phi_{\mathbb Z^2,\beta}[0\leftrightarrow \infty]>0$ (which covers the range $\beta>\beta_c$ and possibly $\beta=\beta_c$ in the (unexpected) case where the phase transition would be discontinuous) is postponed to Section~\ref{sec:proof-case-phiPositive}.

Before diving into the proof, we attract the attention of the reader on the fact that for $\beta>\beta_c$, $\omega$ percolates in a very strong sense (finite connected components have exponential tails, see \cite{ManRao16}) and therefore the statement is far from optimal for $\omega$. For $\widehat\n$, the result is also expected to be suboptimal for $\beta>\beta_c$, since the following conjecture should be true.
\begin{conjecture}
For $\beta>\beta_c$, we have that ${\bf P}_{\mathbb Z^2,\beta}^\emptyset[0\longleftrightarrow \infty\text{ in }\widehat\n]>0
$.
\end{conjecture}
The previous conjecture should be compared to the following known result: $\beta_c$ corresponds to the phase transition for the percolation of the sum of two independent currents, i.e.~for $\n_1+\n_2$ with law ${\bf P}_{\mathbb Z^2,\beta}^{\emptyset,\emptyset}$.
The conjecture that it also corresponds to the phase transition for a single current is interesting from the point of view of phase transition, since it means that $\beta_c$ probably corresponds to a critical point for currents above which loops unfold into infinite loops.

\begin{remark}We wish to highlight the fact that ${\rm P}_{\mathbb Z^2,\beta}^\emptyset[0\longleftrightarrow \infty\text{ in }\eta]>0$ is not necessarily true, as can be seen for the hexagonal lattice in which $\eta$ corresponds to the boundary walls in the high-temperature Ising model on the dual triangular lattice, a model for which it is possible to show that all boundary walls are finite for any value of $\beta$. Nevertheless, in this case the edges with even (and positive) current are sufficient to create an infinite cluster and $\widehat\n$ percolates as soon as $\beta>\beta_c$. We actually believe that the conjecture can be proved with standard techniques in the planar case, but we do not know how to prove it for finite-range interactions, or for larger dimensions.\end{remark}

\subsection{Proof of Theorem~\ref{prop:condition} when $\beta=\beta_c$ and $\phi_{\mathbb Z^2,\beta_c}[0\leftrightarrow \infty]=0$}\label{sec:outline}

From now on, we fix $\beta=\beta_c$ and remove $\beta$ from the notation. 
The proof is based on four important lemmas. We start by a key observation, which is based on the switching principle for even subgraphs of a fixed graph.
\begin{lemma}\label{lem:flowArgument} Assume that $\phi_{\mathbb Z^2}[0\leftrightarrow \infty]=0$, then
\begin{enumerate}
\item If $\phi_{\mathbb Z^2}[\omega\in E]=1$, then ${\rm P}^\emptyset_{\mathbb Z^2}[\eta\in E]=1$,
\item If ${\bf P}_{\mathbb Z^2}^{\emptyset,\emptyset}[\widehat{\n_1+\n_2}\in E]=1$, then ${\bf P}_{\mathbb Z^2}^\emptyset[\widehat{\n}\in E]=1$.\end{enumerate} \end{lemma}
\begin{proof}
  We begin with the proof of the first item. Assume that $\omega$ contains infinitely many disjoint circuits surrounding $0$. Since there is no infinite connected component in $\omega$, there are infinitely many connected components surrounding $\Lambda_R$ but not intersecting it.

For each such connected component $\calC$, choose a cycle $c\subset\calS$ surrounding $0$. Note that $\eta\mapsto \eta\Delta c$ maps\footnote{Observe that if $\eta$ does not contain a cycle surrounding the origin, the flows of $\eta$ and $\eta\Delta c$ through a fixed set of edges $E\subset\calC$ such that $\calC\setminus E$ does not contain any circuit surrounding $0$  changes from even to odd, which guarantees the existence of a cycle in $\eta\Delta c$ surrounding $0$.} even subgraphs without cycle around $0$ to subgraphs with a cycle around $0$. Since $\eta$ and $\eta\Delta c$ have the same law (Corollary~\ref{rmk:cycle}), we deduce that, conditioned on $\omega$, the probability that a connected component surrounding $0$ but not intersecting it contains a cycle in $\eta$ surrounding $0$ is larger than or equal to 1/2.

Since this is true independently for any connected component of $\omega$ surrounding $0$ and not intersecting it, we deduce that $\phi_{\mathbb Z^2}[\omega\in E]=1$ implies ${\rm P}_{\mathbb Z^2}^\emptyset[\eta\in E]=1$.

We now prove the second item using the same technique as for the first item, but with the graph representation of currents instead of the random-cluster model. Let $\n_1$ and $\n_2$ be two independent random currents. Let $\mathcal M$ be the multigraph with vertex-set $\mathbb Z^2$, obtained by putting exactly $\n_1(x,y)+\n_2(x,y)$ edges between $x$ and $y$, for every $x,y\in\bbZ^2$. Then, let $\mathcal N$ be chosen uniformly among all the even subgraphs of $\cal M$ with $\n_2(x,y)$ edges between $x$ and $y$. Conditionally on $\mathcal M$, the law of $\mathcal N$ is simply the law of a random even subgraph of $\mathcal M$. 

With this observation, one can conclude that Item 2 holds, exactly as we did for Item 1. Indeed, if $\widehat{\n_1+\n_2}\in E$, then the graph $\mathcal M$ associated with $\n_1+\n_2$ has infinitely many connected components surrounding $0$ since the assumption that $\phi_{\bbZ^2}[0\leftrightarrow\infty]=0$ classically implies that $\widehat{\n_1+\n_2}$ has no infinite connected component (since the phase transition for Ising corresponds to the percolation phase transition for the sum of two currents, see \cite{Aiz82}).  By the flow argument mentioned above, each such connected component contains a circuit of $\mathcal N$ surrounding $0$ with probability larger than or equal to 1/2, independently for different connected components. Therefore, $\widehat\n_2$ contains infinitely many circuits surrounding the origin.
\end{proof}

\begin{remark}As mentioned above, the proofs of Items 1 and 2 both rely on switching circuits (once with Corollary~\ref{rmk:cycle}, and the other with the switching principle). Notice the strength of the switching principle: if many circuits surrounding the origin exist in $\widehat{\n_1+\n_2}$, then many circuits already exist in $\widehat\n_1$ and $\widehat\n_2$.\end{remark}

A direct consequence of this lemma is that it suffices to show that $\omega\in E$ almost surely, or that $\widehat{\n_1+\n_2}\in E$ almost surely.
The next lemma provides a sufficient condition to show the former.

\begin{definition}\label{def:crossed} A rectangle $R:=[a,b]\times[c,d]$ is said to be {\em crossed horizontally in $f$} if there exists a {\em horizontal crossing in $f$}, i.e.~an open (in $f$) continuous path staying in $R$ and going from the left side $\{a\}\times[c,d]$ to the right side $\{b\}\times[c,d]$ of the rectangle. Equivalently, $R$ is said to be {\em crossed vertically in $f$} if there exists a {\em vertical crossing in $f$}.
\end{definition}

\begin{lemma}\label{lem:P1}
Assume that $\phi_{\mathbb Z^2}[0\leftrightarrow \infty]=0$. Then, if 
\begin{equation}
 \label{eq:P1}\tag{{\bf P1}} \limsup_{n\ge0} \phi_{\mathbb Z^2}[\Lambda_n\text{ is crossed horizontally in }\omega]>0,
\end{equation}
then $\phi_{\mathbb Z^2}[\omega\in E]=1$.
\end{lemma}
In spirit, the proof of this lemma consists in constructing circuits surrounding the origin in $\Lambda_n$ using horizontal and vertical crossings in $\Lambda_n$. We refer to Fig.~\ref{fig:A7} for an illustration of the method in the planar case. Some additional care is required in our context due to the fact that we do not work with a planar model since we have finite-range interactions. Nevertheless, the work has been (almost) entirely done for us: a {\em theory for ``gluing paths''} was developed in \cite{DumSidTas16} and later in \cite{ManRao16,NewTasWu15}, and our context will only require a mild variant of it.

\begin{proof}Consider the event $H_n^+$ that there exists a horizontal crossing of $\Lambda_n$ that passes above the origin. By symmetry, the probability of $H_n^+$ is equal to the probability of the event $H_n^-$ that there exists a horizontal crossing in $\Lambda_n$ that passes below the origin. Furthermore, if $\Lambda_n$ is crossed from left to right, either $H_n^+$ or $H_n^-$ must occur. Therefore,
\begin{equation}
  \label{eq:17}
  \phi_{\bbZ^2}[\Lambda_n \text{ is crossed horizontally in $\omega$}]\le\phi_{\bbZ^2}[H_n^-]+  \phi_{\bbZ^2}[H_n^+]=2 \phi_{\bbZ^2}[H_n^+].
\end{equation}
Let $E_n$ be the event that
\begin{itemize}[noitemsep,nolistsep]
\item there exist horizontal crossings of $\Lambda_n$ passing below and above the origin,
\item there exist vertical crossings of $\Lambda_n$ passing on the left and the right of the origin.
\end{itemize}
By symmetry and the FKG inequality \eqref{eq:FKG},
\begin{equation*}
  \label{eq:8}
  \phi_{\bbZ^2}[E_n]\ge \phi_{\bbZ^2}[H_n^+]^4.
\end{equation*}
Imagine for a moment that $E_n$ is contained in the event $C_n$ that there exists a circuit surrounding 0 which is in a connected component intersecting $\partial\Lambda_n$ (which is the case if the graph is planar; see Fig.~\ref{fig:A7}). By \eqref{eq:17} and \eqref{eq:P1}, we would deduce that 
\begin{equation}\label{eq:a19}\limsup \phi_{\bbZ^2}[C_n]>0.\end{equation}
Together with $\phi_{\mathbb Z^2}[0\leftrightarrow \infty]=0$, this would imply that $\phi_{\bbZ^2}[E]>0$, and therefore by ergodicity of the infinite-volume random-cluster measure \cite[Theorem~4.19]{Gri06}, that $\phi_{\bbZ^2}[E]=1$. 

Unfortunately, the event $E_n$ does not quite imply the occurrence of a circuit around $0$ due to the fact that long edges allow paths to ``jump over each other''. Nevertheless, a technical result presented in the Appendix shows that a slight modification of this argument implies that with positive probability, paths do not jump over each other and simply intersect. More precisely, Theorem~\ref{thm:glueCirc} implies that there exists a constant $c_{\rm glue}>0$ such that for every $n\ge1$,
\begin{align*}\phi_{\bbZ^2}[C_n]&\ge c_{\rm glue} \cdot \phi_{\bbZ^2}[\Lambda_n \text{ is crossed horizontally in $\omega$}]^{20}.\end{align*}
Together with \eqref{eq:P1}, this implies \eqref{eq:a19} and concludes the proof.
\end{proof}
\begin{figure}[htbp]
  \centering
  \includegraphics[width=0.3\linewidth]{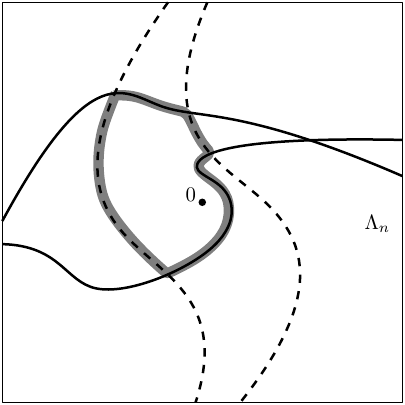}
   \caption{Combining four crossings of $\Lambda_n$, each one passing on one well-chosen side of the origin, creates a circuit surrounding the origin (in light gray). }
  \label{fig:A7}
\end{figure}

We now turn to a sufficient condition to obtain \eqref{eq:P1}.

\begin{definition}An annulus $\Lambda_n\setminus\Lambda_m$ is said to be {\em crossed from inside to outside in $f$} if there exists an open (in $f$) continuous path staying in $\Lambda_n\setminus\Lambda_m$ and going from $\partial\Lambda_n$ to $\partial\Lambda_m$.
\end{definition}

\begin{lemma}\label{lem:P2}
The condition \eqref{eq:P1} is satisfied as soon as
\begin{equation}
\label{eq:P2}\tag{{\bf P2}} \lim_{n\ge0} {\bf P}_{\mathbb Z^2}^\emptyset[\Lambda_{2n}\setminus \Lambda_n\text{ is crossed from inside to outside in }\widehat\n]=0.
\end{equation}

\end{lemma}

Condition \eqref{eq:P2} should be understood, in the light of \eqref{eq:1}, as a statement about correlations in the random-cluster model. Roughly speaking, if \eqref{eq:P2} occurs, then the correlations in the random-cluster model between $\Lambda_n$ and the outside of $\Lambda_{2n}$ are small. This useful fact allows us to show the lemma in three steps:
\begin{enumerate}
\item First, we show, using a second-moment method, that
\be\limsup_{n\to \infty} \phi_{\mathbb Z^2}[\Lambda_{2n}\setminus \Lambda_n\text{ is crossed from inside to outside in }\omega]>0.\ee
\item Second, we use \eqref{eq:P2} to perform a renormalization on crossing probabilities to upgrade the first step to
\be\inf_{n\ge0} \phi_{\mathbb Z^2}[\Lambda_{2n}\setminus \Lambda_n\text{ is crossed from inside to outside in }\omega]>0.\ee
\item Finally, we use a Russo-Seymour-Welsh type argument to show that the previous step implies \eqref{eq:P1}. The Russo-Seymour-Welsh theory was developed by Russo \cite{Rus78} and Seymour-Welsh \cite{SeyWel78} in the case of Bernoulli percolation to relate the probability of crossing long rectangles in the hard direction with the probability of crossing them in the easy direction. This theory was found to be very useful to study planar models, and we refer to \cite{DumTas16} and references therein for a presentation of the current knowledge on the subject. Here, we rely on a version for finite-range interactions of a result \cite{Tas15} extending (in a slightly weaker version) the Russo-Seymour-Welsh theory to planar percolation models satisfying the FKG inequality.
\end{enumerate}
We postpone the proof of this lemma to Section~\ref{sec:P3} and present the last ingredient of the proof.
\begin{lemma}\label{lem:P3}
Assume $\phi_{\mathbb Z^2}[0\leftrightarrow\infty]=0$. If \eqref{eq:P1} and \eqref{eq:P2} are not satisfied, then
$${\bf P}_{\mathbb Z^2}^{\emptyset,\emptyset}[\widehat{\n_1+\n_2}\in E]=1.$$
\end{lemma}
The proof is based on the fact that, on the one hand, by {\bf non}\eqref{eq:P1}, the probability of seeing squares which are crossed in $\omega$ (which contains $\widehat\n$) is small. But on the other hand, {\bf non}\eqref{eq:P2} implies that the probability that the annulus $\Lambda_{2n}\setminus\Lambda_n$ is crossed  from inside to outside  in $\widehat\n$  is larger than some $c>0$ for infinitely many integers $n$. For such values of $n$, we deduce from the two previous claims that the crossings of $\Lambda_{2n}\setminus\Lambda_n$ in $\widehat\n$ are necessarily tortuous. Combining such tortuous paths in $\widehat\n_1$ and $\widehat\n_2$ will allow us to construct circuits in $\widehat{\n_1+\n_2}$.
The complete proof is presented in Section~\ref{sec:P3}.

We are now ready to wrap up the proof of Theorem~\ref{prop:condition} when $\beta=\beta_c$ and $\phi_{\bbZ^2}[0\leftrightarrow\infty]=0$.
\begin{proof}[Proof of Theorem~\ref{prop:condition} for the case $\beta=\beta_c$ and $\phi_{\bbZ^2}{[0\leftrightarrow\infty]}=0$.]   

The first item of Lemma~\ref{lem:flowArgument} and the coupling between $\eta\subset\widehat \n\subset\omega$ show that it suffices to show that $\phi_{\mathbb Z^2}[\eta\in E]=1$. 
To show this, assume first that \eqref{eq:P1} and \eqref{eq:P2} are not satisfied. Since $\phi_{\mathbb Z^2}[0\leftrightarrow \infty]=0$, Lemma~\ref{lem:P3} and the second item of Lemma~\ref{lem:flowArgument} give that $\phi_{\mathbb Z^2}[\eta\in E]=1$. We may therefore assume that \eqref{eq:P1} or \eqref{eq:P2} are satisfied. By Lemma~\ref{lem:P2}, \eqref{eq:P2} implies \eqref{eq:P1} so that \eqref{eq:P1} is necessarily satisfied. Lemma~\ref{lem:P1} gives that $\phi_{\mathbb Z^2}[\omega\in E]=1$. 
\end{proof}

To conclude the proof of this theorem, we need to prove Lemmas~\ref{lem:P2} and \ref{lem:P3}, and treat the case where $\phi_{\bbZ^2}[0\leftrightarrow\infty]>0$.

\subsection{Proof of Lemma~\ref{lem:P2}}\label{sec:P2}

We use the strategy outlined in Section~\ref{sec:outline}, and proceed  in three steps.

\paragraph{Step 1: Second moment method.}
We use a second-moment method to show that
\begin{equation}
\limsup_{n\rightarrow\infty} \phi_{\mathbb Z^2}^\emptyset[\Lambda_{2n}\setminus \Lambda_n\text{ is crossed from inside to outside in }\omega]>0.\label{eq:16}
\end{equation}

Let $\mathsf N$ be the number of pairs of vertices $x\in\Lambda_n$ and $y\in \Lambda_n':=(3n,0)+\Lambda_n$ which are connected in $\omega$. Proposition~\ref{prop:coupling} implies that
\begin{align*}\phi_{\mathbb Z^2}[\mathsf N]&=\sum_{x\in \Lambda_n}\sum_{y\in \Lambda_n'}\phi_{\mathbb Z^2}[x\leftrightarrow y] =\sum_{x\in \Lambda_n}\sum_{y\in \Lambda_n'}\langle\sigma_x\sigma_y\rangle_{\mathbb Z^2}\end{align*}
and
\begin{align*}
\phi_{\mathbb Z^2}[\mathsf N^2]&=\sum_{x,z\in \Lambda_n}\sum_{y,t\in \Lambda_n'}\phi_{\mathbb Z^2}[x\leftrightarrow y,z\leftrightarrow t]\\
& \le \sum_{x,z\in \Lambda_n}\sum_{y,t\in \Lambda_n'}\phi_{\mathbb Z^2}[\mathcal F_{\{x,z,y,t\}}]\\
&\le \sum_{x,z\in \Lambda_n}\sum_{y,t\in \Lambda_n'}\langle\sigma_x\sigma_{z}\sigma_y\sigma_{t}\rangle_{\mathbb Z^2}.
\end{align*}
A classical application of the switching principle gives  (cf.~\eqref{U4}) 
\begin{align*}\langle\sigma_x\sigma_{z}\sigma_y\sigma_{t}\rangle_{\mathbb Z^2}&-\langle\sigma_x\sigma_{z}\rangle_{\mathbb Z^2}\langle\sigma_y\sigma_{t}\rangle_{\mathbb Z^2}-\langle\sigma_x\sigma_y\rangle_{\mathbb Z^2}\langle\sigma_{z}\sigma_{t}\rangle_{\mathbb Z^2}-\langle\sigma_x\sigma_{t}\rangle_{\mathbb Z^2}\langle\sigma_{z}\sigma_y\rangle_{\mathbb Z^2}\\
&=-2\langle\sigma_x\sigma_{z}\sigma_y\sigma_{t}\rangle_{\mathbb Z^2}\cdot{\bf P}^{\{x,z,y,t\},\emptyset}_{\bbZ^2}[x,z,y,t\text{ are all connected in }\widehat{\n_1+\n_2}]\le 0,\end{align*}
which, in turn, implies that$$\langle\sigma_x\sigma_{z}\sigma_y\sigma_{t}\rangle_{\mathbb Z^2}\le\langle\sigma_x\sigma_{z}\rangle_{\mathbb Z^2}\langle\sigma_y\sigma_{t}\rangle_{\mathbb Z^2}+\langle\sigma_x\sigma_y\rangle_{\mathbb Z^2}\langle\sigma_{z}\sigma_{t}\rangle_{\mathbb Z^2}+\langle\sigma_x\sigma_{t}\rangle_{\mathbb Z^2}\langle\sigma_{z}\sigma_y\rangle_{\mathbb Z^2}.$$
Summing over $x,z\in \Lambda_n$ and $y,t\in \Lambda_n'$ gives
\begin{align*}
\phi_{\mathbb Z^2}[\mathsf N^2]&\le 2\phi_{\mathbb Z^2}[\mathsf N]^2 + \Big(\sum_{x,z\in \Lambda_n}\langle\sigma_x\sigma_{z}\rangle_{\mathbb Z^2}\Big)^2.
\end{align*}
The Cauchy-Schwarz inequality (used in the second inequality) then yields
$$\phi_{\mathbb Z^2}[\Lambda_{2n}\setminus\Lambda_n\text{ is crossed from inside to outside}]\ge\phi_{\mathbb Z^2}[\mathsf N>0]\ge \frac{\phi_{\mathbb Z^2}[\mathsf N]^2}{\phi_{\mathbb Z^2}[\mathsf N^2]},$$
so that the claim follows from the two previous inequalities if there exists $c_1>0$ such that for infinitely many $n$,
\begin{equation}\label{eq:ag}\phi_{\mathbb Z^2}[\mathsf N]\ge c_1\sum_{x,z\in \Lambda_n}\langle\sigma_x\sigma_{z}\rangle_{\mathbb Z^2}
.\end{equation}
We now use the monotonicity of spin-spin correlations, which is known in the nearest-neighbor case as the Messager-Miracle-Sole inequality \cite{Sch77,MesMir77,Heg77} and is extended to finite-range interactions in Theorem~\ref{thm:mono} in the Appendix. This theorem (and its subsequent remark on straightforward generalizations to reflections with respect to diagonals),  together with Griffiths' inequality, implies that there exists a constant $c_2>0$ such that 
\begin{equation}\label{eq:MMS}\langle\sigma_{0}\sigma_{y}\rangle_{\mathbb Z^2}\le c_2\langle\sigma_{0}\sigma_{x}\rangle_{\mathbb Z^2}\end{equation}
for any two vertices $ x=(x_1,x_2) $ and $ y =(y_1,y_2) \in \Z^2 $ such that $0 \leq x_1 \leq y_1 $ and $y_2 \geq x_1 + x_2 - y_1 $. 
Applying this inequality to the left side of~\eqref{eq:ag} yields \begin{equation}\label{eq:kkk1}\phi_{\mathbb Z^2}[\mathsf N]\ge c_3n^3 u_{7n},\end{equation} where we abbreviated  $u_k:=k\langle\sigma_0\sigma_{(k,0)}\rangle_{\bbZ^2}$. The same inequality holds true with $u_{7n+i}$, $0\le i\le 6$, instead of $u_{7n}$. 
Moreover, the right side of \eqref{eq:ag} is estimated by
\begin{align}\label{eq:akk}
&\sum_{x,z\in \Lambda_n}\langle\sigma_x\sigma_{z}\rangle_{\bbZ^2}\le n^2\sum_{x\in\Lambda_{2n}}\langle\sigma_0\sigma_x\rangle_{\bbZ^2}\le c_4 n^2\Big(1+\sum_{k=1}^{2n} u_k\Big).
\end{align}
Since $ \beta = \beta_c $, the Simon-Lieb inequality~\cite{Lie80,Sim80} classically implies that
$$ \sum_{x\in \Lambda_{n+R}\setminus\Lambda_n}\langle\sigma_0\sigma_x\rangle_{\bbZ^2}\ge1$$
for every $n\ge1$ (see also the discussion involving $\varphi_\beta(\Lambda_n)$ in \cite{DumTas15}). Therefore, \eqref{eq:MMS} implies that
$(u_n)$ is bounded from below.  
In particular, this guarantees the existence of an infinite number of integers $n$ such that $ (7n+i)^{-1} \sum_{k=1}^{7n+i} u_k \leq 2 u_{7n+i} $ for some $0\le i\le 6$. 
For such integers, one may get \eqref{eq:ag} by plugging the previous bound in \eqref{eq:akk} and then the inequality thus obtained in \eqref{eq:kkk1}. This concludes the proof of Step 1.

\paragraph{Step 2: Renormalization argument}
We implement a renormalization argument to show that
\begin{equation}
\inf_{n\ge0} \phi_{\mathbb Z^2}[\Lambda_{2n}\setminus \Lambda_n\text{ is crossed from inside to outside in }\omega]>0.\label{eq:7}
\end{equation}
For every $n\ge1$, set $$ u_n:= \phi_{\mathbb Z^2}[\Lambda_{2n}\setminus \Lambda_n\text{ is crossed from inside to outside in }\omega].$$
Below, we obtain a recursive inequality for $u_n$ telling us that as soon as $u_n$ drops below a certain value, then $u_n$ decays rapidly.

Let us start by observing that there exists a constant $C>0$ independent of $n$ such that for every $6n\le m\le 36 n$,
\begin{equation}
  \label{eq:11}
  u_m\le C\max_{\substack{x\in\partial \Lambda_{m}\\y\in\partial \Lambda_{2m}}}\phi_{\mathbb Z^2}[\Lambda_n(x)\leftrightarrow \Lambda_n(y)].
\end{equation}
To see this, consider a covering of $\partial \Lambda_{m}$ and $\partial\Lambda_{2m}$ by at most $O(m/n)$ boxes of size $n$, and observe that one of the boxes covering $\partial \Lambda_{m}$ must be connected to one of the boxes covering $\partial \Lambda_{2m}$ when the annulus $\Lambda_{2m}\setminus\Lambda_{m}$ is crossed from inside to outside.

Applying the decorrelation inequality~\eqref{eq:1} with $A:=\{\Lambda_n(x)\leftrightarrow\partial\Lambda_{2n}(x)\}$ and  $B:=\{\Lambda_n(y)\leftrightarrow\partial\Lambda_{2n}(y)\}$, we find that for every $x\in \partial\Lambda_{m}$ and $y\in\partial \Lambda_{2m}$,
\begin{align*}
  \label{eq:9}
  \phi_{\mathbb Z^2}[\Lambda_n(x)\leftrightarrow\Lambda_n(y)]&\le \phi_{\bbZ^2}[A\cap B]\\
&\le  \phi_{\mathbb Z^2}[A]\phi_{\mathbb Z^2}[B] + {\bf P}^\emptyset_{\mathbb Z^2}[\Lambda_{2n}(x)\stackrel{\widehat\n}{\longleftrightarrow}\Lambda_{2n}(y)]\\
&\le u_n^2+\ep_n,
\end{align*}
where 
$$\ep_n:={\bf P}_{\mathbb Z^2}^\emptyset[\Lambda_{4n}\setminus\Lambda_{2n}\text{ is crossed from inside to outside in }\widehat\n].$$ (We used that $m\ge6n$ to guarantee that $\Lambda_{4n}(x)$ and $\Lambda_{2n}(y)$ do not intersect.) Plugging the previous inequality in \eqref{eq:11}, we obtain that for every $n\ge1$ and every $6n\le m\le 36 n$,
\begin{equation}
  \label{eq:13}
  u_m\le C(u_n^{2}+\ep_n).
\end{equation}
By \eqref{eq:P2}, $\ep_n$ converges to $0$ when $n$ tends to infinity. If we assumed that $\inf_{n\ge1}u_n=0$, then the previous equation applied to $m=6n$ implies that there exists $n$ such that  $u_n$ converges to $0$ along powers of $K6^k$ for some $K>0$. Then, \eqref{eq:13} directly implies the convergence to zero for every $n$, which contradicts \eqref{eq:16}. We therefore deduce that $\inf_{n\ge1}u_n>0$, as wanted.

\paragraph{Step 3: Russo-Seymour-Welsh (RSW) argument.}
In this last step, we show that the crossing estimate \eqref{eq:7} implies that \textbf{(P1)} holds.  We proceed by contradiction and show that if the probability of crossing a square tends to 0, then the infimum of the probabilities of crossing annuli is zero. Since we work only with the random-cluster measure, we drop the ``in $\omega$'' in the events below.
The proof is decomposed in two steps. First, we show that 
  \begin{equation}
  \label{eq:58}
  \lim_{n\to \infty} \phi_{\bbZ^2}[\Lambda_n \text{ is crossed horizontally}]=0
\end{equation}
implies
\begin{equation}
  \label{eq:67}
  \liminf_{n\to\infty} \phi_{\mathbb Z^2}\big[ [-3n,3n]\times[-2n,2n] \text{ is crossed vertically}\big]=0.
\end{equation}
Second, we show that \eqref{eq:67} implies
\begin{equation}
  \label{eq:62}
  \liminf_{n\to\infty} \phi_{\mathbb Z^2}[\Lambda_{2n}\setminus \Lambda_n\text{ is crossed from inside to outside}]=0.
\end{equation}
\begin{figure}[htbp]
  \centering
  \includegraphics[width=1\linewidth]{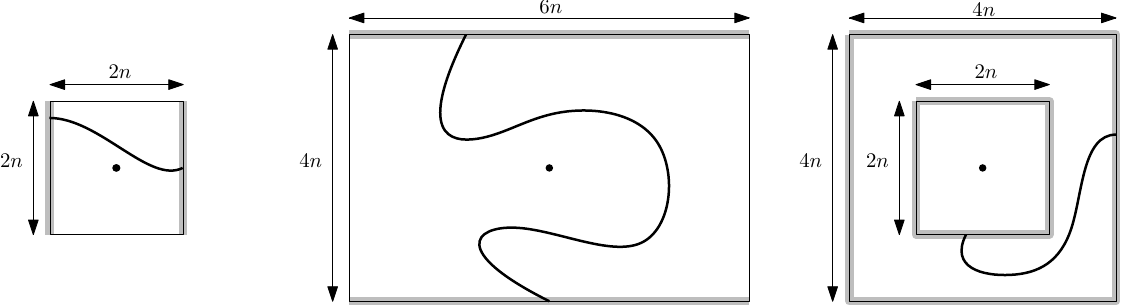}
   \caption{The events involved in \eqref{eq:58}, \eqref{eq:67} and \eqref{eq:62}. }
  \label{fig:crossingEvents}
\end{figure}

\begin{remark}Before proving these implications, let us digress a bit and explain why we call this step an RSW argument. At first sight, this implication may not look like a standard RSW statement which usually gives a lower bound on the probability to cross a rectangle in the long direction, provided a lower bound on the probability to cross a square. To see a connection here with this type of statement, we shall look at a dual version of it. Observe that the absence of a horizontal crossing in the box $\Lambda_n$ corresponds to the existence of a closed dual surface from top to bottom that prevents the existence of any horizontal crossing\footnote{In the planar case, such a closed dual surface corresponds to a closed path  in the dual graph.}.
The complement of the events under consideration in \eqref{eq:58}, \eqref{eq:67} and \eqref{eq:62} are therefore the existence of dual surfaces with a certain topology. From this point of view, the statement is indeed an RSW-type result.
\end{remark}

Let us now dive into the proof. We assume that \eqref{eq:58} holds and our goal is to show that large rectangles are also crossed in the easy direction with low probability. We proceed in two steps:
\begin{enumerate}
\item We define four sequences of crossing events $(A_n)$, $(B_n)$, $(C_n(\alpha_n))$ and $(D_n(\alpha_n))$ and prove that  under the assumption of \eqref{eq:58}, their probabilities tend to 0 as $n$ tends to infinity (this will follow from simple constructions).
\item We use these events in a key construction to show that large rectangles are crossed in the easy direction with low probability. This key construction, presented in Lemma~\ref{lem:key_construction}, is valid only for infinitely many $n$ (and not a priori for every $n$), which explains the presence of a liminf in \eqref{eq:67} and \eqref{eq:62}. \end{enumerate}

\begin{remark}  Let us highlight the fact that some of 
the statements presented  below are made under assumptions which  we are being disproved by contradiction  (e.g., 
\eqref{eq:58}).   If read out off context,  
some of the statements made below may appear wrong or at least in variance with the expected fractal behavior at criticality.  To avoid confusion, it is therefore recommended to read the proof  given below  step by step.
\end{remark}
\begin{figure}[htbp]
  \centering
  \includegraphics[width=.23\linewidth]{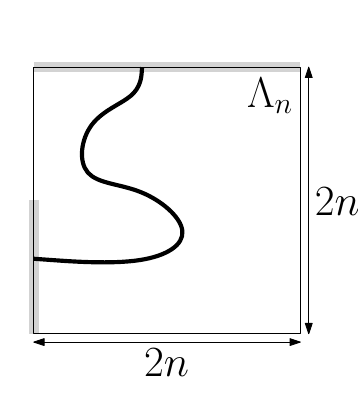}
  \hfill
  \includegraphics[width=.23\linewidth]{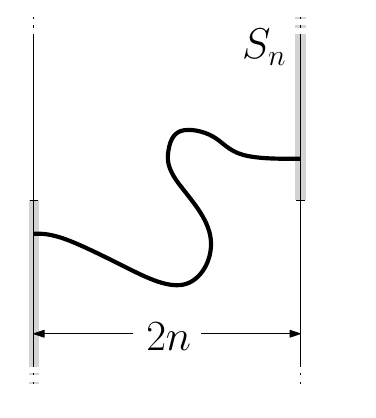}
  \hfill
  \includegraphics[width=.23\linewidth]{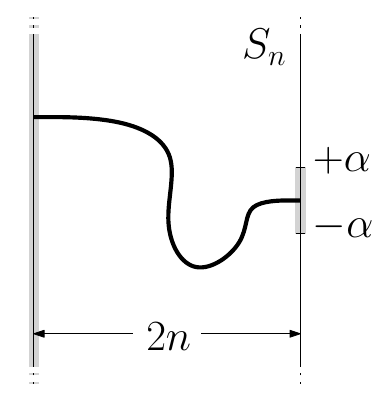}
  \hfill
  \includegraphics[width=.23\linewidth]{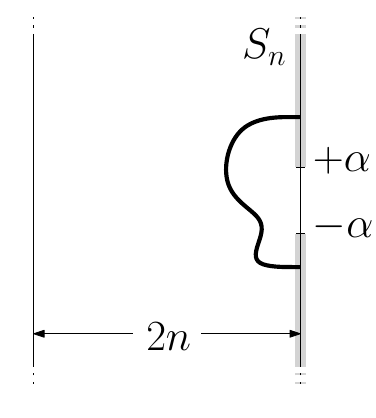}
\hfill
  \caption{Diagrammatic representations of the events $A_n$, $B_n$, $C_n(\alpha)$ and $D_n(\alpha)$. }
  \label{fig:ABCD}
\end{figure}
\bigbreak\noindent
{\bf Step 1: Bounds on the probability of the events $A_n$, $B_n$, $C_n(\alpha_n)$ and $D_n(\alpha_n)$  assuming \eqref{eq:58}, i.e.~assuming that crossing probabilities for squares tend to 0.}
\bigbreak\noindent
Fix $n,\alpha\ge0$. Let $S_n:=[-n,n]\times\bbR$ be the vertical strip of width $2n$. Define
\begin{align*}A_n&:=\{\exists\text{ open continuous path in $\Lambda_n$ from $\{-n\}\times [-n,0]$ to the top side of $\Lambda_n$}\},\\
B_n&:=\{\exists\text{ open continuous  path in $S_n$ from $\{-n\}\times\bbR_-$ to $\{n\}\times \mathbb R_+$}\},\\
C_n(\alpha)&:=\{\exists\text{ open continuous  path in $S_n$ from $\{-n\}\times\bbR$ to $\{n\}\times[-\alpha,\alpha]$}\},\\
D_n(\alpha)&:=\{\exists\text{ open continuous  path in $S_n$ from $\{n\}\times(-\infty,-\alpha]$ to $\{n\}\times[\alpha,\infty)$}\}.
\end{align*}
Let $\tilde A_n$ be the symmetric mirror image of $A_n$ with respect to the $x$-axis. The FKG inequality \eqref{eq:FKG} implies that 
$$\phi_{\bbZ^2}[\tilde A_n\cap A_n]\ge \phi_{\bbZ^2}[A_n]^2.$$
As before (see the proof of Lemma~\ref{lem:P1}), this event does not quite imply the existence of a vertical crossing of $\Lambda_n$, yet, a ``gluing principle'' (Theorem~\ref{thm:gluepaths}), applied to $A_n$ and $\tilde A_n$, gives that
$$\phi_{\bbZ^2}[\Lambda_n\text{ is crossed vertically}]\ge c_{\rm glue}\phi_{\bbZ^2}[A_n]^2.$$
We deduce from \eqref{eq:58} that
  \begin{equation}
    \label{eq:23}
    \lim_{n\to\infty}\phi_{\bbZ^2}[A_n]=0.
  \end{equation}
To bound the probability of $B_n$, distinguish between the cases depending on whether the continuous path crossing $S_n$ from $\{-n\}\times\bbR_-$ to $\{n\}\times\mathbb R_+$ remains or not in the box $\Lambda_n$ to get $$
    \phi_{\bbZ^2}[B_n]\le \phi_{\bbZ^2}[\Lambda_n \text{ is crossed horizontally}] +2 \phi_{\bbZ^2}[A_n],
$$
  which, together with \eqref{eq:58} and \eqref{eq:23}, directly implies
  \begin{equation}
    \label{eq:25}
    \lim_{n\to \infty}\phi_{\bbZ^2}[B_n]=0.
  \end{equation}
Bounding the probability of $C_n(\alpha)$ and $D_n(\alpha)$ is slightly more subtle. First of all, these events satisfy
  \begin{equation}
    \label{eq:611}
    \phi_{\bbZ^2}[B_n]\ge\frac{c_{\mathrm{glue}}}2 \phi_{\bbZ^2}[C_n(\alpha)]\cdot \phi_{\bbZ^2}[D_n(\alpha)].
  \end{equation}
To prove this, first observe that, by symmetry, there is an open continuous path from $\{-n\}\times\bbR_-$ to $\{n\}\times[-\alpha,\alpha]$ inside $S_n$ with probability larger than $\phi_{\bbZ^2}[C_n(\alpha)]/2$, and then use the ``gluing principle'' (Theorem~\ref{thm:gluepaths}) to combine it, using the FKG inequality \eqref{eq:FKG}, with an open path guaranteeing the occurrence of $D_n(\alpha)$.

The equation above shows that either the probability of $C_n(\alpha)$ or the probability of $D_n(\alpha)$ is controlled by the probability of $B_n$. In order to get information on both probabilities, we choose $\alpha$ in such a way that the two probabilities on the RHS of \eqref{eq:611} are close to each other.
 More precisely, define
  \begin{equation}
    \label{eq:65}
    \alpha_n:=\min\big\{\alpha\ge0 \text{ such that } \phi_{\bbZ^2}[C_n(\alpha)] \ge  \phi_{\bbZ^2}[D_n(\alpha)]\big\}.
  \end{equation}
By applying \eqref{eq:611} to $\alpha=\alpha_n$, we get from \eqref{eq:25} that
 \begin{equation}    \lim_{n\to\infty}\phi_{\bbZ^2}[D_n(\alpha_n)]=0.
    \end{equation}
Similarly, by applying \eqref{eq:611} to $\alpha=\alpha_n-1$ ($\alpha_n\ge1$ follows from the finite-energy and the fact that $\phi_{\bbZ^2}[D_n(\alpha_n)]$ is small), we get from \eqref{eq:25} that
$\lim\phi_{\bbZ^2}[C_n(\alpha_n-1)]=0.
$
    From this, one may easily deduce -- for instance since $\phi_{\bbZ^2}[C_n(\alpha+1)]\le \phi_{\bbZ^2}[C_n(2\alpha)]\le 2\phi_{\bbZ^2}[C_n(\alpha)]$ -- that
  \begin{equation}    \lim_{n\to\infty}\phi_{\bbZ^2}[C_n(\alpha_n)]=0.
    \end{equation}
 \bigbreak\noindent
    {\bf Step 2: The key construction.}
Here and below, we used the notation $\alpha_n$ defined in \eqref{eq:65}. Set
\begin{equation*}
      \varepsilon_n:=\max\big\{\phi_{\bbZ^2}[\Lambda_n\text{ is crossed horizontally}],\phi_{\bbZ^2}[A_n],\phi_{\bbZ^2}[B_n],\phi_{\bbZ^2}[C_n(\alpha_n)],\phi_{\bbZ^2}[D_n(\alpha_n)]\big\}.
 \end{equation*}
In this second step, we wish to bound from above the probability to cross a rectangle in the easy direction in terms of $\ep_n$. To do this, we will use a bound (Lemma~\ref{lem:key_construction} below) valid only when $n$ is such that $\alpha_{2n}\le 4\alpha_{n}$. Before focusing on this bound, observe that the set $\{n:\alpha_{2n}\le 4\alpha_{n}\}$ is infinite, since otherwise $\alpha_n$ would grow super linearly in $n$, which is in contradiction with the fact that $\alpha_n\le n$. To see this last fact, simply combine the definition \eqref{eq:65} of $\alpha_n$ with the trivial inequalities $$ \phi_{\bbZ^2}[C_n(n)]\ge  \phi_{\bbZ^2}[\Lambda_n\text{ is crossed horizontally}]\ge \phi_{\bbZ^2}[D_n(n)].$$
Hence, the proof of \eqref{eq:58}$\Rightarrow$\eqref{eq:67} follows from the next lemma and the fact that, by Step 1, \eqref{eq:58} implies that $\varepsilon_n$ tends to 0.
\begin{lemma}\label{lem:key_construction}
  For every $n\in\mathbb N$ such that $\alpha_{2n}\le4\alpha_n$, we have
  \begin{equation}
    \label{eq:68}
    \phi_{\bbZ^2}\big[ [-3n,3n]\times[-2n,2n] \text{ is crossed vertically}\big]\le 14 \varepsilon_n+2\ep_{2n}.
  \end{equation}
\end{lemma}

\begin{figure}[htbp]
  \centering
  \includegraphics[width=1\linewidth]{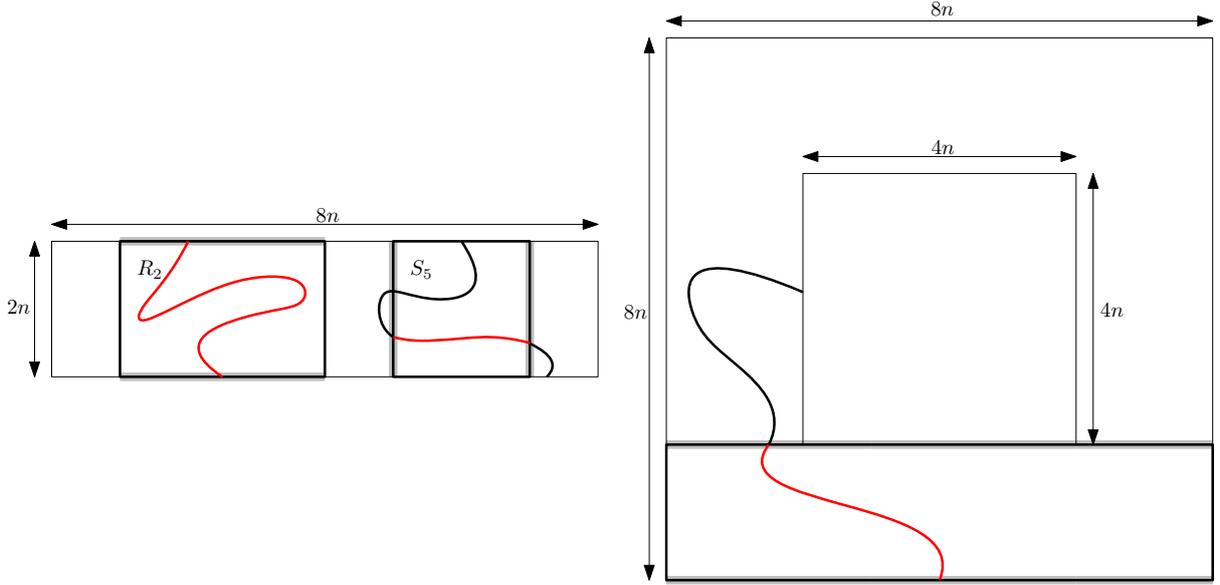}
   \caption{On the left, the rectangle $R_2$ as well as the square $S_5$. One example of vertical crossing of $\overline R$ crossing $R_2$ vertically but not $S_i$ horizontally, as well as one example crossing $S_5$ horizontally but no $R_i$ vertically. On the right, for the annulus to be crossed from inside to outside, one of four rectangles isomorphic to $R$ (two are rotated by $90$ degrees) must be crossed in the easy direction. An example of such a crossing, with the vertical crossing of the bottom rectangle highlighted.}
  \label{fig:ABCDE}
\end{figure}

Before proving the lemma, let us show \eqref{eq:67}$\Rightarrow$\eqref{eq:62}, so that the proof of \eqref{eq:58}$\Rightarrow$\eqref{eq:62} is complete. This part is fairly easy and is summarized in Fig.~\ref{fig:ABCDE}.

Fix $n\ge1$. Consider the rectangle $\overline R:=[-8n,8n]\times[-2n,2n]$ and cover it with the following six translates of the rectangle $R:=[-3n,3n]\times[-2n,2n]$:
\begin{equation}
  \label{eq:66}
  R_i:=(-5n+2in,0)+R,\quad 0\le i\le 5.
\end{equation}
For $0\le i\le 4$, the overlap of $R_i$ and $R_{i+1}$ is a square denoted $S_i$. If $\overline R$ is crossed vertically, then at least one of the rectangles $R_i$ is crossed vertically, or one of the squares $S_i$ is crossed horizontally. Therefore, using translation invariance and the fact that a horizontal crossing in the square $S_i$ occurs with probability lower than a vertical crossing in $R$, we obtain
\begin{equation}
  \label{eq:69}
  \phi_{\bbZ^2}[\overline R \text{ is crossed vertically}]\le 11 \cdot\phi_{\bbZ^2}\big[R \text{ is crossed vertically}\big].
\end{equation}
Observe that  the annulus $\Lambda_{8n}\setminus \Lambda_{4n}$ can be covered by four rectangles isomorphic to $\overline R$ in such a way that any crossing of the annulus must cross at least one of the four rectangles in the easy direction. Hence,  \eqref{eq:69} implies that  \begin{equation}
\label{eq:70}
 \phi_{\bbZ^2}[\Lambda_{8n}\setminus \Lambda_{4n}\text{ is crossed from inside to outside}]\le 44\cdot \phi_{\bbZ^2}\big[R \text{ is crossed vertically}\big],
\end{equation}
which, together with \eqref{eq:67},  concludes the proof of \eqref{eq:62}.
\bigbreak
It only remains to prove Lemma~\ref{lem:key_construction}, which we do now.

\begin{proof}[Proof of Lemma~\ref{lem:key_construction}]
  Consider the square $S:=(-n,0)+\Lambda_{2n}$. Let $I$  and $K$ be its top and bottom sides respectively, and  define the following vertical subsegments of its right  boundary
  \begin{equation}
    \label{eq:71}
    J_1:=\{n\}\times[-2n,-\alpha_{2n}],\quad J_2:=\{n\}\times[-\alpha_{2n},\alpha_{2n}],\quad  J_3:=\{n\}\times[\alpha_{2n},2n].
  \end{equation}
  Then, let $S'$, $I'$, $J'_1$, $J_2'$, $J_3'$ and $K'$ be the images of the sets above through the vertical reflection on the $y$-axis (see Fig.~\ref{fig:SandSprime}).
  \begin{figure}[htbp]
    \centering
    \includegraphics[width=.33\linewidth]{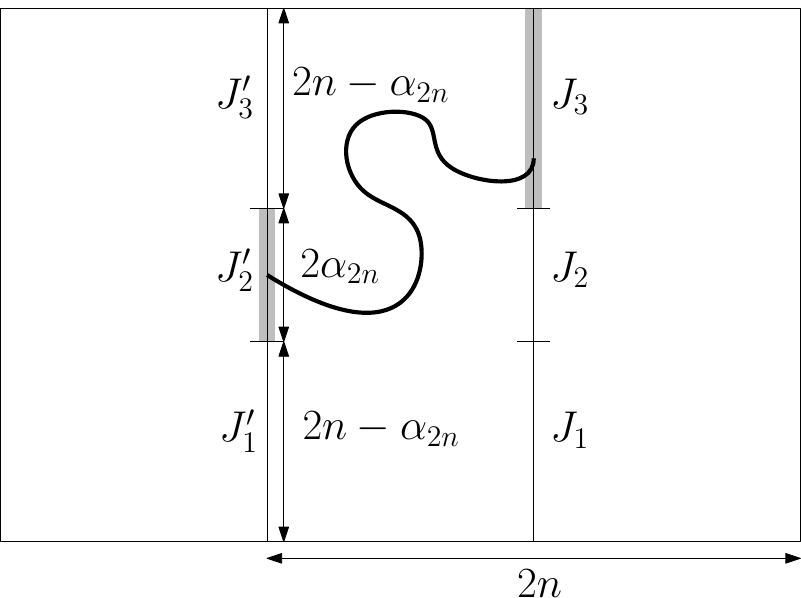}\hfill
      \includegraphics[width=.33\linewidth]{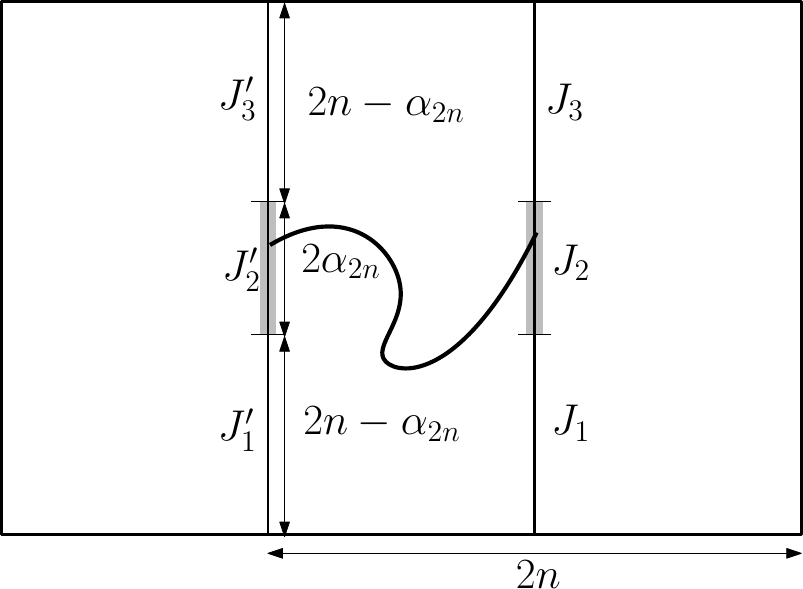}\hfill
        \includegraphics[width=.33\linewidth]{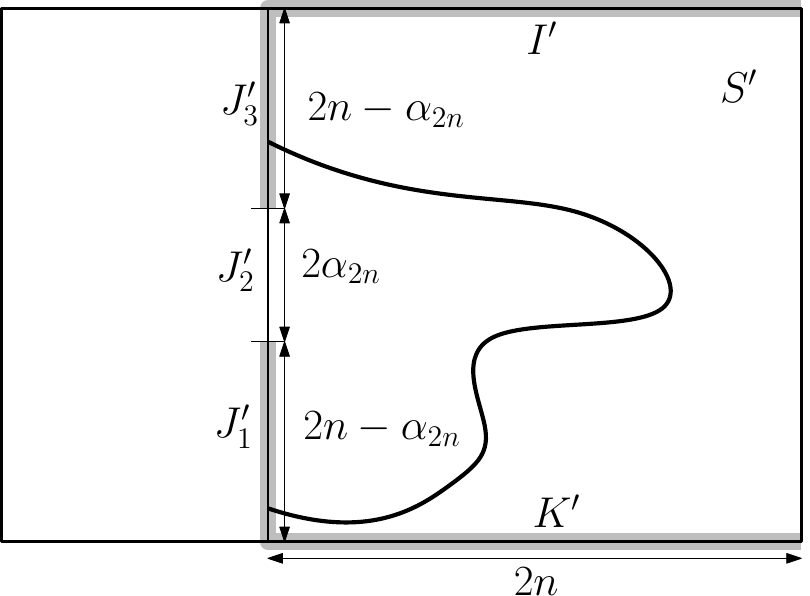}
    \caption{The three cases (i), (ii) and (iii). In the third case, we only depicted one of the two sub-cases in which $I'\cup J_1'$ is connected to $J_3'\cup K'$ in $S'$.}
    \label{fig:SandSprime}
  \end{figure}

 We claim that the existence of a vertical crossing in $[-3n,3n]\times[-2n,2n]$ implies that at least one of the following three events occurs:
  \begin{itemize}[noitemsep]
  \item[(i)] there exist $i\neq j$ such that $J_i'$ is connected to $J_j$ in $S\cap S'$,
  \item[(ii)] the set $J_2'$ is connected to $J_2$ in $S\cap S'$,
  \item[(iii)] either $I\cup J_1$ is connected to $J_3\cup K$ in $S$, or $I'\cup J_1'$ is connected to $J_3'\cup K'$ in $S'$.
  \end{itemize}
  To prove this, consider a vertical crossing $\gamma$ of $R$. Assume that (i) and (ii) do not occur. We wish to show that (iii) occurs. The assumption that (i) and (ii) do not occur implies that every horizontal crossing of $S\cap S'$ must be of one of the following types:
  \begin{description}[noitemsep]
  \item[type 1]  the continuous path crosses $S\cap S'$ from $J_1'$ to $J_1$,
  \item[type 2] the continuous path crosses $S\cap S'$ from $J_3'$ to $J_3$.
  \end{description}
We wish to prove that (iii) necessarily occurs if $\gamma$ contains only continuous paths of types 1 and 2. We distinguish between the different possible cases:
 \begin{itemize}
 \item  {\em $\gamma$ has no subpath of type 1 or 2.} Then, $\gamma$ does not cross $S\cap S'$  horizontally. Therefore, $\gamma$ is either contained in $S$ or in $S'$, and (iii) occurs because either $S$ or $S'$ is crossed vertically.

 \item {\em $\gamma$ has a subpath of type 1 but no subpath of type 2.} Then, consider the last subpath $\pi$ of type 1 when following $\gamma$ from the top to the bottom of $R$. Without loss of generality, we assume that $\pi$ crosses $S\cup S'$ from left to right. In this case, $\pi$ must be connected to the bottom side of $S'$ inside $S'$. Hence, $J_1'$ is connected to the bottom side of $S'$ inside $S'$, which implies that (iii) occurs.

 \item {\em  $\gamma$ has a subpath of type 2 but no subpath of type 1.} This case is treated similarly to the previous one.

\item {\em  $\gamma$ has subpaths of both types.} Then, consider two consecutive subpaths $\pi_1$ and $\pi_2$ of different types. By symmetry, we may assume that $\pi_1$ is of type 1 and crosses $S\cap S'$ from left to right (this implies that $\pi_2$ is of type 2 and crosses $S\cap S'$ from right to left) and that $\pi_1$ is visited before $\pi_2$ when following $\gamma$ from bottom to top.   
In this case, we see that the part of $\gamma$ between the beginning of $\pi_1$ and the end of $\pi_2$ must be contained in $S'$. Hence, $J_1'$ is connected to $J_3'$ inside $S'$, which implies that (iii) occurs.
\end{itemize}
Overall, we just proved that if $[-3n,3n]\times[-2n,2n]$ is crossed vertically, then either (i), (ii) or (iii) occurs. We deduce that
\begin{align}
\nonumber
  \phi_{\bbZ^2}[[-3n,3n]\times[-2n,2n] &\text{ is crossed vertically}]\\
  &\le  \phi_{\bbZ^2}[\text{\textrm{(i)} occurs}]+\phi_{\bbZ^2}[\text{\textrm{(ii)} occurs}]+\phi_{\bbZ^2}[\text{\textrm{(iii)} occurs}]  \label{eq:72}
\end{align}
and we just need to bound the probability of the events on the right-hand side in terms of $\ep_n$.

First, if (i) occurs, then at least one of four symmetric versions of $B_n$ must occur, hence
\begin{equation}
  \label{eq:74}
   \phi_{\bbZ^2}[\text{\textrm{(i)} occurs}]\le 4 \phi_{\bbZ^2}[B_n]\le 4\ep_n.
\end{equation}
To bound the probability that (ii) occurs, we use our assumption that $\alpha_{2n}\le 4\alpha_{n}$. Thanks to this inequality, we see that (ii) implies, in particular, that $\{n\}\times[-4\alpha_n,4\alpha_n]$ is connected in $\bbS_n$ to $\{-n\}\times\bbR$. Hence,
\begin{equation}
  \label{eq:75}
  \phi_{\bbZ^2}[\text{\textrm{(ii)} occurs}]\le \phi_{\bbZ^2}[C_n(4\alpha_n)]\le 4 \phi_{\bbZ^2}[C_n(\alpha_n)]\le 4\ep_n.
\end{equation}
Finally, by symmetry, we have
\begin{align}
  \label{eq:76}
  \phi_{\bbZ^2}[\text{\textrm{(iii)} occurs}]&\le 2\phi_{\bbZ^2}[I\cup J_1 \leftrightarrow J_3\cup K\text{ in } S]\notag\\
&\le 2\phi_{\bbZ^2}[ I\leftrightarrow K \text{ in } S] + 4 \phi_{\bbZ^2}[  J_1 \leftrightarrow K \text{ in } S] +  2 \phi_{\bbZ^2}[  J_1\leftrightarrow J_3 \text{ in } S] \notag\\
&\le 2\phi_{\bbZ^2}[ \Lambda_n\text{ is crossed horizontally}]+4\phi_{\bbZ^2}[A_n]+2\phi_{\bbZ^2}[D_{2n}(\alpha_{2n})]\notag\\
&\le 6\ep_n+2\ep_{2n}.
\end{align}
Plugging the three bounds \eqref{eq:74}, \eqref{eq:75} and \eqref{eq:76} in \eqref{eq:72} concludes the proof of the lemma.
\end{proof}

\subsection{Proof of Lemma~\ref{lem:P3}}\label{sec:P3}

The proof of this lemma is fairly technical and relies on tricky constructions involving crossings. Before diving into it, let us summarize what the hypothesis that \eqref{eq:P1} and \eqref{eq:P2} do not hold means here.

On the one hand, {\bf non}\eqref{eq:P2} implies
$$\limsup_{n\to \infty} \mathbf P_{\mathbb Z^2}^\emptyset[\Lambda_{2n}\setminus\Lambda_n\text{ is crossed from inside to outside in }\widehat\n]>0.$$
Using a covering of the annulus with four $4n\times n$ rectangles as in the proof of \eqref{eq:67}$\Rightarrow$\eqref{eq:62}, we directly get that
\begin{equation}
  \label{eq:27}
 \limsup_{n\to \infty} \mathbf P_{\mathbb Z^2}^\emptyset\big[[0,n]\times[0,4n]\text{ is crossed horizontally in $\widehat\n$} ]>0.
\end{equation}
On the other hand, {\bf non}\eqref{eq:P1} can be rewritten as
 \begin{equation}
  \label{eq:14}
  \lim_{n\to \infty} \phi_{\bbZ^2}[\Lambda_n \text{ is crossed horizontally}]=0.
\end{equation}
By the coupling between the random-cluster $\omega$ and the random current $\widehat{\mathbf n}$, this implies that large squares are crossed with low probability in $\widehat{\mathbf n}$ also.

To summarize the situation, in the configuration $\widehat\n$, some large rectangles are crossed in the easy direction with a positive probability by~\eqref{eq:27}, while squares are crossed with a low probability by~\eqref{eq:14}. This second fact constrains the geometry of potential crossings: they cannot be too straight (otherwise, some squares would be crossed) and are therefore forced to ``oscillate'' a lot. The core of the argument will consist in using these oscillations to construct circuits in the sum of two currents.

The general idea of the proof of Lemma~\ref{lem:P3} goes as follows. Let $H_n$ be the event that there exists a circuit in $\widehat{\n_1+\n_2}$ surrounding 0 in a connected component intersecting $\partial\Lambda_n$. We aim at proving the following result:
\begin{equation}
  \label{eq:27a}
 \limsup_{n\to \infty} \mathbf P_{\mathbb Z^2}^{\emptyset,\emptyset}\big[H_n]>0.
\end{equation}
This would conclude the proof. Indeed, since there is no infinite connected component in $\widehat {\n_1+\n_2}$, we deduce that $E$ is the intersection on every $n\ge1$ of the event $H_{\ge n}$ that $\widehat{\n_1+\n_2}$ contains a circuit surrounding 0 which is contained in a connected component of radius larger or equal to $n$. Since the sequence $H_{\ge n}$ is decreasing, we deduce that 
$$ \mathbf P_{\mathbb Z^2}^{\emptyset,\emptyset}[E]=\lim_{n\rightarrow\infty} \mathbf P_{\mathbb Z^2}^{\emptyset,\emptyset}[H_{\ge n}]\ge \limsup_{n\to \infty}  \mathbf P_{\mathbb Z^2}^{\emptyset,\emptyset}[H_n]>0.$$
 Since $ \mathbf P_{\mathbb Z^2}^{\emptyset,\emptyset}$ is ergodic (as product of two measures $\mathbf P_{\mathbb Z^2}^{\emptyset}$, which were proved to be ergodic in \cite{AizDumSid15}), we deduce that $ \mathbf P_{\mathbb Z^2}^{\emptyset,\emptyset}[E]=1$ as wanted.
 
To prove \eqref{eq:27a}, we proceed as follows. In a first step, we do as in the previous section: we work with the random-cluster measure and show that the probabilities of two families of events tend to 0. In a second step, we work with one current and prove that two families of events have probabilities which are {\em not} tending to 0. In a third step, we use the first two steps to combine two paths in two independent currents $\widehat\n_1$ and $\widehat\n_2$ to create a circuit in $\widehat{\n_1+\n_2}$ around the origin.

\begin{figure}[htbp]
  \centering
\includegraphics[width=0.80\textwidth]{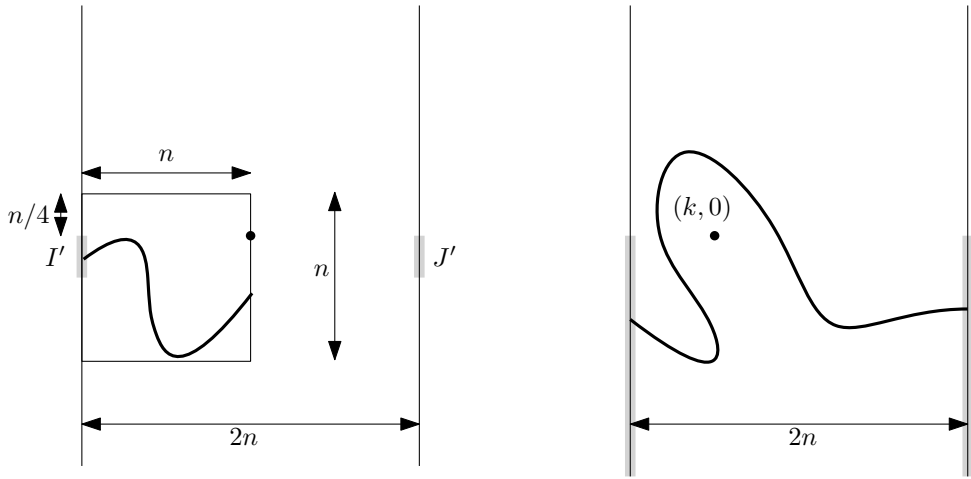}
 \caption{The events $E_n$ and $F_n(k)$.}
  \label{fig:circuit}
\end{figure}

\paragraph{Step 1: Construction of low probability events under the assumption of {\bf non}\eqref{eq:P1}} 
 
Fix two integers $n\ge0$ and $-n\le k\le n$. Recall that $S_n$ is the strip $[-n,n]\times\bbR$. Also, recall the definition of  $B_n$ from the previous section, and the fact that assuming {\bf non}\eqref{eq:P1},
its probability tends to $0$ by \eqref{eq:25}.

Set vertical segments $I':=\{-n\}\times[-n/4,0]$ and $J':=\{n\}\times[-n/4,0]$ (see Fig.~\ref{fig:circuit}) and define
\begin{align*}
E_n&:=\{\exists\text{ open continuous path in $S_n$ from $I'$ to $\{0\}\times\bbR$ not intersecting $[-n,n]\times\{n/4\}$}\},\\
F_n(k)&:=\{\exists\text{ open  continuous path in $S_n$ from $\{-n\}\times\bbR_-$ to $\{n\}\times \mathbb R_-$ passing above $(k,0)$}\}.
\end{align*}
Let $\Lambda_n'$ be the box $[-n,0]\times[-3n/4,n/4]$. We immediately have that
 $$\phi_{\bbZ^2}[E_n]\le \phi_{\bbZ^2}[\Lambda_n'\text{ crossed horizontally}]+\phi_{\bbZ^2}[I'\leftrightarrow [-n,0]\times\{-3n/4\}\text{ in }\Lambda_n'].$$
 Using the FKG inequality \eqref{eq:FKG} and the ``gluing principle'' (Theorem~\ref{thm:gluepaths}), we find that
 $$\phi_{\bbZ^2}[\Lambda'_n\text{ crossed vertically}]\ge c_{\rm glue}\phi_{\bbZ^2}[I'\leftrightarrow [-n,0]\times\{-3n/4\}\text{ in }\Lambda_n']^2.$$
 The last two displayed equations, together with the fact that the probability of $\Lambda_n'$ being crossed tends to 0 by \eqref{eq:14}, give that
    \begin{equation}
    \label{eq:50}
    \lim_{n\to \infty}\phi_{\bbZ^2}[E_n]=0.
  \end{equation}
Using the ``gluing principle'' (see \eqref{eq:190a} in Remark~\ref{rmk:h} after the proof of Theorem~\ref{thm:gluepaths}) once again, one can prove that for every $k$, the event $B_n$ occurs with probability larger than $c_{\rm glue}\phi_{\bbZ^2}[F_n(k)]^2$ (simply use $B_n$ and its reflection with respect to the $x$-axis). 
Hence,
\begin{equation}
  \label{eq:26}
   \lim_{n\to \infty}\max_{-n\le k \le n}\phi_{\bbZ^2}[F_n(k)]=0.
\end{equation}
From now on, define
$$\varepsilon_n:=\max\{\phi_{\mathbb Z^2}[B_n],\phi_{\mathbb Z^2}[E_n],\phi_{\mathbb Z^2}[F_n(k)]:-n\le k\le n\}$$
and observe that \eqref{eq:25} (which follows from \eqref{eq:14}), \eqref{eq:50} and \eqref{eq:26} imply that $\varepsilon_n$ tends to 0.

 \begin{figure}[htbp]
  \centering
\includegraphics[width=1.00\textwidth]{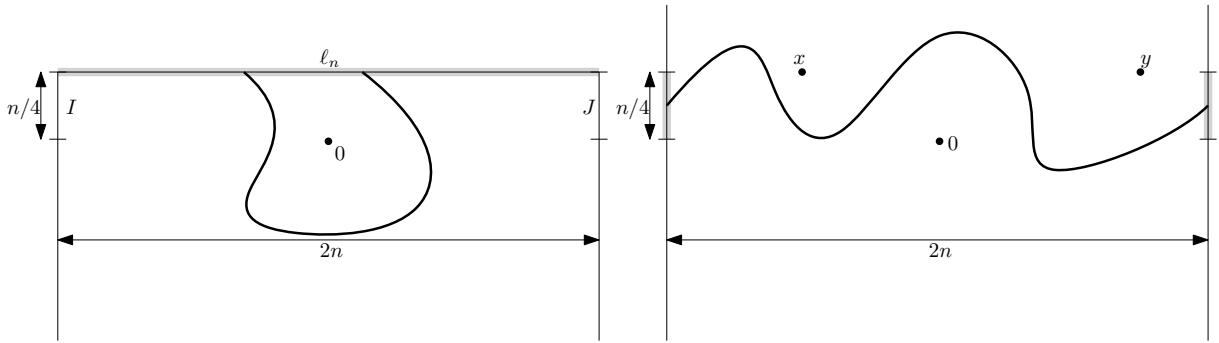}
 \caption{From left to right, the events $W_n$ and $V_n(x,y)$.}
  \label{fig:circuit2}
\end{figure}

\paragraph{Step 2:  Construction of events with good probability under {\bf non}\eqref{eq:P1} and {\bf non}\eqref{eq:P2}.} Set $S_n^+$ to be the half-strip $[-n,n]\times(-\infty,n/4]$. Also, define the vertical segments $I:=\{-n\}\times [0,n/4]$ and $J:=\{n\}\times[0,n/4]$.

Below, we use the notion of continuous path passing {\em below $x$} or {\em above $x$}. We believe that the reader can come up with a proper version of this notion on its own and do not focus too much on the details here, especially that similar considerations are used in the proof of Theorem~\ref{thm:glueCirc} in the appendix, where details are provided.
Define the horizontal line $\ell_n:=\bbR\times\{n/4\}$ (see Fig.~\ref{fig:circuit2}). For $x,y\in \ell_n$, define 
\begin{align*}
U_n&:=\{\exists\text{ open continuous path from $I$ to $J$}\},\\
V_n(x,y)&:=\{\exists\text{ open continuous path in $S_n$ from $I$ to $J$ passing above 0 and below $x$ and $y$}\},\\
W_n&:=\{\exists\text{ open continuous path from $\ell_n$ to $\ell_n$ in $S^+_n$ passing below 0}\}.
\end{align*}
Let us start by bounding from below the probability of the event $U_n$. To do so, observe that
$$\mathbf P_{\mathbb Z^2}^\emptyset[I\leftrightarrow \{2n\}\times\mathbb R]\ge \frac{\mathbf P_{\mathbb Z^2}^\emptyset\big[[-n,n]\times[0,8n]\text{ is crossed horizontally in $\widehat\n$} ]}{32}.$$
Furthermore, we find
\begin{align*}\mathbf P_{\mathbb Z^2}^\emptyset[U_n]&\ge \mathbf P_{\mathbb Z^2}^\emptyset[I\leftrightarrow \{n\}\times\mathbb R]-\mathbf P_{\mathbb Z^2}^\emptyset[I\leftrightarrow (\{n\}\times\mathbb R)\setminus J]\\
&\ge \mathbf P_{\mathbb Z^2}^\emptyset[I\leftrightarrow \{n\}\times\mathbb R]-2\phi_{\bbZ^2}[B_n].\end{align*}
We deduce from
 \eqref{eq:27} and the fact that $\varepsilon_n$ tends to 0 that
\begin{equation}\label{eq:aag}
\limsup_{n\rightarrow\infty} \mathbf P_{\mathbb Z^2}^\emptyset[U_n]>0.
\end{equation}
We now wish to bound the probability of $V_n(x,y)$ and $W_n$ in terms of the probability of $U_n$. On the one hand, for each $x$ and $y$, we have that
\begin{equation}\label{eq:gb} \mathbf P_{\mathbb Z^2}^\emptyset[V_n(x,y)] \ge\mathbf P_{\mathbb Z^2}^\emptyset[U_n]-3\max_{k\in[-n,n]}\mathbb P_{\mathbb Z^2}^\emptyset[F_n(k)]\ge \mathbf P_{\mathbb Z^2}^\emptyset[U_n]-3\varepsilon_n.\end{equation}
On the other hand, note that the event $W_n$  includes the event that
\begin{itemize}[noitemsep,nolistsep]
\item[(i)] there exists a continuous path from $I'$ to $J'$ (see Step~1 for definitions) passing below 0,
\item[(ii)] $E_n$ as well as the reflection $\widetilde E_n$ of $E_n$ with respect to the $y$-axis do not hold.
\end{itemize}
Indeed, consider a continuous path $\gamma$ from $I'$ to $J'$ passing below $0$. Let $x$ and $y$ be respectively the first and last intersection points of $\gamma$ with $[-n,n]\times\{n/4\}$ when following $\gamma$ from the left side to the right side of the rectangle. Note that since $E_n$ and $\widetilde E_n$ do not occur, $x$ and $y$ exist and $x$ is connected to $I'$ in $[-n,0]\times\bbR$, while $y$ is connected to $J'$ in $[0,n]\times\bbR$. This observation immediately implies that $\gamma$ contains a subpath $\gamma'$ from $\ell_n$ to $\ell_n$ passing below $0$.

We deduce that
\begin{equation}
  \label{eq:53}
  \mathbf P_{\mathbb Z^2}^\emptyset[W_n] \ge  \mathbf P_{\mathbb Z^2}^\emptyset[U_n]-\phi_{\bbZ^2}[F_n(0)]-2\phi_{\bbZ^2}[E_n]\ge\mathbf P_{\mathbb Z^2}^\emptyset[U_n]-3\varepsilon_n\end{equation}
($\mathbf P_{\mathbb Z^2}^\emptyset[U_n]-\phi_{\bbZ^2}[F_n(0)]$ is a lower bound for the probability of (i) happening). Combining \eqref{eq:aag}, \eqref{eq:gb} and \eqref{eq:53} with the fact that $\varepsilon_n$ tends to 0 implies that
\begin{equation}
\label{eq:599}
\limsup_{n\rightarrow\infty} \min_{x,y\in\ell_n}\mathbf P_{\mathbb Z^2}^\emptyset[V_n(x,y)] \mathbf P_{\mathbb Z^2}^\emptyset[W_n]>0.
\end{equation}

\paragraph{Step 3: Creation of infinitely many circuits in the sum of two currents.} Consider a pair of currents $(\n_1,\n_2)$ with law ${\bf P}^{\emptyset,\emptyset}_{\bbZ^2}$ (with this definition, $\n_1$ and $\n_2$ are two independent sourceless currents on $\mathbb Z^2$). On the event $\{\widehat\n_1\in W_n\}$, define $\Gamma_1$ as the minimal (for an arbitrary ordering on paths) continuous path in $\widehat \n_1\cap S^+_n$ that starts and ends at $\ell_n$, and passes below $0$. Summing over the possible values for $\Gamma_1$, we find\footnote{The sum on $\gamma_1$ in \eqref{eq:19} and  \eqref{eq:29} refers to the sum on possible values for $\Gamma_1$ on the event $\widehat\n_1\in W_n$.}

  \begin{align}
    \label{eq:19}
\sum_{\gamma_1}{\bf P}^{\emptyset,\emptyset}_{\bbZ^2}[\Gamma_1 =\gamma_1]&\ge  {\bf P}^{\emptyset,\emptyset}_{\bbZ^2} [\widehat\n_1 \in W_n]={\bf P}_{\bbZ^2}^\emptyset[W_n].
  \end{align}
  Fix a continuous path $\gamma_1$ for which ${\bf P}^{\emptyset,\emptyset}_{\bbZ^2}[\Gamma_1 =\gamma_1]>0$, and denote by $o(\gamma_1)$ and $e(\gamma_1)$ the first and last vertices of $\gamma_1$ (notice that by definition, these two points lie on $\ell_n$).    When $\Gamma_1=\gamma_1$ and $\widehat\n_2\in V_n(o(\gamma_1),e(\gamma_1))$, we almost created a circuit around the origin in the sense that $\gamma_1$ together with any continuous path from $I'$ to $J$ that passes below $o(\gamma_1)$ and $e(\gamma_1)$ and above $0$ almost create a circuit.

In order to create a true circuit in $\widehat{\n_1+\n_2}$, we apply a local modification argument to $\n_2$. Pick $x_1$ and $x_2$ to be the first and last vertices of $\gamma_1$ within the interaction range $R$ of a path in $\widehat\n_2$ connecting them and passing above the origin.  
Consider the map $T$ that associates with $\n_2$ the configuration $\n_2$ except on edges $\{x,y\}$ with $J_{x,y}>0$ and $x,y$ at a distance less than $R$ of $x_1$ or $x_2$, for which the current is turned to 2. Note that for any $\widehat\n_2\in V_n(o(\gamma_1),e(\gamma_1))$, $\n_1+T(\n_2)$ contains a circuit surrounding the origin. Also, one may easily check that there exists a constant $c=c(\beta,J)>0$  such that
\begin{equation}\label{eq:199}{\bf P}_{\bbZ^2}^\emptyset[T(V_n(o(\gamma_1),e(\gamma_1)))]\ge c{\bf P}_{\bbZ^2}^\emptyset[V_n(o(\gamma_1),e(\gamma_1))].\end{equation}
Therefore,
by \eqref{eq:19}, \eqref{eq:199} and independence, we obtain   \begin{align}
    \label{eq:29}
{\bf P}^{\emptyset,\emptyset}_{\bbZ^2}[H_n]
&\ge \sum_{\gamma_1}{\bf P}^{\emptyset,\emptyset}_{\bbZ^2}[\Gamma_1=\gamma_1,\widehat\n_2\in T(V_n(o(\gamma_1),e(\gamma_1)))]\\
&= \sum_{\gamma_1}{\bf P}_{\bbZ^2}^\emptyset[\Gamma_1=\gamma_1]{\bf P}_{\bbZ^2}^\emptyset[T(V_n(o(\gamma_1),e(\gamma_1)))]\nonumber\\
&\ge  {\bf P}_{\bbZ^2}^\emptyset [ W_n] \min_{x,y\in\ell_n} {\bf P}_{\bbZ^2}^\emptyset [ V_n(x,y)].\nonumber
  \end{align}
Inequality~\eqref{eq:27a} thus follows from this inequality combined with \eqref{eq:599}. This ends the proof of the lemma, and therefore of Theorem~\ref{prop:condition} in the case where $\phi_{\mathbb Z^2}[0\leftrightarrow\infty]=0$.

\subsection{Proof of  Theorem~\ref{prop:condition} in the case $\phi_{\mathbb Z^2,\beta}[0\leftrightarrow\infty]>0$}
\label{sec:proof-case-phiPositive}

Below, we fix $\beta\ge\beta_c$ and drop it from the notation. Also, we assume that $\phi_{\mathbb Z^2,\beta}[0\leftrightarrow\infty]>0$. Recall that this is not expected to be the case when $\beta=\beta_c$ so that the proof of this section is mostly relevant  to the case $\beta>\beta_c$. 

Recall from \cite{AizDumSid15} that one can define an infinite-volume measure ${\bf P}_{\mathbb Z^2}^\emptyset$ for the random (sourceless) current on $\mathbb Z^2$. Furthermore, it is ergodic and the infinite connected component, when it exists, is unique almost surely. Since the high-temperature expansion corresponds to odd currents, we can construct an ergodic infinite-volume measure ${\rm P}_{\mathbb Z^2}^\emptyset$ on random even subgraphs of $\bbZ^2$.

Observe that ${\rm P}_{\mathbb Z^2}^\emptyset$ satisfies a form of finite-energy formula in the sense that  for any finitely supported $\eta_0$ with $\partial\eta_0=\emptyset$, there exists $c(\eta_0)>0$ such that for any $\eta$,
  \begin{equation}
    \label{eq:2}
    {\mathrm P}^\emptyset_{\mathbb Z^2}[\eta\Delta\eta_0]\ge c(\eta_0) {\mathrm P}^\emptyset_{\mathbb Z^2}[\eta].
  \end{equation}

\begin{proof}[Proof of Theorem~\ref{prop:condition}] We prove Theorem~\ref{prop:condition}  by contradiction. Assume for a moment that
\begin{equation} \label{eq:af}{\rm P}^\emptyset_{\mathbb Z^2}  [\exists\text{ infinitely many circuits surrounding $0$ in }\eta]<1.\end{equation}
This implies that there exists $n\ge0$ with
${\rm P}^\emptyset_{\mathbb Z^2}[\exists\text{ a circuit surrounding $\Lambda_n$}]<1.$
In such case, \eqref{eq:2} immediately implies that
\begin{equation}\label{eq:5}
{\rm P}^\emptyset_{\mathbb Z^2}[\exists\text{ no circuit surrounding $0$ in }\eta]=:c>0.\end{equation}
Our assumption $\phi_{\mathbb Z^2}[0\leftrightarrow\infty]>0$ and the fact that the infinite connected component in $\omega$ is unique almost surely \cite[Theorem~4.63]{Gri06} imply that for every $n$,
\begin{equation}
  \label{eq:4}
  \lim_{m\to \infty}\phi_{\mathbb Z^2}[\Lambda_m \text{ is crossed from left to right without intersecting $\Lambda_n$}]=1. 
\end{equation}
Combined with~\eqref{eq:3} in the appendix, we deduce that $\Lambda_n$ is almost surely surrounded by a circuit of $\omega$. Hence, one can always choose a circuit $\gamma=\gamma(\omega)$ in $\omega$ that surrounds the box $\Lambda_n$. Now, if there is no circuit surrounding $0$ in $\eta$, then a parity argument shows that there must exist a circuit surrounding $0$ with diameter larger than $n$ in the configuration $\eta\Delta\gamma$ (notice that the circuit may not surround the box $\Lambda_n$). Since $(\omega,\eta\Delta\gamma)$ has the same law as $(\omega,\eta)$ by Corollary~\ref{rmk:cycle}, we deduce that
    \begin{equation}
      \label{eq:6}
     {\rm P}^\emptyset_{\mathbb Z^2}[\exists\text{ a circuit of diameter larger than $n$ surrounding $0$ in $\eta$}]\ge c/2.
    \end{equation}

The equation above guarantees the existence of a large circuit around $0$, but this circuit may come close to the origin. As stated in the Lemma~\ref{lem:p} below, we will prove that it is not always the case and that the circuit has a positive probability to surround a large box around $0$.  This lemma, combined with the  ergodicity of ${\rm P}_{\mathbb Z^2}^\emptyset$, implies that $0$ is surrounded almost surely by a circuit in $\eta$, which contradicts \eqref{eq:af} and concludes the proof of Theorem~\ref{prop:condition}.\end{proof}

    \begin{lemma}\label{lem:p}
       For every $k\ge 0$, we have
  $ {\rm P}_{\mathbb Z^2}^\emptyset[\exists\text{ a circuit surrounding $\Lambda_{k}$ in $\eta$}]\ge c/4.
   $
    \end{lemma}

\begin{proof}
The strategy of the proof will be different depending on whether there is an infinite connected component in $\eta$ or not.

Let us first assume that ${\rm P}^\emptyset_{\mathbb Z^2} [0\leftrightarrow\infty]=0$. Then, one can fix $n$ large enough such that $\Lambda_k$ is connected to $\partial \Lambda_n$ with probability smaller than $c/4$. If this event does not occur, then a circuit around $0$ of diameter larger than $2n$ necessarily surrounds $\Lambda_k$. For this choice of $n$, \eqref{eq:6} implies that
$$
{\rm P}^\emptyset_{\mathbb Z^2} [\exists\text{ a circuit of diameter at least $2n$ surrounding $\Lambda_k$ in $\eta$}]\ge \frac c2-\frac c4 =\frac c4,
$$
which establishes the claim in this case.

Let us now assume that ${\rm P}^\emptyset_{\mathbb Z^2} [0\leftrightarrow\infty]>0$. To prove the claim in this case, we apply a refined version of the argument used to prove \eqref{eq:6}. The difference here is that we work on an event which ensures that the circuit constructed by ``switching'' the path $\gamma$ always surrounds $\Lambda_k$.

Consider pairs of configurations $(\omega,\eta)$ coupled according to the measure $\bbP$ introduced in Theorem~\ref{thm:connection}. Let $\mathcal E$ be the event that
\begin{itemize}[noitemsep,nolistsep]
\item $0$ is not surrounded by a circuit in $\eta$,
\item in the random-cluster configuration $\omega$, there exists a circuit that surrounds the box $\Lambda_n$,
\item in the configuration $ \eta$ restricted to the annulus $\Lambda_{n}\setminus\Lambda_k$, there exists a unique connected component crossing the annulus $\Lambda_{n}\setminus\Lambda_k$ from inside to outside.
\end{itemize}
One may choose $n=n(k)$ such that $\bbP[\calE]\ge c/2$. Indeed, by \eqref{eq:5}, the first item occurs with probability $c$. As explained below \eqref{eq:4}, the second item occurs almost surely independently of the choice of $n$. Finally, by uniqueness of the infinite connected component in $\eta$ (which can be proved using the standard Burton-Keane argument exactly as for $\widehat\n$ in \cite{AizDumSid15}), one can choose $n$ large enough such that the third item occurs with probability larger than $1-c/2$.

Now, fix $(\omega,\eta)\in \calE$ and choose the minimal (for the lexicographical order induced by $\prec$) circuit $\gamma=\gamma(\omega)$ in $\omega$ that surrounds the box $\Lambda_n$.  Then, $\eta\Delta\gamma$ contains a circuit surrounding $\Lambda_k$ (the existence of a circuit surrounding 0 is obtained by a reasoning spelled out several times already, and the fact that there exists such a circuit surrounding $\Lambda_k$ is guaranteed by the third item of $\calE$). Since $(\omega,\eta\Delta\gamma)$ has the same law as $(\omega,\eta)$ by Corollary~\ref{rmk:cycle}, we find that\begin{equation}
      {\rm P}_{\mathbb Z^2}^\emptyset[\exists \text{   a circuit surrounding $\Lambda_{k}$ in $\eta$}]\ge  \mathbb P[\calE] \ge c/2.
    \end{equation}
 \end{proof}
 \begin{remark}
 In the last proof, it was natural to think that the uniqueness of the infinite connected component in $\eta$ should imply directly the existence of a circuit around $\Lambda_k$. The main difficulty was that ${\rm P}_{\mathbb Z^2}^\emptyset$ does not satisfy the FKG inequality, and that it is therefore not  straightforward to show that the infinite connected component contains a circuit that surrounds $\Lambda_k$.
 \end{remark}
 
\medskip

\section{Pfaffian structure of order-disorder correlation functions} \label{sec:OD}

\subsection{Order-disorder variables and their correlation functions} 

A  sine qua non for the Pfaffian structure of multi-spin corrections is that there is a natural  distinction between even and odd pairings of the variables.  Such a well-defined signature does exist for sites lined along the boundary of a planar region.  However, this fact does not extend to general configurations of $2n$ points in the plane (for instance, the length-minimizing pairing  of four points is not  invariant under the plane's conformal mappings).  Thus, the Pfaffian structure of spin correlation is not natural for the critical correlations of spins in the bulk.
Nevertheless, the Pfaffian relation of the boundary correlation functions \eqref{Pf_boundary} does admit an extension to the bulk: it can be found in the correlation functions of  Ising model's fermionic  order-disorder operators to which we turn in this section. 

Attention to this structure was called in the work of  Kadanoff-Ceva~\cite{KadCev71}.  They pointed out that disorder operators act as anti-commuting spinors in their effect on the functional sum yielding  the partition function.  They also stressed that this is in line with Kaufmann's algebraic approach~\cite{Kau49} to  Onsager's solution of the model on $\Z^2$.   

Still, even for fermions, the Pfaffian structure is rather exceptional since it implies that all the higher correlations can be algebraically determined from the two-point function.   Examples where that happens  are found in  the thermal equilibrium states of systems of \emph{non-interacting} fermions.   For the nearest-neighbor Ising model on $\Z^2$, such a representation was found early on in the fermionic representation of the transfer matrix of  Schultz-Mattis-Lieb \cite{SchMatLie64}.   It took a bit longer to realize that Pfaffian correlations of these particular operators are a feature of {\em all} planar Ising models.   Traces of related statements, though perhaps not quite in this explicit form, were made in the literature  ever since  the observation by Hurst and Green that the Ising model's partition function itself is given by the Pfaffian of the Kasteleyn matrix~\cite{HurGre60}.   
In~\cite{McCPerWu81} these were  linked with non-linear difference equations obeyed by the correlation functions of various combinations of order and disorder variables.   Somewhat more specific statements on Pfaffian correlation functions 
from a combinatorial perspective can be found in the  recent papers \cite{AW17,CheCimKas16,Lis16}.  Our goal in this section is to present an elementary derivation of the Pfaffian structure of  the correlations of paired order-disorder operators from the random current perspective.

The \emph{disorder operators} $\mu_{\ell}$ are associated with   lines $\ell$ drawn in the plane in such a way that they avoid the vertices of $\V(\G)$ and cross\footnote{In the whole section, we consider the standard notion of crossing used in topology: a point at the intersection of two curves is a {\em crossing} if the first curve passes from one side of the second curve to the other.} only finitely many times the edges of $E(\G)$ (cf.~Figure~\ref{fig:ord_disord}).   Associated with $\mu_\ell$ is the change in the Hamiltonian 
${\bf H}_\G \, \mapsto R_{\ell} {\bf H}_\G$  in which each edge crossing $\ell$ signifies a sign flip 
of the coupling $ J_e \mapsto -J_e$ of the crossed edge $ e \in \E(\G) $.

Correspondingly,  the correlation function involving a collection of disorder variables $(\mu_{\ell_j})_{j\le n}$ and a function $F:\{\pm1\}^{\V(\G)}\rightarrow\mathbb C$  is defined by
\be  \label{eq:def_tau}
\big\langle  F  \prod_{j=1}^n \mu_{ \ell_j}  \big\rangle_{\G , \beta}
\ := \
\frac{1}{2^{|V(\G)|}Z(\G,\beta)}\sum_{\sigma\in\{\pm1\}^{\V(\G)}}   F(\sigma)   \  \exp\Big(-\beta \, \big[(\prod_{j=1}^n R_{\ell_j} ) H_\G\big](\sigma) \Big)\, .
\ee

One may note the partial homotopy invariance in the dependence of the  disorder operator $\mu_\ell$   on the line: for lines  with the same end points,  $\mu_\ell$ and $\mu_{\ell'}$ differ by just a gauge transformation under which the spins are flipped in a region enclosed by the loop\footnote{The region {\em enclosed by a loop} is the set of points of the complement of the loop in $\bbR^2$ that can be reached from infinity by crossing the loop an odd number of times.} obtained by concatenating the two.

To simplify the discussion, it is convenient to adopt a convention in which the disorder operators are associated with sites.  For that, one of the end-points of the disorder variables will always be placed within a specified cell, to which we refer as the grand-central.  
More precisely, consider a collection of paired order-disorder variables
\be 
 \tau_j :=  \sigma_{x_j} \cdot \mu_{\ell_j} 
\ee  
in which $x_j \in V(\G)$ and the corresponding line $\ell_j $ links the grand-central with a site $x'_j$ lying within a cell whose border includes $x_j$.
\begin{figure}[h]
\begin{center}
\includegraphics [width = 0.6 \textwidth]{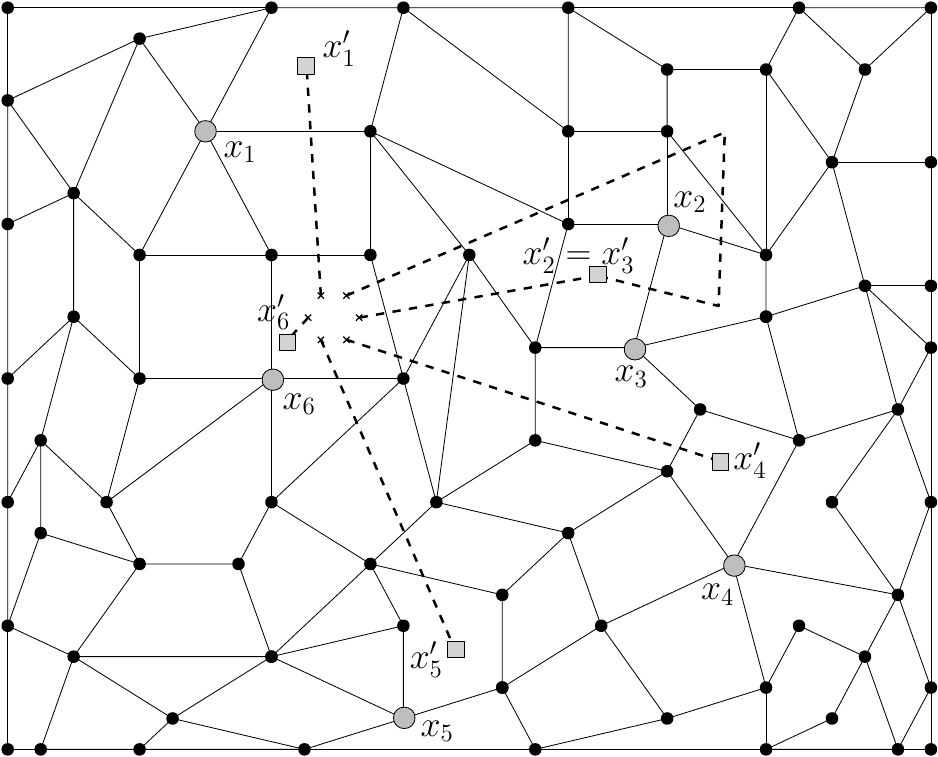}
\caption{Paired order-disorder variables.  Disorder lines are depicted in dots, and  spin sites $x_1,\dots,x_6$ in gray bullets.  Each disorder line has one end in the common ``grand-central'' cell.  The lines do not intersect and are labeled  $ \ell_1, \dots,\ell_6 $ cyclically.
\label{fig:ord_disord}
}
\end{center}
\end{figure}

The following theorem extends Theorem~\ref{thm:pf_boundary} to order-disorder variables.

\begin{theorem}[Pfaffian structure for order-disorder correlation functions]\label{thm:Pf_OD}
Fix a planar graph $G$, arbitrary nearest-neighbor couplings $J$, and $\beta\ge0$. Let $(\tau_1,\dots,\tau_n)$ be a collection of paired order-disorder variables for which the disorder lines are disjoint (except possibly at their endpoints) and  are labeled in a cyclic order. Then, we have
\be  \label{eq:Pf_gen}
 \langle \tau_{1}\cdots\tau_{2n}\rangle_{\G,\beta} = 
\Pf_n \big( [\langle \tau_i\tau_j\rangle_{\G,\beta}]_{1\le i<j\le 2n} \big) .
\,
\ee
\end{theorem}

As explained above, the assumptions made here on the lines $\ell_j$ do not amount to  loss of generality, since any deformation of the lines is simply compensated by a   gauge transformation. 

Theorem~\ref{thm:Pf_OD}  includes Theorem~\ref{thm:pf_boundary}  as a special case:   For any collection of spins on the boundary of a single face, one may take the grand central to be this face, and all the disorder lines $\ell_j$  confined to it (as depicted by  $x'_6$ in Figure~\ref{fig:ord_disord}).  In this case, the disorder operators act as an identity, and 
\eqref{eq:Pf_gen} reduces to \eqref{Pf_boundary}.

Another noteworthy example is the case where the vertices  $ \{ x_{2j-1},x_{2j} \}$ are pairwise adjacent,
with the disorder lines chosen so that the actions of $\ell_{2j-1}$  and $\ell_{2j}$  overlap (i.e.~the two traverse the same cells) and thus cancel each other.   In that case, the pairwise product of two ``order-disorder''  operators reduces to the  \emph{energy operator}   
$\tau_{2j-1} \tau_{2j}  =  \sigma_{x_{2j-1}} \sigma_{x_{2j}}$, so that
\be
 \big\langle  \prod_{j=1}^{2n} \tau_j  \big\rangle_{\G, \beta}  =  \big\langle  \prod_{j=1}^{n} \sigma_{x_{2j-1}} \sigma_{x_{2j}} \big\rangle_{\G, \beta} \, .
\ee

\subsection{Proof of Theorem~\ref{thm:Pf_OD}}\label{sec:3}

Like Theorem~\ref{thm:pf_boundary},  Theorem~\ref{thm:Pf_OD} will be established here by induction on $n\ge1$.  Thanks to \eqref{eq:PFcrit}, the key step is to show that $\langle\tau_1\cdots\tau_{2n}\rangle_{\G,\beta} =  Q(x_1,\dots, x_{2n}) $ with 
 \be \label{eq:suff3}
 Q(x_1,\dots, x_{2n})   := 
  \sum_{j=2}^{2n} (-1)^{j+1}  \,  
\langle \prod_{\substack{1\le i\le 2n \\ i \notin  \{1,j\} }}\tau_i \rangle_{\G,\beta} \,   \langle \tau_1 \tau_j \rangle_{\G,\beta}  \, .
\ee

In the following discussion we identify disorder lines $\ell_j$ with the multiset composed of edges of $E(G)$ taken with the multiplicity of the number of times they are crossed\footnote{Note that we are counting the number of crossings (in the topological sense) between continuous curves, which can be smaller than the number of intersections.}  by the disorder line. From now on, we always make this identification. Also, since multisets can be interpreted as integer-valued functions, sums and differences of disorder lines can be defined as sums and differences of the corresponding multisets.

Since the disorder operators simply reverse the couplings, whose contribution to the weight $w(\n)$ includes the factor $J_{x,y}^{\n(x,y)}$, the random current expression for the expectation value of the product of order-disorder operators is given by
\be\label{RC_withDisorder}
 \langle \prod_{i=1}^{k}  \sigma_{x_i} \,  \prod_{j=1}^{n}  \mu_{\ell_j} \rangle_{\G,\beta}
= \frac{1}{ Z(G,\beta)}  \sum_{\n:\partial \n = \{x_1,\dots,x_k\} }   w(\n)  (-1)^{(\n | \mathcal{L}) } \, , 
 \ee
where $\mathcal L:=\ell_1+\dots+\ell_n$ and, for a multiset $\mathcal E$ of edges in $E(G)$, $(\n | \mathcal E) $  is  given by (in the sum below, $\{x,y\}$ appears as many times as it appears in the multiset $\mathcal E$)
\be 
 (\n | \mathcal{E}) := \sum_{\{x,y\}\in\mathcal E}\n(x,y).
\ee 
Using the abbreviation $ X := \{ x_1,\dots, x_{2n} \} $, we obtain that for each $j\in \{2,\dots,2n\}$, the corresponding term in \eqref{eq:suff3} takes the form (we use $(\calN|\mathcal E):=(\n|\calE)$ for $\calN$ associated with $\n$)
\begin{align}\label{eq:onetermod}
Z(G,\beta)^2  \ &  \langle \prod_{\substack{1\le i\le 2n \\ i \notin  \{1,j\} }}\tau_i \rangle_{\G,\beta}  \ \langle \tau_1 \tau_j \rangle_{\G,\beta}   = \sum_{\substack{\n_1: \partial \n_1 = X \backslash \{x_1, x_j\} \\ \n_2:\partial \n_2 = \{x_1, x_j\} }}   w(\n_1) w(\n_2)  \,    (-1)^{(\n_1 | \mathcal{L}-(\ell_1+\ell_j)) }      (-1)^{(\n_2 | \ell_1+ \ell_j ) }   \notag \\
& =   \sum_{\mathcal{M} : \partial \mathcal{M} = X } w(\mathcal{M}) \sum_{\substack{\mathcal{N} \subset \mathcal{M} \\ \partial \mathcal{N} =\{ x_1,x_j\} } } (-1)^{(\mathcal{N} | \mathcal{L} ) } (-1)^{(\mathcal{M}| \ell_1+\ell_j) }  \notag \\
& = \sum_{\mathcal{M} : \partial \mathcal{M} = X } w(\mathcal{M}) \;  \mathbb I[x_1 \stackrel{\mathcal{M} }   {\longleftrightarrow} x_j ]  \sum_{\substack{\mathcal{N} \subset \mathcal{M} \\ \partial \mathcal{N} =\{ x_1,x_j\} } }  (-1)^{(\mathcal{M}|\ell_1+\ell_j)} (-1)^{(\mathcal{N} \Delta \mathcal{K}_j | \mathcal{L})} \notag \\
& =   \sum_{\substack{\n_1: \partial \n_1 = X  \\ \n_2:\partial \n_2 = \emptyset }}   w(\n_1)   w(\n_2)  (-1)^{ (\n_1 | \mathcal{L} ) }  (-1)^{(\n_1 + \n_2 | \ell_1+\ell_j )}  (-1)^{ (\mathbf k_{j} | \mathcal{L} ) } \; \mathbb I[x_1 \stackrel{\n_1+\n_2}   {\longleftrightarrow} x_j ] \, .
\end{align}  
Here, the second line follows from~\eqref{eq:key} and the additivity property $(\n|\calE_1+\calE_2)=(\n|\calE_1)+(\n|\calE_2)$.
The third inequality invokes the switching principle~\eqref{eq:switch} in which $ \mathcal{K}_j $ represents an arbitrary subgraph of $ \mathcal{M} $ with $  \partial \mathcal{K}_j = \{ x_1,x_j \} $. The latter is identified with $ \mathbf k_{j} \leq \n_1 + \n_2 $ for which  $ \partial \mathbf k_{j} = \{ x_1,x_j \} $.  The last line uses \eqref{eq:key} and the additivity property $(\n+\m|\calE)=(\n|\calE)+(\m|\calE)$.   

It is worth noting that the  switching which leads to \eqref{eq:onetermod} allows an arbitrary  choice of the $\mathbf k_j \le \mathbf n_1+\mathbf n_2$ with $\partial \mathbf k_j=\{x_1,x_j\}$.   It is explained by the fact that currents $\n_1+\n_2$ contributing to the last sum of \eqref{eq:onetermod} may be limited by the following consistency condition, which has been noted to be  of  relevance in discussions of Ising models with couplings of mixed signs on arbitrary graphs,  regardless of planarity.   

\begin{definition} [The $S_{\calE}$-condition]
Given a  graph $\G$ and a multiset  ${\mathcal E}$ of edges in $E(\G)$, we say that a currents configuration 
$\m$ satisfies the  \emph{$S_{\mathcal E}$-condition} (written $\m\in S_{\cal E}$) if  each loop $\gamma$ supported on the corresponding multigraph $\mathcal M$  satisfies $(-1)^{(\gamma|\calE)}=1$. 
\end{definition}  
Equivalently, we could have said that the loop $\gamma$ intersects an even number of times the edges appearing an odd number of times in $\calE$. 
The following lemma will play a role in the proof of Theorem~\ref{thm:Pf_OD}.

\begin{lemma}  \label{S_switch}  For any graph $\G$, any $A,B \subset \V(\G)$,  any two multisets $\mathcal E_1,\mathcal E_2$ of edges in $\E(\G)$,
and any function $f$ on currents, we have 
\begin{multline}\sum_{\substack{\n_1:\partial \n_1=A \\ \n_2:\partial \n_2=B }} f(\n_1+\n_2)\,
(-1)^{(\n_1 |\mathcal E_1 )}w_\beta(\n_1) (-1)^{(\n_2 | \mathcal E_2 )}
w_\beta(\n_2)\\
=\ \sum_{\substack{\n_1:\partial \n_1=A \\ \n_2:\partial \n_2=B}}f(\n_1+\n_2)\,
(-1)^{(\n_1 | \mathcal E_1 )}
w_\beta(\n_1) (-1)^{(\n_2 | \mathcal E_2 )}w_\beta(\n_2)\,  \mathbb I\left[\n_1+\n_2\in S_{\mathcal E_1+\mathcal E_2}  \right] \, .
\end{multline}
\end{lemma}
\begin{proof} 
We split the sum on the left-hand side into two sums over currents with $ \n_1 + \n_2 \in S_{\mathcal E_1+\mathcal E_2}   $ on one hand and with $ \n_1 + \n_2 \notin S_{\mathcal E_1+\mathcal E_2}$ on the other hand, and prove that the latter sum is zero. 
Indeed, we get, in the notation of~\eqref{eq:key}, 
\begin{align*} 
\sum_{\substack{\n_1:\partial \n_1=A \\ \n_2:\partial \n_2=B}}f(\n_1+\n_2)\,&
(-1)^{(\n_1 | \mathcal E_1 )}
w_\beta(\n_1) (-1)^{(\n_2 | \mathcal E_2 )}w_\beta(\n_2)\,  \mathbb I\left[\n_1+\n_2\not \in S_{\mathcal E_1+\mathcal E_2}  \right] \\ 
&= \sum_{\m :\, \partial \m =A\Delta B} f(\m) w_\beta(\m)  \mathbb I\left[\m  \not \in S_{\mathcal E_1+\mathcal E_2}  \right] \sum_{\substack{ \mathcal N \subset \mathcal M \\ \partial \mathcal N = A} }
(-1)^{(\mathcal N  | \mathcal E_1 )}
  (-1)^{(\mathcal M  - \mathcal N  | \mathcal E_2 )}  = 0\,.
\end{align*} 
The last equality follows from the observation that when the $ S_{\mathcal E_1+\mathcal E_2}$-condition is violated, there exists a loop $\gamma \subset \mathcal M$ such that the summands change sign under the change of variable $\mathcal N \mapsto \mathcal N \Delta \gamma$.  
\end{proof}

In the present context, ${\mathcal E} \subset   E(\G)$  is given by $\mathcal L(=\ell_1+\dots+\ell_n)$. 
In this case, the $S_{\mathcal L}$-condition  is satisfied if and only if each loop supported on $\mathcal M$ encloses an even number of endpoints of the disorder lines. 
The above considerations allow one to reduce the proof of Theorem~\ref{thm:Pf_OD} to a simple topological argument.  

\begin{proof}[Proof of Theorem~\ref{thm:Pf_OD}]
Combining~\eqref{eq:onetermod} and Lemma~\ref{S_switch}, we may write the quantity $Q(x_1,\dots, x_{2n}) $ defined in \eqref{eq:suff3} as   
$$
Q(x_1,\dots, x_{2n})    =   \frac{1}{Z(G,\beta)^2 }   \sum_{\substack{\n_1: \partial \n_1 = \{x_1,\dots,x_{2n}\}  \\ \n_2: \partial \n_2 = \emptyset }}   w_\beta(\n_1)   w_\beta(\n_2)  (-1)^{ (n_1 | \mathcal{L} ) } \, R(\n_1 + \n_2 )
$$
with
$$ 
R( \m ) :=  \mathbb I\left[\m\in S_{{\mathcal L}} \right]  \,   \sum_{j=2}^{2n} (-1)^{j}  \,  
 (-1)^{(\m | \ell_1 +\ell_j )}\,(-1)^{ (\gamma_{j} | \mathcal{L} ) } \, 
\mathbb I[x_1 \stackrel{\m}{\longleftrightarrow} x_j ] \, , 
$$
where $\gamma_j$ are edge-disjoint paths  supported on $\mathcal M$ linking $x_j$ with $x_1$.   For each $j$, such a path exists under the imposed condition  $x_1 \stackrel{\m}{\longleftrightarrow} x_j $.  
The selection of the path $\gamma_j$  is not unique, but  under the $ S_{{\mathcal L}}$-condition, all choices yield the same parity factor $ (-1)^{ (\gamma_{j} | \mathcal{L} ) }$.  

The proof is completed with the help of Lemma~\ref{lem:topology} below together with the representation formula~\eqref{RC_withDisorder} for $\langle\tau_1\cdots\tau_{2n}\rangle_{\G,\beta}$.
\end{proof}

\begin{lemma}\label{lem:topology}   
Fix a planar  graph $\G$ and a collection of paired order-disorder operators as in Theorem~\ref{thm:Pf_OD}.  Then, for any  currents configuration $\m$ with $\partial {\m} = \{x_1,\dots,x_{2n}\}$ satisfying the corresponding 
$S_{\mathcal L}$-condition,  we have
\be  \label{top_claim1}
R( \m )  =  1 \, .
\ee
\end{lemma}
\begin{proof}
In this paragraph, we see disorder lines $\ell_i$ and paths $\gamma_j\subset\calM$ as continuous curves in the plane. Consider the loop $\calC$ (in $\bbR^2$) obtained by concatenating
\begin{itemize}[noitemsep,nolistsep]
\item the disorder line $ \ell_1$, 
\item a segment between $x'_1$ and $x_1$ (staying in the face of $x'_1$), 
\item the path $\gamma_j$, 
\item a segment between $x_j$ and $x'_j$ (staying in the face of $x'_j$), 
\item the disorder line  $\ell_j$, and finally
\item a segment connecting the endpoints of $\ell_1$ to $\ell_j$ within the grand central. 
\end{itemize}
This operation creates a loop out of two disorder lines and the path $\gamma_j$ in $\calM$ pairing $x_1$ and $x_j$. In a similar fashion, one can see the remaining disorder lines and the edges of $\calM\setminus \gamma_j$ as a family of loops by adding paths pairing the $x_i$ with $x'_i$ for every $i\in\{1,\dots,2n\}\setminus\{1,j\}$, and paths pairing the endpoints of the disorder lines in the grand central in the natural order.

Since the parity of the number of crossings between  $\calC$ and these other loops is even (as it is for any intersections between collections of loops in the plane), we deduce that the parity of the number of crossings between $\calC$ and the union of the disorder lines and the edges of $\calM\setminus\gamma_j$ is of the same parity as $j$.
Also, note that $\ell_1+\ell_j$ does not intersect $\cal L-\ell_1-\ell_j$ since the disorder lines are avoiding each other. We deduce from all these facts, using the additivity property of $(\n|\calE)$ again, that\begin{align*} (-1)^{j} \,  (-1)^{(\m | \ell_1+ \ell_j )}  (-1)^{ (\gamma_{j} | \mathcal{L} ) }&= \  (-1)^{(\ell_1+\ell_j | \m-\gamma_j )+(\gamma_j | \mathcal{L}-\ell_1-\ell_j)+(\gamma_j | \m- \gamma_j )}(-1)^{(\m | \ell_1+ \ell_j )}(-1)^{ (\gamma_{j} | \mathcal{L} ) }  \notag \\
&= \ (-1)^{(\gamma_j |  \m- \gamma_j )}  \, . 
\end{align*}
The claim made in \eqref{top_claim1}  is, therefore, equivalent to the statement that, under the $S_\calL$-condition,
\be \label{j_sum} 
1- \sum_{j=2}^{2n} (-1)^{(\gamma_j |  \m- \gamma_j )} \mathbb I[x_1 \stackrel{\m}{\longleftrightarrow} x_j ]   = 0 \, . 
\ee 

To prove that, notice that the complement in the plane of the set of lines of $\m$ naturally decomposes  into connected components (each of those points  are connected by a line which does not cross an edge with $\m(x,y) >0$). 
Now, consider the  clustering of  $\{ x_j \}$ according to the connected components of the dual sites $x'_1,\dots,x'_{2n}$ (which, by assumption, cannot lie on the partitioning lines). 

By the $S_\calL$-condition, each connected component contains an even number of sites. 
The claim thus follows from the fact that in each connected component, the signs $(-1)^{(\gamma_j |  \m- \gamma_j )}$ of those sources contributing on the left-hand side of~\eqref{j_sum} alternate when following the 
boundary of the connected component; see Fig.~\ref{fig:alt_signs}. \end{proof}

\begin{figure}[h]
\begin{center}
\includegraphics [angle=270 , width = 0.5 \textwidth]{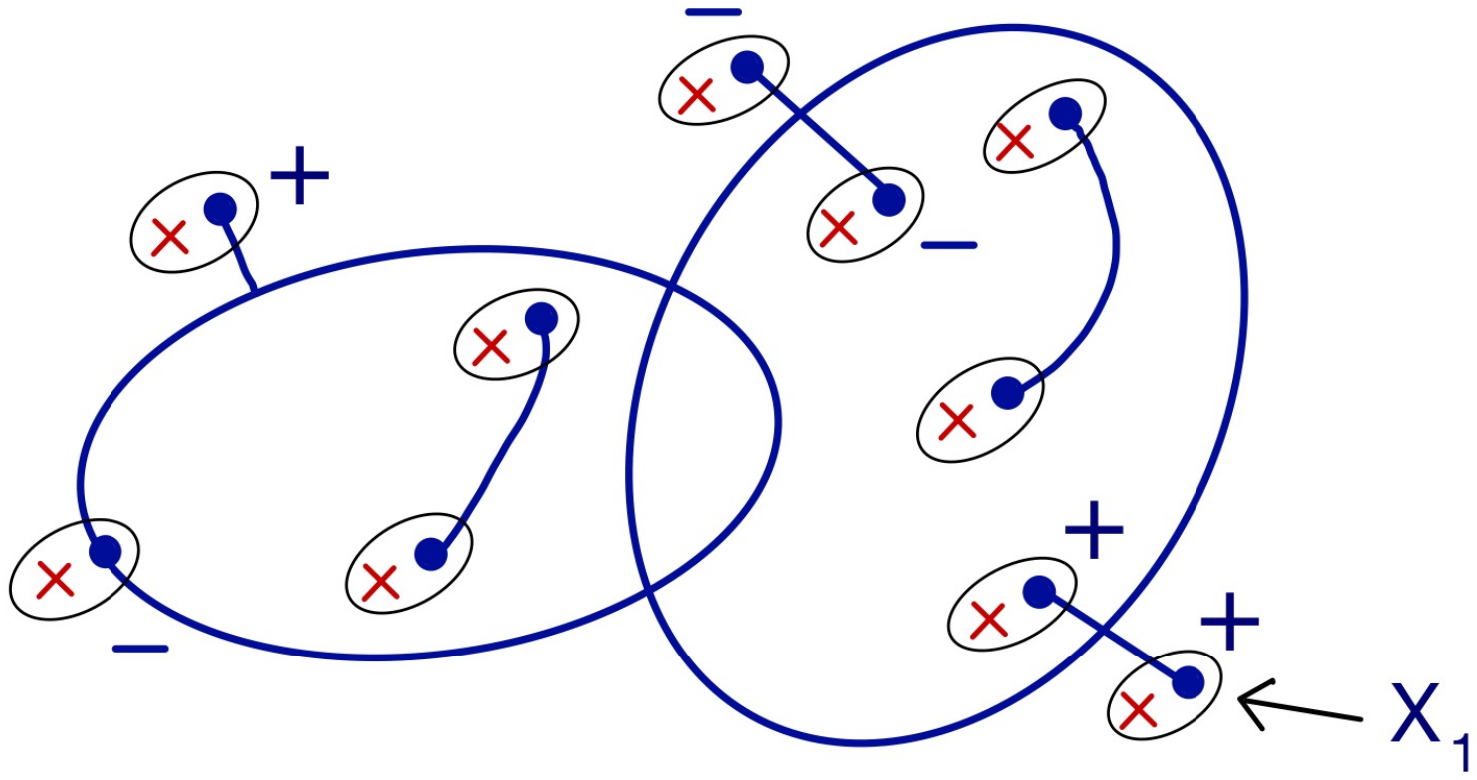} \caption{The relative signs associated with the pairing of sources by a currents configuration for which  the $S_\calL$-condition holds.  The currents $\m$,  depicted here schematically by lines,  partition the plane into connected components.   For sites $x_j$ with $x'_j$  within a given  connected component, the signs of \eqref{j_sum}  
alternate   in the order that $x_j$ are connected to the region's boundary (we adopted the convention that  the sign given to $x_1$ is $ +1 $). 
Sites not connected to $x_1$ make no contribution to the sum. 
Given the currents configuration, the relative signs in the sum on the left-hand side of \eqref{j_sum}  are not affected by the location of the  disorder lines and the choice of the $\gamma_j$.  These are, therefore, omitted here.  
\label{fig:alt_signs}
}
\end{center}
\end{figure}

\begin{remark}Note that in the above argument, the loops of the random currents play a role in partitioning the plane into domains and in ordering the sites connected to $x_1$ in an $\m$-dependent cyclic order.  Within each cluster, the relevant cancellation occurs  by a mechanism similar to that which we encountered earlier in the discussion of the boundary spin variables. \end{remark}

Let us end this section by noting that through a combination of the methods used in this section with those used in the proof of Theorem~\ref{thm:Pf_finite_range}, one may deduce the following statement.  

\begin{corollary}   
Let $J$ be a set of coupling constants for an Ising model over the  upper half-plane $\bbH$ which are:  
\begin{itemize}[noitemsep]
\item[(i)] ferromagnetic, i.e.~that $J_{x,y}\ge0$ for every $ x,y\in\mathbb \bbH$, 

\item[(ii)]  translation invariant, i.e.~that $J_{x,y}=J(x-y)$, and 

\item[(iii)]  finite-range and such that the associated graph is connected,
 
\item[(iv)] reflection invariant: $J_{0,(a,b)}=J_{0,(-a,b)}=J_{0,(a,-b)}=J_{0,(b,a)}$ for every $a,b\in \bbZ$.   \end{itemize}
 Then, the $n$-point  correlation functions of the order-disorder pairs $\{(x_j,x'_j)\}$ (defined as in Theorem~\ref{thm:Pf_OD}, but not restricted to the boundary), satisfy 
 \be \label{OD_bulk}
 \big| \langle \tau_{1}\cdots\tau_{2n}\rangle_{\Z^2,\beta_c} \ - \  \Pf_n \big( [\langle \tau_i\tau_j\rangle_{\Z^2,\beta_c}]_{1\le i<j\le 2n} \big) \big| \  \leq \  
  o(1)  \,\langle \sigma_{x_1}\cdots\sigma_{x_{2n}}\rangle_{\Z^2,\beta_c} ,
\,
\ee
where  $o(1)$ is a function of the points $x_1,\dots,x_{2n}$   which tends to zero for configuration sequences with  $\min\{|x_i-x_j|,1\le i<j\le 2n\} \to \infty$. 
\end{corollary} 
However,  the arguments contained in the previous sections also indicate that at the critical point, 
$$
\big| \langle \tau_{1}\cdots\tau_{2n}\rangle_{\Z^2,\beta_c} \big|  \ll \   \langle \sigma_{x_1}\cdots\sigma_{x_{2n}}\rangle_{\Z^2,\beta_c} \, .
$$
Hence, the above bound is short of the statement that the Pfaffian structure of the bulk order-disorder correlation functions is universally shared by finite-range interactions.

\section{Discussion}
 
As mentioned before in this article, we did not manage to prove the fractal properties of the random current. In fact, we did not even prove that the phase transition is continuous (we simply showed Theorem~\ref{thm:Pf_finite_range} under both scenarios --  continuous and discontinuous).  Let us, therefore, list here a number of problems and directions in which the analysis could perhaps be extended.  

\begin{question}[continuity]
Show that the order parameter of the finite-range Ising model is continuous, i.e.~vanishes, at the critical temperature.   
\end{question}

For the nearest-neighbor interactions, such a continuity result was proven for the square lattice  by Yang~\cite{Yan52}, and for all planar lattices in the recent work~\cite{DumMar17}.  In higher dimensions, the continuity was proved for the nearest neighbor models on  $\Z^d$ for all $d>2$: for $d\ge4$ in \cite{AizFer86}, with also some results for long-range interactions, and $d=3$ in \cite{AizDumSid15}, with the help of some of  the techniques which are mentioned in Section~\ref{sec:RC}.

The next step would be to prove mixing properties and a Russo-Seymour-Welsh type statement for random currents. We believe that such a statement would have far-reaching consequences for the study of universality in the two-dimensional Ising model. Recall that ${\bf P}^\emptyset_{\bbZ^2,\beta_c}$ denotes the (single) random current measure in the full plane at $\beta_c$.

\begin{question}[polynomial mixing of currents]
Show that there exist $c,C>0$ such that for any two events $A$ and $B$ depending on edges in the box of size $n$ and outside the box of size $N\ge n$ respectively, we have
$$|{\bf P}_{\bbZ^2,\beta_c}^\emptyset[A\cap B]-{\bf P}^\emptyset_{\bbZ^2,\beta_c}[A]{\bf P}^\emptyset_{\bbZ^2,\beta_c}[B]|\le C\big(\tfrac nN\big)^c\cdot{\bf P}^\emptyset_{\bbZ^2,\beta_c}[A]{\bf P}^\emptyset_{\bbZ^2,\beta_c}[B].$$ 
\end{question}

This polynomial mixing is closely related to the following crossing estimate, which we derive in this article for infinitely many values of $n$ only (see Definition~\ref{def:crossed} for a definition of being crossed).

\begin{question}[crossing estimates of currents]
Show that there exists $c>0$ such that for every $n$, we have
$${\bf P}^\emptyset_{\bbZ^2,\beta_c}[\text{the box of size $n\times 2n$ is crossed in the hard direction by a path of positive current}]\ge c.$$ 
\end{question}

Last but not least, let us mention that the most natural question left open by this article is to replace the $o(1)\langle \sigma_{x_1}\cdots\sigma_{x_{2n}}\rangle_{\bbZ^2,\beta_c}$ in \eqref{OD_bulk} by $o(1)|\langle\tau_1\cdots\tau_{2n}\rangle_{\bbZ^2,\beta_c}|$.
\begin{question}
Extend the Pfaffian relation of Theorem~\ref{thm:Pf_OD} to finite-range Ising models on $\bbZ^2$.
\end{question}

\appendix 
\bigskip
\noindent {\Large \bf Appendix}

\section{Extension to finite-range models of the Messager Miracle-Sol\'e inequality}

In Sections~\ref{sec:4} and \ref{sec:5}, use is made of a generalization of the Messager-Miracle-Sol\'e inequality \cite{Heg77,MesMir77,Sch77} to finite-range interactions.  The original inequality states that for Ising models on graphs with a reflection symmetry $\cal R$, and correspondingly symmetric nearest-neighbor interactions,  the correlation function $\langle \sigma_A \sigma_B\rangle $ at sites which are separated by a reflection plane, can only increase when $B$ is replaced by its image, $\mathcal{R}(B)$, on the side of the reflecting plane containing $A$, i.e.
\be 
\langle \sigma_A \sigma_B\rangle  \ \leq \ \langle \sigma_A \sigma_{\mathcal R(B)}\rangle. 
\ee   The example of principal interest here is the model on $\bbZ^2$ or on the half-plane $\mathbb H$, with the reflections $\mathcal{R}: (a, L - b) \mapsto (a, L + b)$  for some fixed $L$.   Following is the  statement which is used here.


\begin{theorem}\label{thm:mono}  
Let $\G$ be the graph $\Z^d$, or the corresponding half spaces $\Z\times  \Z_+^{d-1}$, and 
$J$ a set of coupling constants for an Ising model which are:  

\begin{enumerate}[noitemsep]
\item[(i)]  ferromagnetic ($J_{x,y}\ge0$ $\forall  x,y\in \G$), 
\item[(ii)]  translation invariant, 
\item[(iii)] finite-range and such that the associated graph is connected, 
\item[(iv)] reflection invariant:  $ J_{A,B} = J_{\mathcal{R}_L(A), \mathcal{R}_L(B)}$ for any $ L \in \mathbb{Z} $, where $\mathcal{R}_L $ stands for the reflection 
$\mathcal{R}_L(L-u,\tilde x):=(L+u,\tilde x)$ and we write 
$x  \in \Z^d$ as $(x_1,\tilde x)$.
\end{enumerate} 
Then, there exists  $ C>0$ such that for every point $x=(x_1,\tilde x) \in \G$  with 
$x_1 >L$, we have
\be\label{eq:MMSh}
 \langle\sigma_0\sigma_x\rangle_{\G,\beta}\ \le \ C \langle\sigma_0\sigma_{\mathcal{R}_L(x)}\rangle_{\G,\beta}.
\ee 
\end{theorem}
The theorem and its proof naturally extend  to reflections of $\Z^d$  with respect to other symmetry hyperplanes.  In two dimensions, this means reflections with respect to lines passing through lattice sites which are   parallel to one of the diagonals $x_1\pm x_2 =0 $. 

Our proof involves  a random current reflection principle, which also provides an alternative derivation of the original Messager-Miracle-Sol\'e theorem for the nearest-neighbor interactions (for which $C(\beta) =1$).   
The reflection part of the argument can be summarized by the following statement.

\begin{lemma}\label{lem:extendMMS}
In the situation of Theorem~\ref{thm:mono}, for all $ x = (x_1,\tilde x) $ with $ x_1 > L $, we have
\be \label{eq:MMSR}
\langle\sigma_0\sigma_x\rangle_{G,\beta}\le \sum_{j=0}^{R-1} \langle\sigma_0\sigma_{\mathcal{R}_{L+j}(x)}\rangle_{G,\beta} \, .
\ee
\end{lemma}

\begin{proof} 
Given a hyperplane $ \mathcal{P}_L := \{ (L,\tilde x) : \, \tilde x \in \mathbb{Z}_+^{d-1} \} $, which passes through lattice sites, let us partition the edge-set $E(\G)$ into three disjoint sets 
\begin{align*}
E_-(\mathcal{P}_L) &:=\{\{u,v\}  : \text{none of the sites is to the right of $\mathcal{P}_L$ and at least one is on the left} \}\, , \notag \\
E_+(\mathcal{P}_L) &:=\{\{u,v\} : \text{none of the sites is to the left of $\mathcal{P}_L$ and at least one is on the right} \}\, ,  \notag \\
E_{0}(\mathcal{P}_L) &:= E(G) \backslash ( E_-(\mathcal{P}_L) \cup E_+(\mathcal{P}_L)) \,  . 
\end{align*}

In  the case of the nearest-neighbor interaction,  any path (along edges with $J_{u,v} \neq 0$) starting on the right of  $\mathcal{P}_L$ and ending to the left of it will at some point hit the plane $\calP_L$.  Let us consider this case first. 

Assume for a moment that $G$ is finite and symmetric under the reflection $\calR_L$.
Write the correlation in terms of  random currents $\n$  (cf.~\eqref{eq:spin correlations}). Then,  decompose the currents $\n$ into $ \n_-$, $\n_+$, and $\n_0$ according to the above decomposition of $E(\G)$ -- $\n_\alpha$ is the restriction of $\n$ to edges in $E_\alpha(\calP_L)$.   For the purpose of applying the switching principle,  consider  the multigraph $\calM$ obtained by taking the union of the multigraph $\calN_-$ associated with $\n_-$ and the reflection with respect to $\mathcal{P}_L$ of the multigraph $\calN_+$ associated with $\n_+$. 

Let  $x =\{x_1, \tilde x\}$ be a vertex of $\G$ with $x_1 >  L$.  Since any current contributing to $ \langle\sigma_0\sigma_{x} \rangle  $ has sources $ \partial \n = \{ 0, x \} $, the graph $\calM$ must contain a path  
from $\mathcal{R}_L(x)$ to a site on $\mathcal{P}_L$, an event which we denote by $\mathcal{E}_L(\calM)$.  Let us denote by $\calK_L (\calM)$ the first  such path according to some arbitrary ordering of the finite number of possibilities.  Using the switching principle with respect to $\calK _L(\calM)$, we find that
\begin{align}\label{eq:arguMMS}
\sum_{\n:\partial \n = \{ 0, x \} } w(\n) &  =   \sum_{\n_{0} } w(\n_{0}) \mkern-20mu \sum_{\substack{ \n_-,\n_+: \\ \partial(\n_-+\n_+)=\{0,x\}\Delta \partial \n_{0} }}w(\n_-)w(\n_+) \;  \mathbbm I[\mathcal{E}_L(\calM)] \notag \\
&= \sum_{\n_{0} } w(\n_{0}) \mkern-20mu \sum_{\substack{ \n_-,\n_+:\\ \partial(\n_-+\n_+)=\{0,x\}\Delta \partial \n_{0}}} w(\n_-)w_\beta(\n_+) \;  \mathbbm I[\mathcal{E}_L(\calM)] \notag  \\
&\le\sum_{\n:\partial\n=\{0,\mathcal{R}_L(x)\}}w(\n).
\end{align}
Hence, for the nearest-neighbor case, we get
\be
\langle\sigma_0\sigma_x\rangle_{\G,\beta}\ \le \  \langle\sigma_0\sigma_{\mathcal{R}_L(x)}\rangle_{\G,\beta} \, ,
\ee 
which is the original Messager Miracle-Sol\'e inequality (one may pass to the limit to obtain the result on $\bbZ^d$ or $\bbZ\times\bbZ_+^{d-1}$).

In the finite-range case, the above argument has to be modified since a path from $x$ to $0$ may skip over any a priori specified line.  However, no path along edges with $J_{u,v} \neq 0$ may avoid the stack of lines  $\mathcal{P}_{L+j}$ with $j = \{0,\dots, R-1\}$ so that $\mathcal E_{L+j}(\calM)$ must occur for some $0\le j<R$. 
A slight technical problem arises from the fact that $G$ had to be finite in the previous reasoning to ensure that weighted sums over currents were well-defined. Here, one cannot consider a graph $G$ which is simultaneously symmetric under the reflections $\calR_{L+j}$ for $0\le j<R$. Nevertheless, note that the previous reasoning, applied to a fixed $j$, implies that 
$$\langle\sigma_0\sigma_x\rangle_{\G,\beta}{\bf P}_{G,\beta}^{\{0,x\}}[\calE_{L+j}(\calM)]\le \langle\sigma_0\sigma_{\mathcal{R}_{L+j}(x)}\rangle_{\G,\beta}.$$
Taking the limit as $G$ tends to $\bbZ\times\bbZ_+^{d-1}$ or $\bbZ^d$, one deduces the corresponding inequality in infinite-volume.
The above argument, carried for each $j$ separately then leads to
\begin{align*}
\langle\sigma_0\sigma_x\rangle_{\G,\beta}&  \leq    
\langle\sigma_0\sigma_x\rangle_{\G,\beta} \Big(\sum_{j=0}^{R-1} {\bf P}_{G,\beta}^{\{0,x\}}[\calE_{L+j}(\calM)]\Big) \le  \sum_{j=0}^{R-1}  \langle\sigma_0\sigma_{\mathcal{R}_{L+j}(x)}\rangle_{\G,\beta}
\end{align*}
from which \eqref{eq:MMSR} follows. 
\end{proof} 

The above result will be combined with a ``local'' estimate, whose purpose is to bridge the gap between \eqref{eq:MMSR} and \eqref{eq:MMSh}.

\begin{proof}[Proof of Theorem~\ref{thm:mono}]
Starting with Lemma~\ref{lem:extendMMS}, we get
\begin{eqnarray}
\langle\sigma_0\sigma_x\rangle_{G,\beta} &\leq&  \sum_{j=0}^{R-1}\langle\sigma_0\sigma_{\mathcal{R}_{L+j}(x)}\rangle_{G,\beta} \notag \\ 
&\leq& \sum_{j=0}^{R-1}\frac{\langle\sigma_0\sigma_{\mathcal{R}_L(x)}\rangle_{G,\beta}}{\langle\sigma_{\mathcal{R}_{L+j}(x)}\sigma_{\mathcal{R}_L(x)}\rangle_{G,\beta}}\ \le \  C \, \langle\sigma_0\sigma_{\mathcal{R}_L(x)}\rangle_{G,\beta} \, .
\end{eqnarray} 
Here, the second inequality follows from Griffiths' inequality and the last by setting
 $$ C=C(\beta,R,J,d):=   \frac{R}{ \inf\{ \langle\sigma_{0}\sigma_{y}\rangle_{G,\beta} \, : \, \|y\| \leq R \}} ,$$  which is finite due to the assumed connectedness of the graph $G$. 
\end{proof}

\section{Gluing Lemmas}

Below, consider an ordering $\prec$ of the vertices of $G$ which is such that the distance to the origin is non-increasing in the ordering (this will be relevant only for Theorem~\ref{thm:glueCirc}). We order directed paths of $G$ by taking the corresponding lexicographical order, as in Section 2.3 of \cite{DumSidTas16} (notice that here, paths are defined to be edge-avoiding, while \cite{DumSidTas16} uses vertex-avoiding paths, but the argument applies exactly in the same way). 

We will repeatedly use the following combinatorial lemma. \begin{lemma}
\label{lem:combi}
Fix $\beta>0$. There exists a constant $C_0=C_0(J,\beta)>0$ such that the following holds.  Consider two events $E$ and $F$ and a map $\Phi$ from $E$ to $F$. We assume that there exists $s>0$ such that $\Phi(\omega)$ differs from $\omega$ in only $s$ edges, and 
  $|\Phi^{-1}(\omega')|\le 2^s$.
    Then,
  \begin{equation}
     \phi_{\bbZ^2,\beta}[E]\leq C_0^s \phi_{\bbZ^2,\beta}[F].
  \end{equation}
\end{lemma}

\begin{proof}
The lemma was stated and proved in \cite{DumSidTas16} in the context of Bernoulli percolation. 
 The proof uses that the measure has the finite-energy property \cite[(3.4)]{Gri06}, i.e.~the property that conditioned on the other edges, an edge is open or closed with probability bounded away from 0 and 1. Since the random-cluster measure $\phi_{\bbZ^2,\beta}$ also satisfies the finite-energy property, the following lemma can be proved exactly in the same way as \cite[Lemma~7]{DumSidTas16}.
\end{proof}

\begin{figure}[htbp]
  \centering
  \includegraphics[width=1\linewidth]{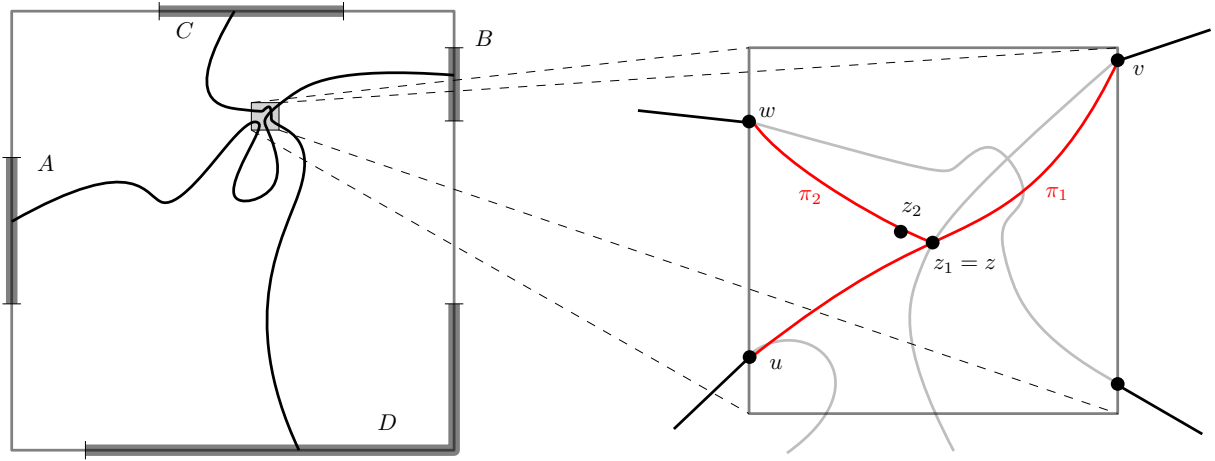}
   \caption{A picture of the configuration near $z$. }
  \label{fig:surgery}
\end{figure}

\begin{theorem}
  \label{thm:gluepaths} Fix $\beta>0$. There exists a constant $c_{\mathrm{glue}}>0$ such that the
  following holds. Let $S=[a,b]\times[c,d]$ with $b-a $ and $d-c$ large enough. Let $A,B,C,D$ be
  four disjoint subsets of $\partial S$   such that 
  any continuous path from $A$ to $B$ in $S$ intersects 
  any continuous path from $C$ to $D$ in $S$. Then,
\begin{equation}
\label{eq:190}
       \phi_{\mathbb Z^2,\beta}[A\lr[S] C] \ge c_{\mathrm{glue}}\, \phi_{\mathbb Z^2,\beta}[A\lr[S]B] \,\phi_{\mathbb Z^2,\beta}[C\lr[S]D].
  \end{equation}
\end{theorem}

\begin{remark}We note that in the simpler case where $G$ is planar, the FKG
  inequality \eqref{eq:FKG} implies that 
\begin{align*} \phi_{\mathbb Z^2,\beta}[A\lr[S] C] &\ge\phi_{\mathbb Z^2,\beta}\left[ \{A\lr[S]B\}\cap\{  C\lr[S] D\}\right]\\
&\ge\phi_{\mathbb Z^2,\beta}\left[ A\lr[S]B\right] \phi_{\mathbb Z^2,\beta}\left[  C\lr[S] D\right], \end{align*} and thus
  \eqref{eq:190} is valid with $c_{\mathrm{glue}}=1$.\end{remark}

For completeness, and because the strategy will be useful in the next theorem, we present an outline of the proof, which can be found in \cite[Section 2.3]{DumSidTas16}. 
\begin{proof} [Proof of Theorem~\ref{thm:gluepaths}] For this proof, we slightly abuse notation and use $\Lambda_{s}(z)$ to denote $\Lambda_{s}(z) \cap S$. We will use $r$ for a large but fixed integer (see below for the constraint that $r$ must satisfy).
Define
\begin{align*}
E&:=\left\{A\lr[S]B\right\}\cap\left\{ C\lr[S]D\right\} \cap \left\{ A\lr[S] C \right\} ^{c},\\
F&:=\left\{ A\lr[S]C \right\}.
\end{align*}
We construct a map
$\Phi : E \rightarrow F,
$
modifying each $\omega$ in at most $\deg(G)|\Lambda_r|$ places such that for any $\omega ^{\prime }$ in the image of $\Phi$, we have $|\Phi^{-1}(\omega')|\le 2^{\deg(G)|\Lambda_r|}$. 

The construction of $\Phi $ is similar to that in \cite{DumSidTas16}, as we now describe. We recommend that the reader take a look at Fig.~\ref{fig:surgery}. Consider $\omega\in E$.   Let $\Gamma(\omega)$ be the minimal (for $\prec$) open path (for the lexicographical ordering) from $A$ to $B$ in $S$.  Consider a vertex $z \in \Gamma(\omega)$  such that $\Lambda_r\left( z\right)$ is connected to $C$ in $S$ by an open path $\pi $, such that $\mathrm{dist}_{\|\cdot\|_\infty}(\pi, \Gamma)\le r$. We define the configuration $\Phi(\omega)$ in three steps:
\begin{enumerate}
\item Define $u$ and $v$
to be respectively the first and last vertices (when going from $A$ to $B$)
of $\Gamma(\omega)$ which are in $\Lambda_{r}\left(
z\right) $. Choose $w$ on the boundary of $
{\Lambda_{r}\left( z\right) }$ in such a way that there exists an open
path $\pi$ (which could be a singleton) from $w$ to $C$. By definition, $w$ is distinct from $u$, $z$
and $v$. Choose two vertices $z_1=z$ and $z_2$ in $\Lambda_r(z)$ that are distinct from $u,v,w$ and such that $z_2$ is maximal among neighbors of $z_1$ (for the ordering $\prec$). 
\item Close all edges of $\omega $ having one endpoint in ${\Lambda_{r}\left( z\right) }$ except the edges of $\Gamma(\omega)$ reaching $u$, $v$ and $ w $.

\item Open a path $\pi_1$ going from $u$ to $v$ in $\Lambda_r(z)$ and passing by $z_1$, and a path $\pi_2$ starting at $z_1$ and going to $w$ by avoiding $\pi_1$ and by going first to $z_2$. Here, one may choose $r$ in such a way that such paths always exist. This is the only constraint on $r$ we will assume, in particular, $r$ will be independent of everything else.
 \end{enumerate}
Set $\Phi(\omega)$ to be the resulting configuration.
By construction, $\Phi(\omega)\in F$ and $\omega$ was modified in only $\deg(G)|\Lambda_r|$ places. 

Now, given
$\omega ^{\prime }$ in the image of $\Phi$, by the same argument as in
\cite{DumSidTas16}, $z$ is the only site in the new minimal path $\Gamma (\omega')$
(from $A$ to $B$) that is connected to $C$ without using any edge in $\Gamma
(\omega')$. The only non-trivial claim of the previous sentence is that $z$ belongs to $\Gamma(\omega')$. One can see that as follows: since $C$ is not connected to  $A$ in $S$, the path $\Gamma(\omega')$ agrees
with $\Gamma(\omega)$ up to $u$. Then, because $z_2$ is maximal among neighbors of $z_1$, $\Gamma(\omega')$ cannot go through $z_2$ and must go to $v$, and then
$\Gamma(\omega')$ must agree with $\Gamma(\omega)$ from $v$ to the end. Therefore, $\omega
^{\prime }$ must be equal to $\Phi(\omega)$ for some $z$ that can be uniquely
determined by $\omega^{\prime }$. Thus, a configuration in $\Phi^{-1}(\omega')$ must agree with $\omega'$ except possibly at  the edges intersecting  $\Lambda_r( z)$, where $z$ is determined from $\omega'$. 

This shows that
$\Phi $ satisfies the conditions of Lemma~\ref{lem:combi}, with
$s=\mathrm{deg}(G)\vert\Lambda_{r}\vert $. The claim follows trivially by setting $c_{\rm glue}:=C_0^{\mathrm{deg}(G)|\Lambda_r|}$.
\end{proof}

\begin{remark}
The constant $r$ must be chosen large enough that the rewiring works. The same feature will arise in the next proof.
\end{remark}

\begin{remark}\label{rmk:h}
The above strategy  clearly extends to the quantities 
$B_n$ and $F_n(k)$ defined in Sections~\ref{sec:P2} and \ref{sec:P3},  
and leads to a proof of the bound
\begin{equation}
\label{eq:190a}
       \phi_{\mathbb Z^2,\beta}[B_n] \ge c_{\rm glue}\,\phi_{\bbZ^2,\beta}[F_n(k)]\phi_{\bbZ^2,\beta}[F_n(k)'],
       \end{equation}
       where $F_n(k)'$ is the symmetric of $F_n(k)$ with respect to the $x$-axis. Indeed, simply consider the minimal path from $\{-n\}\times\bbR_-$ to $\{n\}\times\bbR_-$ passing above $(k,0)$, and glue it to a path going from $\{-n\}\times\bbR_+$ to $\{n\}\times\bbR_+$ and passing below $(k,0)$. The reconstruction works as above.
\end{remark}

The previous claim could have been directly extracted from \cite{DumSidTas16} and there was no need to include the proof here. The reason why we did so is that we will need a slightly more tedious ``gluing lemma'' enabling us to create circuits, and that the proof of the previous statement is somehow a baby example of what we do for the next one.

Let $C_n$ be the event that there exists a circuit  surrounding $0$ in a connected component that is connected to $\partial\Lambda_n$. 

\begin{theorem}
\label{thm:glueCirc}
  There exists a constant $c_{\rm glue}>0$ such that for every $n\ge1$, we have
  \begin{equation*}
    \phi_{\mathbb Z^2,\beta}[C_n]\ge c_{\rm glue}\cdot\phi_{\mathbb Z^2,\beta}[\Lambda_n \text{ is crossed horizontally}]^{20}. 
  \end{equation*}
\end{theorem}

\medskip

\begin{remark}The theorem above is straightforward when the graph $G$ is planar, as seen in the proof of Lemma~\ref{lem:P1} (we simply combined four crossings of $\Lambda_n$ together to create a circuit).
When the graph $G$ is not planar, we wish to apply the same strategy and use a gluing procedure similar to the proof of Theorem \ref{thm:gluepaths} in order to create a circuit out of four well-chosen square crossings. This approach is more complicated than it may look in first thought, due to two difficulties:
\begin{description}
\item[First difficulty] It is not  true that a left-right continuous path either goes above or below the origin: it may have a more complicated topology and wind non trivially (in a topological sense) around the origin (see Fig.~\ref{fig:A1}, \ref{fig:A2} and \ref{fig:A3}). In the proof, we will define a formal notion of winding around 0 for a horizontal crossing of $\Lambda$ and divide into cases depending on the value of a typical winding.
\item[Second difficulty]Another difficulty arises when gluing paths in order to create a circuit. Due to the fact that the topology of the ``glued'' connected component changes when a circuit is created, the reconstruction step consisting in identifying the vertex $z$ is much harder than in Theorem~\ref{thm:gluepaths}. This difficulty explains why the proof below is more technical than the proof of Theorem~\ref{thm:gluepaths}: we will not be able to glue directly horizontal crossings together, and will use an intermediate path $\Gamma$ that crosses the annulus $\Lambda_n\setminus\Lambda_m$ from inside to outside (where $\Lambda_m$ is the smallest box centered at $0$ connected to $\partial \Lambda_n$) to be able to localize the places where we modified the configuration.\end{description}
\end{remark}
\begin{figure}[htbp]
  \centering
  \begin{minipage}[t]{.31\linewidth} \centering
    \includegraphics[width=.7\textwidth]{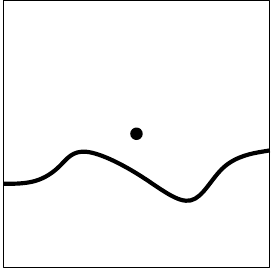}
    \caption{A left-right path with winding 0 passing below the origin.}\label{fig:A1}
  \end{minipage}
\hfill
\begin{minipage}[t]{.31\linewidth} \centering
    \includegraphics[width=.7\textwidth]{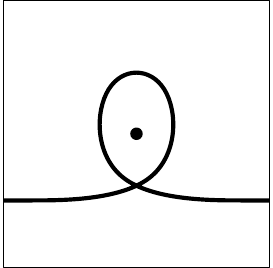}
    \caption{A left-right path with winding 1 going counterclockwise around 0.}\label{fig:A2}
  \end{minipage}
\hfill
\begin{minipage}[t]{.31\linewidth} \centering
    \includegraphics[width=.7\textwidth]{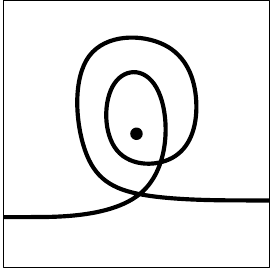}
    \caption{A left-right path with winding 2.}\label{fig:A3}
  \end{minipage}
\end{figure}

\begin{proof}[Proof of Theorem~\ref{thm:glueCirc}]

Set $\Lambda_n^\bullet:= \Lambda_n\setminus\{0\}$ for the punctured box. 

We consider the event $C_n$ that there exists an open circuit in $\Lambda_n^\bullet$ that surrounds $\{0\}$ and is connected to  the external boundary of $\Lambda_n^\bullet$. We also use the denomination $\Lambda_n^\bullet$ is horizontally crossed to denote the event that $\Lambda_n$ is horizontally crossed by a path which does not hit $0$. We wish to prove  that there exists a constant $c>0$ such that    
  \begin{equation}
\label{eq:28}
    \phi_{\mathbb Z^2,\beta}[C_n] \ge c\cdot  \phi_{\mathbb Z^2,\beta}[\Lambda_n^\bullet\text{ is horizontally crossed}]^{20}.  
  \end{equation}
This immediately concludes the proof since the probability of $\Lambda_n$ being crossed by a path staying in $\Lambda_n^\bullet$ can easily be compared to the probability of $\Lambda_n$ being crossed using the finite-energy \cite[(3.4)]{Gri06} (we leave it as an easy exercise to the reader).

We now define the winding of a left-to-right continuous path\footnote{Formally, we should restrict our attention to a left-to-right continuous path that intersects itself finitely many times. One can easily alter the definition of the continuous path associated with a path $\pi$ of vertices in such a way that this property is trivially satisfied (simply modify the path from $v_i$ to $v_{i+1}$ locally).} $\gamma$ in $\Lambda_n^\bullet$. We say that $\gamma$ has \emph{winding $0$} if it is homotopic among left-to-right continuous paths in $\Lambda_n^\bullet$ to a simple left-to-right continuous path included in $\mathbb R^2\times\{0\}$ (a simple path is a path with no double point). For $k\ge1$, we say that $\gamma$ has \emph{winding $k$} if it is homotopic to a left-right path in $\mathbb R^2\times\{0\}$ with $k$ double points and it is not homotopic to a left-to-right path in $\mathbb R^2\times\{0\}$ with $k-1$ double points (see Fig. \ref{fig:A1}, \ref{fig:A2} and \ref{fig:A3} for paths with different windings).

 Let $E_0$ (resp.\@ $E_1$, $E_{\ge 2}$) be the event that there exists a horizontal crossing of $\Lambda_n^\bullet$ with winding $0$ (resp.\@ $1$, at least $2$). The union bound ensures that at least one of the three events $E_0$, $E_1$, $E_{\ge 2}$ occurs with probability larger than $1/3$ of the probability that $\Lambda_n^\bullet$ is crossed horizontally.
Therefore,  \eqref{eq:28} follows from the following three inequalities:
\begin{align}
  \label{eq:34} &\phi_{\mathbb Z^2,\beta}[C_n]\ge c\, \phi_{\mathbb Z^2,\beta}[E_{\ge2}],\\
  \label{eq:35} &\phi_{\mathbb Z^2,\beta}[C_n]\ge c\, \phi_{\mathbb Z^2,\beta}[E_1]^4,\\
  \label{eq:44} &\phi_{\mathbb Z^2,\beta}[E_1]\ge c\, \phi_{\mathbb Z^2,\beta}[E_0]^5,
\end{align}
where $c=c(J,\beta)>0$ denotes a positive constant that is  independent of $n$. 
     
The main difficulty is the proof of \eqref{eq:34}. The inequality \eqref{eq:35} can be proved similarly. Also, \eqref{eq:44} is simply a successive application (five times) of a ``gluing lemma'' very similar to Theorem~\ref{thm:gluepaths}. Since this case is simpler than the other ones, and not conceptually harder than the proof of Theorem~\ref{thm:gluepaths}, we omit the proof here and focus on the proofs of \eqref{eq:34} and \eqref{eq:35}.

\paragraph{Proof of \eqref{eq:34}.}  We will use the combinatorial statement of Lemma~\ref{lem:combi}. Once again, fix $r$ large enough (but not depending on $n$).
Define a mapping  $\Phi:E_{\ge2}\setminus C_n\to C_n$ as follows.

Consider $\omega\in E_{\ge2}\setminus C_n$. Let $x_0$ be the minimum point (for the ordering $\prec$) connected to $\partial\Lambda_n$. Notice that the ordering $\prec$ is such that $x_0$ is chosen among the closest points to $0$ connected (in $\omega$) to $\partial\Lambda_n$. Let $\Gamma(\omega)$ be the minimal path in $\omega$ from $x_0$ to $\partial\Lambda_n$. To define  $\Phi(\omega)$, we distinguish between two cases:

\begin{figure}[htbp]
  \centering
  \includegraphics[width=1.00\linewidth]{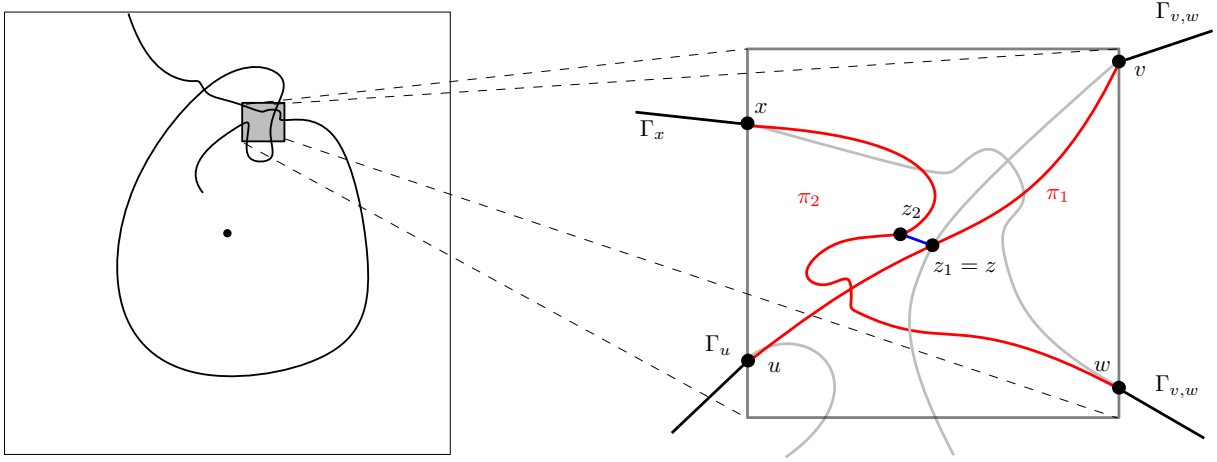}
   \caption{A picture of the topological configuration near the point $z$ in Case 1. }
  \label{fig:A10}
\end{figure}

\paragraph{Case 1:} The path $\Gamma(\omega)$ contains a circuit surrounding $0$. Thus, we can pick $z\in\Gamma(\omega)$ and four distinct vertices $u$, $v$, $w$ and $x$ in 
$\partial\Lambda_r(z)$ such that
  \begin{itemize}
  \item $u$ and $x$ are respectively the first and last vertices (when going from $x_0$ to the boundary of $\Lambda_n^\bullet$) of $\Gamma(\omega)$ which are in $\Lambda_r(z)$. We define $\Gamma_u$ and $\Gamma_x$ to be respectively the parts of $\Gamma(\omega)$ that connect $x_0$ to $u$, and $x$ to $\partial\Lambda_n$.
  \item there exists a subpath $\Gamma_{v,w}$ of $\Gamma(\omega)\setminus\Lambda_r(z)$ from $v$ to $w$ such that $\Gamma_{v,w}\cup[v,w]$ is a circuit of $\Lambda_n^\bullet$ surrounding $0$. Here, we do not require the segment $[v,w]$ to be an edge of $G$, therefore the circuit $\Gamma_{v,w}\cup[v,w]$ does not a priori correspond to a circuit of $G$.
  \end{itemize}
   These vertices are illustrated on Fig.~\ref{fig:A10}.
   We now construct a configuration by performing a local modification in $\Lambda_r(z)$. Choose $z_1=z$ and $z_2$ in $\Lambda_r(z)$ in such a way that 
   \begin{itemize}
   \item $z_2$ is the maximal neighbour of $z_1$ for $\prec$,
   \item there exist two disjoint paths $\pi_1$, $\pi_2$ of $G$ connecting respectively $u$ to $v$ and $w$ to $x$ inside $\Lambda_r(z)$, and such that $z_1$ belongs to $\pi_1$ and $z_2$ belongs to $\pi_2$. We further assume that $\pi_1$ and $\pi_2$ do not contain vertices smaller than $x_0$ for the ordering $\prec$.
   \end{itemize}
   (The existence of $z_1$, $z_2$ and $\pi_1,\pi_2$ is guaranteed by choosing $r$ large enough.) Then, close all the edges intersecting $\Lambda_r(z)$ and, finally, open the edges in $\pi_1$, $\pi_2$ as well as the edge from $z_1$ to $z_2$. We set $\Phi(\omega)$ to be the resulting configuration. One can check that $\omega'\in \Phi(\omega)$ contains a circuit.

\begin{figure}[htbp]
  \centering
  \includegraphics[width=0.40\linewidth]{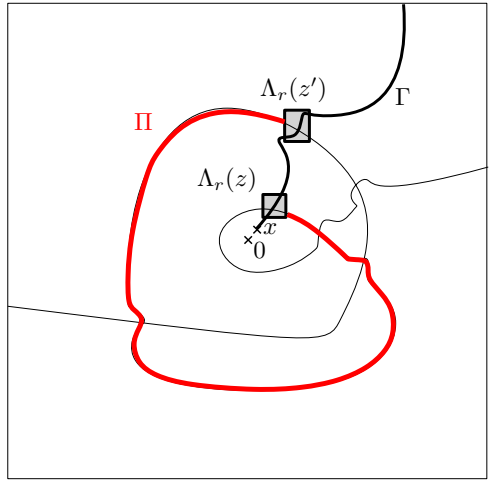}
   \caption{A picture of the topological configuration and the points $z$ and $z'$ in Case 2. }
  \label{fig:A5}
\end{figure}

\paragraph{Case 2:}
The path $\Gamma(\omega)$ does not contain a circuit surrounding $0$. Since $\mathcal E_{\ge 2}$ occurs, there exists an open path $\pi$ in $\bbR^2\setminus \Gamma(\omega)$ such that $(\Gamma(\omega)+\Lambda_r)\cup\pi$ contains a circuit around 0, and the vertices in $\pi$ are connected to $\partial\Lambda_n$; see Fig.~\ref{fig:A5} \footnote{The existence of such a path can be obtained by observing that this property is preserved by continuous deformation of paths, and by drawing a generic example.}. We define $\Pi$ to be the minimal such path for the lexicographical ordering. Now, let $z$ and $z'$ be in $\Gamma$ so that $\Pi$ starts on $\Lambda_r(z)$ and ends on $\Lambda_r(z')$. Choose $u$ and $v$ (resp.~$u'$ and $v'$) to be the first and last visits of $\Gamma(\omega)$ to $\Lambda_r(z)$  and $\Lambda_r(z')$ respectively. We also set $w$ and $w'$ to be the beginning and end of $\Pi$. 

  Construct a configuration by performing a local modification in $\Lambda_{r}(z)$ and $\Lambda_{r}(z'$) very similar to the one made in Theorem~\ref{thm:gluepaths}. Choose $z_1$ and $z_2$ in $\Lambda_r(z)$ in such a way that 
   \begin{itemize}
   \item $z_2$ is the maximal neighbour of $z_1$,
   \item there exist two disjoint paths $\pi_1$, $\pi_2$ of $G$ connecting respectively $u$ to $v$ and $z_1$ to $w$ inside $\Lambda_r(z)$, and such that $z_1$ belongs to $\pi_1$ and $z_2$ is the first step of $\pi_2$. We further assume that $\pi_1$ and $\pi_2$ do not contain vertices smaller than $x_0$ for the ordering $\prec$.
   \end{itemize}
   Now, close all the edges with one endpoint in $\Lambda_r(z)$ which are not the edges of $\Gamma(\omega)$ and $\Pi$ respectively arriving at $u$, $v$ and $w$. Then, open the paths $\pi_1$ and $\pi_2$.
   
   One can perform a similar modification in $\Lambda_r(z')$. We set $\Phi(\omega)$ for the resulting configuration. Note that $\omega'=\Phi(\omega)$ contains a circuit.

\paragraph{Reconstruction of the pre-images.}
Let $\omega'$ be in the image of $\Phi$. We wish to prove that 
\begin{equation}\label{eq:45}
|\Phi^{-1}(\omega')|\le 2^{(\mathrm{deg}(G)|\Lambda_r|)^2}.
\end{equation}
First, $x_0$ can be reconstructed as the minimal (for $\prec$) vertex of $\omega'$ connected to $\partial\Lambda_n$, since the modifications were made in such a way that no vertex smaller than $x_0$ was connected to the boundary by the procedure (notice that the fact that $\Pi$ was made of vertices connected to $\partial\Lambda_n$ guarantees that none of these vertices is smaller than $x_0$). Also, the local modifications were made in such a way that $\Gamma(\omega')$ must coincide with $\Gamma(\omega)$ except in possibly one box of size $r$ (if we are in Case 1) or two boxes of size $2r$ (if we are in Case 2). Hence, the modifications are such that whether $\Gamma(\omega')$ contains a circuit surrounding $0$ enables us to decide whether $\omega'$ was constructed from $\omega$ from Case 1 or Case 2. 
Now,
\begin{itemize}
\item In Case 1, the edge from $z_1$ to $z_2$ can be reconstructed since it is the only open edge that does not belong to $\Gamma(\omega')$ and that is pivotal for the existence of a circuit around $0$ (closing this edge destroys the circuit).

\item In Case 2, the path $\Pi$ can be reconstructed as the minimal path in the complement of $\Gamma(\omega')$ such that $(\Gamma(\omega')+\Lambda_r) \cup \Pi$ contains a circuit around 0, since nothing was modified outside $\Gamma$ except possibly near $z$ and $z'$, but that the path $\Pi$ was the only path not destroyed by the procedure of closing edges having one endpoint in the boxes around $z$ and $z'$.
\end{itemize}
This finishes the proof of \eqref{eq:34}. 

\paragraph{Proof of \eqref{eq:35}.} A horizontal crossing in $\Lambda_n^\bullet$ with winding 1 can either go clockwise or counterclockwise around $0$ (see Fig.~\ref{fig:A2}).  By symmetry, there exists a horizontal crossing with winding 1 that goes counterclockwise around 0 with probability larger than $\phi_{\mathbb Z^2,\beta}[E_1]/2$. Let $F$ be the event that there exist in $\Lambda_n^\bullet$, 
\begin{itemize}[noitemsep,nolistsep]
\item a horizontal crossing with winding 1 that goes clockwise around $0$,
\item a horizontal crossing with winding 1 that goes counterclockwise around $0$,
\item a vertical crossing with winding 1 that goes clockwise around $0$, and
\item a vertical crossing with winding 1 that goes counterclockwise around $0$.
\end{itemize}
By the FKG inequality \eqref{eq:FKG}, we have
$$\phi_{\bbZ^2,\beta}[F]\ge(\phi_{\mathbb Z^2,\beta}[E_1]/2)^4.$$ Now, the event $F$ satisfies a property similar to the event $E_{\ge 2}$: for any deterministic path $\gamma$ from 0 to the boundary of $\Lambda_n^\bullet$, there exists an open path $\pi$ in $\gamma$ such that $\pi\cup \gamma$ contains a circuit.
Based on this property, we can prove 
\begin{equation}
  \phi_{\mathbb Z^2,\beta}[C_n]\ge c\,\phi_{\mathbb Z^2,\beta}[F]
\end{equation}
exactly as we proved \eqref{eq:34} (simply replace any instance of $E_{\ge2}$ by $F$ in the proof above). This concludes the proof of \eqref{eq:35}.
\end{proof}

\begin{remark}
  The proof of \eqref{eq:28} can be adapted in a straightforward way to show that for every fixed $k$ and every $n$ large enough,
\begin{equation}
\label{eq:3}
    \phi_{\mathbb Z^2,\beta}[C_{k,n}] \ge c_{\rm glue}\,\phi_{\mathbb Z^2,\beta}[\exists \textrm{ continuous path from left to right in $\Lambda_n\setminus\Lambda_k$}]^{20},
  \end{equation}    
  where $C_{k,n}$ denotes the event that there exists an open circuit in $\Lambda_n\setminus\Lambda_k$ that surrounds 0 and is connected to $\partial\Lambda_n$.\end{remark}

\paragraph{Acknowledgements} This work was supported in parts by the NSF grants  PHY-1305472, DMS-1613296, the NCCR SwissMAP, the IDEX grant from Paris-Saclay, and a Princeton University/University of Geneva partnership fund.

\end{document}